\documentclass[11pt,reqno]{amsart}
\usepackage[foot]{amsaddr}
\usepackage{nicefrac,bm,upgreek,mathtools,verbatim}
\usepackage[final]{hyperref}
\usepackage[mathscr]{eucal}
\usepackage{dsfont}
\usepackage[normalem]{ulem}
\usepackage{amsopn}

\usepackage[margin=1in]{geometry}

\usepackage{latexsym}
\usepackage{epsfig}
\usepackage{caption}
 
\usepackage{amsmath}
\usepackage{amsfonts}
\usepackage{array}
\usepackage{dsfont}
\usepackage{amssymb}
\usepackage{comment}

\usepackage{graphicx}
\usepackage{multicol}
\usepackage{tikz}
\usepackage{scrextend}
\usepackage{enumitem}
\usepackage{mathtools}
\usepackage{amsthm}
\usepackage{breqn}
\usepackage{xcolor}

\usepackage{tikz}
\usepackage{bbm}
\usepackage{algpseudocode}
\usepackage{algorithm}

\newcommand{\mypara}[1]{\vspace{1.2ex}\noindent\textbf{\textit{#1}}\quad}
\newcommand{\mysubpara}[1]{\vspace{0.8ex}\noindent\textit{#1}\quad}

\newcommand{\cF}{\mathcal{F}}\newcommand{\bR}{\mathbb{R}}

\newcommand{\altc}{\hat{c}}
\newcommand{\bx}{\mathbb{X}}
\newcommand{\bu}{\mathbb{U}}
\newcommand{\bW}{\mathbb{W}}
\newcommand{\Exp}{\mathbb{E}}
\newcommand{\Pro}{\mathbb{P}}

\newcommand{\ct}{\mathcal{T}}

\newcommand{\ctp}{\mathcal{S}}

\newcommand{\gsp}{\gamma^{*}_{\altc,\cs}}
\newcommand{\jb}{J_{\beta}}
\newcommand{\wc}{\mathcal{W}}
\newcommand{\clip}{\Vert c\Vert_{\mathrm{Lip}(\bx)}}
\newcommand{\flip}{\Vert f\Vert_{\mathrm{Lip}}}

\newcommand{\gglip}{\Vert g\Vert_{\mathrm{Lip}}}

\newcommand{\cinf}{\Vert c\Vert_{\infty}}

\newcommand{\ctt}{\mathcal{T}}
\newcommand{\cs}{\mathcal{S}}
\newcommand{\tlip}{\Vert \ct\Vert_{\mathrm{Lip}}}

\newcommand{\jsc}{J_{\beta}^{*}(c,\ct)}
\newcommand{\jscp}{J_{\beta}^{*}(\altc,\ctp)}

\newcommand{\vcn}{v_{\altc,\cs}}
\newcommand{\hlip}{\left\Vert h_{c,\ct}^{*}\right\Vert_{\mathrm{Lip}}}

\newcommand{\flipx}{\Vert f\Vert_{\mathrm{Lip}(\bx)}}
\newcommand{\flipw}{\Vert f\Vert_{\mathrm{Lip}(\mathbb{W})}}
\newcommand{\flipxu}{\Vert f\Vert_{\mathrm{Lip}(\bx \times \bu)}}
\newcommand{\rlipx}{\Vert r\Vert_{\mathrm{Lip}(\bx)}}
\newcommand{\rlipxu}{\Vert r\Vert_{\mathrm{Lip}(\bx \times \bu)}}
\newcommand{\revise}[1]{#1}

\DeclareMathOperator*{\argmin}{argmin}

\usepackage{color}
\definecolor{dmagenta}{rgb}{.4,.1,.5}
\definecolor{dblue}{rgb}{.0,.0,.5}
\definecolor{mblue}{rgb}{.0,.0,.7}
\definecolor{ddblue}{rgb}{.0,.0,.4}
\definecolor{dred}{rgb}{.7,.0,.0}
\definecolor{dgreen}{rgb}{.0,.5,.0}
\definecolor{Eeom}{rgb}{.0,.0,.5}

\allowdisplaybreaks
\newtheorem{lemma}{Lemma}[section]
\newtheorem{theorem}{Theorem}[section]

\newtheorem{corollary}{Corollary}[section]

\theoremstyle{definition}
\newtheorem{definition}{Definition}[section]
\newtheorem{assumption}{Assumption}[section]

\newtheorem{example}{Example}[section]

\theoremstyle{remark}
\newtheorem{remark}{Remark}[section]

\numberwithin{theorem}{section}
\numberwithin{equation}{section}

\hypersetup{
  colorlinks=true,
  citecolor=mblue,
  linkcolor=mblue,
  urlcolor = blue,
  anchorcolor = blue,
  frenchlinks=false,
  pdfborder={0 0 0},
  naturalnames=false,
  hypertexnames=false,
  breaklinks}

\usepackage[capitalize,nameinlink]{cleveref}

\usepackage[abbrev,msc-links,nobysame,citation-order]{amsrefs}

\crefname{section}{Section}{Sections}
\crefname{subsection}{Section}{Sections}
\crefname{condition}{Condition}{Conditions}
\crefname{hypothesis}{Hypothesis}{Conditions}
\crefname{assumption}{Assumption}{Assumptions}
\crefname{lemma}{Lemma}{Lemmas}
\crefname{fact}{Fact}{Facts}

\Crefname{figure}{Figure}{Figures}

\crefformat{equation}{\textup{#2(#1)#3}}
\crefrangeformat{equation}{\textup{#3(#1)#4--#5(#2)#6}}
\crefmultiformat{equation}{\textup{#2(#1)#3}}{ and \textup{#2(#1)#3}}
{, \textup{#2(#1)#3}}{, and \textup{#2(#1)#3}}
\crefrangemultiformat{equation}{\textup{#3(#1)#4--#5(#2)#6}}
{ and \textup{#3(#1)#4--#5(#2)#6}}{, \textup{#3(#1)#4--#5(#2)#6}}
{, and \textup{#3(#1)#4--#5(#2)#6}}

\Crefformat{equation}{#2Equation~\textup{(#1)}#3}
\Crefrangeformat{equation}{Equations~\textup{#3(#1)#4--#5(#2)#6}}
\Crefmultiformat{equation}{Equations~\textup{#2(#1)#3}}{ and \textup{#2(#1)#3}}
{, \textup{#2(#1)#3}}{, and \textup{#2(#1)#3}}
\Crefrangemultiformat{equation}{Equations~\textup{#3(#1)#4--#5(#2)#6}}
{ and \textup{#3(#1)#4--#5(#2)#6}}{, \textup{#3(#1)#4--#5(#2)#6}}
{, and \textup{#3(#1)#4--#5(#2)#6}}

\crefdefaultlabelformat{#2\textup{#1}#3}

\begin{document}
\title[Robustness to Model Approximation, Learning, and Sample Complexity]
{Robustness to Model Approximation, Model Learning From Data, and Sample Complexity in Wasserstein Regular MDPs}

\author[Yichen Zhou]{Yichen Zhou$^\dag$}
\address{$^\dag$Department of Mathematics and Statistics,
Queen's University, Kingston, ON, Canada}
\email{yichen.zhou@queensu.ca}

\author[Yanglei Song]{Yanglei Song$^{\dag}$}
\email{yanglei.song@queensu.ca}

\author[Serdar Y\"{u}ksel]{Serdar Y\"{u}ksel$^{\dag}$}
\email{yuksel@queensu.ca}

\begin{abstract}
The paper studies the robustness properties of discrete-time stochastic optimal control under Wasserstein model approximation for both discounted-cost and average-cost criteria. Specifically, we study the performance loss when applying an optimal policy designed for an approximate model to the true dynamics compared with the optimal cost for the true model under the sup-norm-induced metric, and relate it to the Wasserstein-1 distance between the approximate and true transition kernels. A primary motivation of this analysis is empirical model learning, as well as empirical noise distribution learning, where Wasserstein convergence holds under mild conditions but stronger convergence criteria, such as total variation, may not. We discuss applications of the results to the disturbance estimation problem, where sample complexity bounds are given, and also to a general empirical model learning approach, obtained under either Markov or i.i.d.~learning settings.
\end{abstract}


\maketitle

\section{Introduction}
In this paper, we study the continuity and robustness properties of discrete-time stochastic optimal control under model approximation across various performance criteria. Specifically, we characterize how the accuracy of a model approximation affects the performance loss incurred when applying an optimal policy designed for the approximate model to the true system dynamics. This notion of robustness is of significant practical importance, as learning-based algorithms are seldom implemented with exact model knowledge. A particularly important context is in empirical model learning, as well as empirical noise distribution learning. We discuss such applications in detail and also arrive at sample complexity results.

As will be discussed in the literature review, the term \textit{robustness} has been interpreted in various ways across diverse contexts and methodological frameworks. In this paper, we define robustness in terms of performance degradation---specifically, the loss incurred when a control policy designed for an approximate or incorrect model is applied to the true system, measured relative to the optimal cost achievable with full knowledge of the true model.

\subsection{Preliminaries on Markov Decision Processes}\label{prelim}
Consider a discrete-time controlled Markov process $\{(X_t, U_t): t \geq 0\}$ defined on a probability space $(\Omega, \mathcal{F}, \mathbb{P})$, with Polish state and action spaces $(\bx, d_{\bx})$ and $(\bu, d_{\bu})$, respectively. These spaces are endowed with their corresponding Borel $\sigma$-algebras, denoted by $\mathcal{F}_{\bx}$ and $\mathcal{F}_{\bu}$. Let $(\bx \times \bu, \sigma(\mathcal{F}_{\bx} \times \mathcal{F}_{\bu}))$ be the product measurable space.

\begin{definition}
A map $\mathcal{T}:\mathcal{F}_{\bx}\times\bx\times\bu\to[0,1]$ is a \textit{controlled transition kernel} if a) for every  $A\in\mathcal{F}_{\bx}$, the map $(x,u)\mapsto \mathcal{T}(A|x,u)$ is $\sigma(\mathcal{F}_{\bx}\times\mathcal{F}_{\bu})$-measurable, and b) for every  $(x,u)\in \bx\times \bu$, the map $A\mapsto\mathcal{T}(A|x,u)$ is a probability measure on $(\bx,\mathcal{F}_{\bx})$.
\end{definition}

We assume that the true dynamics of $\{X_t,U_t: t \geq 0\}$ is given by a controlled transition kernel $\ct$, that is, for $(x,u) \in \bx \times \bu$ and $A \in \cF_{\bx}$, $\Pro\left(X_{t+1} \in A \vert X_{t} = x, U_{t} = u\right) = \ct(A|x,u)$ for $t \geq 0$. Further, we define the history (or path) space at time $t$ as $(\mathbb{H}_{t},\mathcal{F}_{\mathbb{H}_{t}}):=(\bx^{t+1}\times \bu^{t},\sigma\left(\mathcal{F}_{\bx}^{t+1}\times\mathcal{F}_{\bu}^{t}\right))$. An admissible policy at time $t$ is a measurable function $\gamma_t: \mathbb{H}_t \to \mathcal{P}(\bu)$, where $\mathcal{P}(\bu)$ denotes the space of probability measures on $(\bu, \mathcal{F}_{\bu})$, endowed with the weak topology. At each time $t \geq 0$, the decision-maker observes the realized history $h_t := \{X_{[0,t]}, U_{[0,t-1]}\} \in \mathbb{H}_t$ and selects an action $U_t \in \bu$ according to the distribution $\gamma_t(h_t)$. The system then incurs a cost $c(X_t, U_t)$ and transitions to a new state $X_{t+1}$ according to $\ct(\cdot \mid X_t, U_t)$. This procedure is repeated over time.

We refer to the quadruple $(\bx,\bu,\ct,c)$ as a discrete-time Markov decision process, abbreviated as MDP, and denote the set of admissible policies by
\begin{equation*}
\Gamma_{A}:=\{\{\gamma_{t}\}_{t=0}^{\infty}:\;\forall\;t\in\mathbb{N},\;\gamma_{t}:\;\mathbb{H}_{t}\to\mathcal{P}(\bu)\;\;\; \text{measurable}\}.
\end{equation*}
Our objective is to minimize the accumulated cost over time by selecting a control policy $\{\gamma_t : t \in \mathbb{N}\}$ from $\Gamma_A$, according to one of the performance criteria specified below.

\begin{itemize}
    \item[(a)] \textbf{Discounted-cost Criterion.}  Given a discount factor $\beta \in (0,1)$, for $x \in \bx$, define
\begin{align*}
&\text{the value function under policy $\gamma$:} &&J_{\beta}(c,\ctt, \gamma)(x):=\mathbb{E}^{\ctt,\gamma}\left[\left.\sum_{t=0}^{\infty}\beta^{t}c(X_{t}, U_{t})\right|X_{0}=x\right],\\
&\text{and the optimal value function}: &&J_{\beta}^{*}(c,\ctt)(x):=\inf_{\gamma\in\Gamma_{A}}J_{\beta}(c,\ctt, \gamma)(x).
\end{align*}

\item[(b)] \textbf{Average-cost criterion.}  For $x \in \bx$, define
    \begin{align*}
    &\text{the value function under policy $\gamma$} && J_{\infty}(c, \ct, \gamma)(x) := \limsup_{T \to \infty} \frac{1}{T} \mathbb{E}^{\ct, \gamma} \left[ \left. \sum_{t=0}^{T-1} c(X_t, U_t) \,\right| X_0 = x \right], \\
    &\text{and the optimal value function:} && J_{\infty}^*(c, \ct)(x) := \inf_{\gamma \in \Gamma_A} J_{\infty}(c, \ct, \gamma)(x).
    \end{align*}
\end{itemize}

Here, $\mathbb{E}^{\ct,\gamma}$ denotes the expectation under the controlled transition kernel $\ct$ and the policy $\gamma$. 

\subsection{Problem Statement, Organization, and Contributions}
Consider a reference MDP $(\bx, \bu, \ct, c)$, where a decision-maker seeks to perform optimal control. In practice, it is rarely the case that the decision-maker has complete knowledge of the MDP---particularly the cost function $c$ and the transition kernel $\ct$. It is therefore natural to consider a learning-based scheme in which the decision-maker first estimates an approximate MDP $(\bx, \bu, \ctp, \altc)$ using samples generated from the reference MDP, derives an optimal policy $\gamma_{\altc,\ctp}^*$ for the approximate model, and then applies this policy to the true MDP.

This gives rise to a fundamental question: how much performance is lost by applying an optimal policy derived from an approximate model to the true system? We refer to this performance degradation due to model mismatch as the \textit{robustness error}. Specifically, under the discounted and average-cost criteria, respectively, the robustness error due to model mismatch is given by
\begin{equation}\label{eq:approximationproblem}
\left\Vert J_{\beta}(c,\ct,\gamma_{\altc,\ctp}^{*})-J_{\beta}^{*}(c,\ct)\right\Vert_{\infty} \quad \text{ and } \quad
\left\Vert J_{\infty}(c,\ct,\gamma_{\altc,\ctp}^{*})-J^{*}_{\infty}(c,\ct)\right\Vert_{\infty},
\end{equation}
which we want to upper bound with a function of some discrepancy metric between $(c,\ct)$ and $(\altc,\ctp)$.

To study the robustness problem in \eqref{eq:approximationproblem}, as in \cite{kara2020robustness}, we first analyze the continuity of optimal value functions with respect to the model. Specifically, since $\gamma_{\altc,\ctp}^{*}$ is optimal for the MDP $(\ctp,\altc)$, we have the decomposition
\begin{align*}
    \left\Vert J_{\beta}(c,\ct,\gamma_{\altc,\ctp}^{*})-J_{\beta}^{*}(c,\ct)\right\Vert_{\infty}
    \leq \left\Vert J_{\beta}(c,\ct,\gamma_{\altc,\ctp}^{*})-J_{\beta}(\hat{c},\ctp,\gamma_{\altc,\ctp}^{*})\right\Vert_{\infty} 
     + \left\Vert J_{\beta}^{*}(\hat{c},\ctp)-J_{\beta}^{*}(c,\ct)\right\Vert_{\infty}.
\end{align*}
The first term measures the difference between the value functions when the same policy $\gamma_{\altc,\ctp}^{*}$ is applied to two MDPs. The second term,
$\left\Vert J_{\beta}^{*}(\hat{c},\ctp)-J_{\beta}^{*}(c,\ct)\right\Vert_{\infty}$, measures the difference between the \emph{optimal} value functions under the two MDPs. Bounding this quantity is therefore known as studying the continuity of the \emph{optimal} value function. It is important in its own right and also plays a key role in bounding \eqref{eq:approximationproblem}.

The above results are critical to related problems in machine learning and control, particularly model-based offline reinforcement learning and the certainty equivalence principle, where the goal is to determine the sample size needed to learn a sufficiently accurate model for optimal control under a performance criterion. Specifically, suppose that we are given transition data $\{(X_t,U_t,X_{t+1})\}_{t=1}^{N}$ of the true MDP
$(\bx,\bu,\ct,c)$, either generated independently via a transition kernel simulator or collected along a sample path.
Using these data, we construct an estimated model
$(\bx,\bu,\ct_N,c_N)$.
The expected robustness error under the statistical settings is defined as
\begin{equation}\label{eq:learningproblem}
\mathbb{E}\!\left[\left\Vert
J_{\beta}(c,\ct,\gamma^{*}_{c_N,\ct_N})
-
J_{\beta}^{*}(c,\ct)
\right\Vert_{\infty}\right]
\quad\text{and}\quad
\mathbb{E}\!\left[\left\Vert
J_{\infty}(c,\ct,\gamma^{*}_{c_N,\ct_N})
-
J_{\infty}^{*}(c,\ct)
\right\Vert_{\infty}\right],
\end{equation}   
where the expectations are taken with respect to the law of the data $\{(X_t,U_t,X_{t+1})\}_{t=1}^{N}$. We aim to bound \eqref{eq:learningproblem} with a decreasing function of $N$. Because the state and action spaces are general, we learn a finite
approximation of the MDP and evaluate it through the performance of the induced optimal policy. This creates an approximation–estimation trade-off, which we analyze rigorously.

\mypara{Organization of the paper.} Section \ref{sec:approximation} derives upper bounds on the quantity
in \eqref{eq:approximationproblem}. Sections \ref{secSampleCo}
and \ref{secDistApprx} translate this bound into statistical guarantees
for the quantity in \eqref{eq:learningproblem}. In particular, we study
(a) offline learning of continuous MDPs via state-space quantization and
(b) disturbance estimation, yielding \textit{parametric} sample-complexity
rates.

\mypara{Contributions.} We summarize our technical contributions as follows:
\begin{itemize}
\item[(i)] In Section \ref{discountedcontinuitysection}, we derive conditions under which
the optimal \textit{discounted-cost} value function is Lipschitz continuous
with respect to the model, namely the cost function and the transition kernel
(Theorem \ref{upperbound1proof}).
Section \ref{averagecontinuitysection} establishes analogous results for the
optimal \textit{average-cost} value function
(Theorems \ref{upperboundacoe1} and \ref{vanishingproof1}).
For the average-cost setting, we provide two approaches: one based on a
minorization condition and another based on the vanishing-discount method.

\item[(ii)] In Section \ref{discountrobustnesssection}, we derive conditions under which
the \textit{robustness error} of the optimal \textit{discounted-cost} value
function is Lipschitz continuous with respect to the model
(Theorem \ref{thrm:discounted_robust}).
Section \ref{averagerobustnesssection} establishes the corresponding result
for the \textit{average-cost} criterion
(Theorems \ref{upperbound3acoe} and \ref{vanishingproof2}), 
using the two approaches introduced earlier.

\item[(iii)] Section \ref{secSampleCo} develops robustness guarantees for empirical model
learning and provides explicit convergence rates leading to
\textit{parametric}  sample-complexity bounds. We propose data-driven learning
algorithms and analyze two data-generation scenarios:
(a) data collected along a controlled sample path
(Theorems \ref{cor:single_tranjectory} and \ref{cor:quantize}), and
(b) data generated by a simulator of the Markov kernel
(Theorem \ref{thm:independent_transitions_both}).

\item[(iv)] Section \ref{secDistApprx} extends the analysis to robustness measured with
respect to the distribution of the driving noise
(Theorems \ref{cor:drivingnoisemainresult1}
and \ref{drivingnoisemainresult}).
We derive convergence rates in the sample size when the driving noise
distribution is estimated empirically
(Theorem \ref{thrm:empirical_driving_noise_rates}),
and obtain improved, \textit{parametric} rates under additional uniform
regularity conditions \revise{under the discounted-cost criterion} (Theorem \ref{cor:mu_n_result}).
We further consider the case in which both the model and the driving noise
are learned from data (Theorem \ref{cor:model_noise_both_learned} and the
two subsequent examples).

\end{itemize}

\subsection{Literature Review}\label{previousresults}
Stochastic control and reinforcement learning under model misspecification are fundamental to their applications and have been extensively studied. %

We start by clarifying the difference between our definition of robustness to that of robust control literature. Classical robust control methodology such as $H_\infty$ control typically considers a deterministic family of models and focuses on designing controllers that work sufficiently well for all of them (see e.g. \cite{basbern,zhou1996robust}). Modern statistical extensions of this idea, such as coarse-ID control \cite{Dean18} and distributionally robust stochastic control \cite{MohajerinEsfahani2017,insoonyangdrc,joseblanchet2024statisticallearningdistributionallyrobust}, construct the family of uncertain models via confidence sets arising from parameter estimation. Such a robustness formulation is often used to account for guaranteed stabilization \cite{Dean18} and non-Markovian dynamics \cite{joseblanchet2024statisticallearningdistributionallyrobust}. It also has a natural connection to the minimax risk \cite{joseblanchet2024statisticallearningdistributionallyrobust,bauerle2022distributionally}. In recent studies of optimal control and reinforcement learning, another line of studies focuses on the robustness to a single model misspecification \cite{GheshlaghiAzar2013,Mania2019CertaintyEI,simchowitz20a}, in part due to the simpler controller design that arises from using a single nominal model. In the special case of linear quadratic regulator, \cite{Mania2019CertaintyEI} showed that when the approximate model is sufficiently close to the reference model, solving optimal control with respect to the single approximate model can also guarantee stabilization while achieving a second-order decay in the performance gap. In this paper, we focus on the latter setting.

\mypara{Robustness to model approximation.} The robustness formulation in this paper is therefore to consider a single approximate model, and establish upper bounds for the errors induced by applying a policy designed for the approximate model. 
Our setting has been considered in  \cite{Lan81,kara2020robustness, kara2022robustness,muller1997does,hernandez1987approximation,GordienkoSystemControlLetters, GordienkoODE, GordienkoKybernetic,bozkurt2024modelapproximationmdpsunbounded} (discrete-time) and in \cite{pradhan2022robustness} (continuous-time). To our knowledge, one of the earliest studies in this setting is \cite{Lan81}, where the author considered fully observed discrete-time controlled models,   established continuity of the optimal value function with respect to models, and gave a set convergence result for sets of optimal control actions. Related results are presented in \cite{muller1997does,hernandez1987approximation}. A closely related sequence of works is \cite{kara2020robustness, kara2022robustness}, where the authors study robustness to incorrect transition kernels, focusing primarily on asymptotic convergence under weak convergence and setwise convergence, as well as on strong uniformity results with respect to total variation convergence. Specifically, a robustness result is obtained, in an asymptotic convergence sense, for the discounted-cost criterion in \cite{kara2020robustness} and for the average-cost criterion in \cite{kara2022robustness}, respectively, under {\it continuous} weak convergence of transition kernels; uniformity results under total variation convergence are also established. A unified perspective is given in \cite{KaraYuksel2021Chapter}, which considers quantized approximations as a special case. 

Another closely related sequence of studies is \cite{GordienkoSystemControlLetters, GordienkoKybernetic, GordienkoODE}, which, to our knowledge, are among the first to consider quantitative performance bounds induced by model mismatch. In \cite{GordienkoSystemControlLetters}, the authors focused on the discounted-cost criterion, obtaining a weighted norm version of the first inequality in our Theorem \ref{upperbound3proof}. In \cite{GordienkoODE}, the authors focused on the average-cost criterion, obtaining results similar to Theorem \ref{upperbound3acoe} by imposing contraction assumptions on Bellman operators of the reference and approximate models. In \cite{GordienkoKybernetic}, the authors worked with the ``relative stability index'', i.e., the robustness error divided by the optimal value function under the true model, and obtained an upper bound under the discounted criterion that is independent of the discount factor under an assumption that is similar to that of our Theorem \ref{upperbound3acoe}. We remark that a recent work \cite{bozkurt2024modelapproximationmdpsunbounded} independently recovers and extends the results in \cite{GordienkoSystemControlLetters}. We acknowledge that Theorem \ref{upperbound1proof} and \ref{thrm:discounted_robust} already appeared in \cite{bozkurt2024modelapproximationmdpsunbounded} as Corollary 1, and our Theorem \ref{upperbound3proof} is a refinement of Theorem 5 of the same paper.
In the first half of this paper, our focus is to provide a unified proof strategy to obtain results mentioned above under both discounted and average-cost criterion, and discuss the set of assumptions behind them, what is shared and what is different. 

\mypara{Model learning from data.} In the second half of the paper, we connect our robustness error bounds to the statistical rates at which the performance of a policy, obtained under a model estimated via sampling, approaches that of the optimal policy. We will refer to this general setting as {\it model learning from data}. 

\mysubpara{Model-based offline learning.} If the model estimation is done via sampling the transition kernel and cost function, the setting will fall under model-based offline learning. For finite MDPs, the statistical rates of error bounds have been well studied for both model estimation with i.i.d.~samples \cite{GheshlaghiAzar2013,pmlr-v125-agarwal20b} and model estimation from one trajectory \cite{10.1609/aaai.v37i7.25989}. Results for continuous MDPs often impose structural assumptions on the transition kernel, such as the kernel can be parameterized  \cite{Dufour2015-qu, NEURIPS2021_c21f4ce7}. In Section \ref{secSampleCo}, we build on the rich literature concerning the fundamental performance loss of static discretization in MDPs \cite{VanRoy2006, yu2012discretizedapproximationspomdpaverage, Saldi2015, Saldi20152, SaLiYuSpringer} to establish statistical rates for model-based learning under MDPs with Wasserstein-1 Lipschitz transition kernels, under both i.i.d.~sampling (Theorem \ref{thm:independent_transitions_both}) and one trajectory of samples (Theorem \ref{cor:single_tranjectory} and \ref{cor:quantize}). To our knowledge, these results are new to both RL and control literature. We remark that our setting has been considered in recent works of regret minimization with adaptive discretization \cite{Sinclair2023,kar2025policyzoomingadaptivediscretizationbased, maran2024noregretreinforcementlearningsmooth}. However, their focus is on adaptively finding the best discretization scheme, while our work focuses on the statistical performance of a fixed state partition. We also remark a recent line of works \cite{kara2022nearoptimalityfinitememory, KSYContQLearning,karayukselNonMarkovian} connecting static discretization to learning on partially observable MDPs. 

\mysubpara{Disturbance estimation.} Another framework of model estimation is by approximating the distribution of the driving ``disturbance process'' \cite{GordienkoMinorant, Gordienko2007,Gordienko2008, Gordienko2022}: consider a stochastic dynamical system $X_{t+1}=f(X_{t}, U_{t},W_{t})$, where $\{W_{t}\}_{t=0}^{\infty}$ is an i.i.d.~process with distribution $\mu$. If an alternative distribution $\nu$ is proposed and an approximate optimal policy is computed under $\nu$, the objective is to characterize the relationship between the error induced by disturbance approximation and a certain distance between $\mu$ and $\nu$. In \cite{Gordienko2022}, a robustness result is given for the discounted-cost criterion with the bounded-Lipschitz distance between $\mu$ and $\nu$. In \cite{GordienkoMinorant},  similar results are obtained for the average-cost criterion, using either the total variation distance or bounded-Lipschitz distance. 

Results on robustness to weak convergence in \cite{kara2020robustness,kara2022robustness} and to Wasserstein convergence in \cite{Gordienko2007, Gordienko2008,Dufour2015-qu} imply empirical consistency with i.i.d.~learning of models, as noted in these studies. In \cite{Gordienko2007, Gordienko2008}, the authors consider a learning scenario  where $\nu=\frac{1}{n+1}\sum_{i=0}^{n}\delta_{w_{i}}$ is the empirical measure of $\mu$ under samples $\{w_{t}\}_{t=0}^{n}$. The chosen distance between measures is Wasserstein-1 distance, with the discounted-cost criterion considered in \cite{Gordienko2007}, and the average-cost criterion in \cite{Gordienko2008}. In Section \ref{secDistApprx}, we show that the disturbance distribution approximation problem can be viewed as a special case of model approximation, and through Theorem \ref{thrm:discounted_robust} and \ref{vanishingproof2}, we obtain results analogous to those in \cite{Gordienko2007, Gordienko2008}, with relaxed assumptions and improved bounds, as presented in Theorem \ref{thrm:empirical_driving_noise_rates} and \ref{cor:mu_n_result}. We remark a further related result on empirical consistency from a different perspective is presented in \cite{hanson2021learning}.

\subsection{Regularity of MDPs}
In this subsection, we introduce definitions and assumptions related to regularity of MDPs used throughout the paper. A central goal of this paper is to establish robustness of optimal control against perturbations of the MDP. This would require the MDP to possess certain regularities. Indeed, as we will thoroughly review in Section \ref{subsec:OE}, even classical existence of optimal control would require the following set of regularity conditions.

\begin{assumption}[Basic regularity]\label{basicassump}
Unless otherwise noted, we assume that any MDP $(\bx,\bu,\ct,c)$ in this paper satisfies the following:\\
\noindent \textbf{(a).} $(\bx, d_{\bx})$ is Polish, and $(\bu,d_{\bu})$ is compact.

\noindent \textbf{(b).} $c:\bx\times\bu\to\bR$ is nonnegative, bounded (i.e., $\cinf<\infty$), and  continuous on both $\bx$ and $\bu$.

\noindent \textbf{(c).} For any $v\in C_{b}(\bx)$, $\int_{\bx}v(y)\ct(dy|x,u)$ is a continuous function of $\bx\times \bu$, i.e., $\ct$ is weakly continuous on $\bx\times\bu$.
\end{assumption}

We now discuss additional regularities. For a test function $f:\bx \to \bR$, we define the uniform discrepancy between
two controlled transition kernels $\ct$ and $\cs$ by
\begin{equation*}
d_{f}(\ct,\cs)
:=\sup_{x \in \bx,\,u \in \bu }
\left|
\int f(y)\,\ct(dy|x,u)-\int f(y)\,\cs(dy|x,u)
\right|.
\end{equation*}
We can take a further supremum over all $1-$Lipschitz functions to get the uniform Wasserstein-1 discrepancy between kernels:
\begin{equation*}
d_{\wc_{1}}(\ct,\cs):=\sup_{\flip\leq 1}d_{f}(\ct,\cs),
\end{equation*}
where $\flip$ denotes the Lipschitz seminorm of a function $f:\bx\to\bR$. It is indeed a uniform Wasserstein-1 discrepancy due to an exchange of supremum:
$$
d_{\wc_{1}}(\ct,\cs)=\sup_{\flip\leq 1}
\sup_{x \in \bx,\,u \in \bu }\left|
\int f(y)\,\left(\ct(dy|x,u)-\cs(dy|x,u)\right)
\right|=\sup_{x \in \bx,\,u \in \bu }\wc_{1}\left(\ct(\cdot|x,u),\ctp(\cdot|x,u)\right),
$$
where at the last step we used the dual characterization of Wasserstein-1 distance between probability measures (see e.g. \cite[Chapter 6]{Villani2009}).

We introduce two discrepancy measures, as they serve complementary purposes. The bound on the robustness error expressed in terms of $d_f$ is tighter and more suitable for statistical analysis. 
In contrast, $d_{\wc_1}$ is more interpretable and is often sufficient in practice. Such comparison will be clearer in Section \ref{secSampleCo} and \ref{secDistApprx}. The following lemma, which follows directly from the definitions, compares $d_{\wc_1}$ and $d_f$.

\begin{lemma}\label{trivialinequality} 
Let $f:\bx\to\bR$ be $\flip$-Lipschitz. For  controlled
transition kernels  $\ct$ and $\cs$,
\begin{equation*}
d_{f}(\ct,\cs)\leq \flip d_{\wc_{1}}(\ct,\cs).
\end{equation*}
\end{lemma}

Next, we state the central assumption of the paper. 
Let $\mathcal{P}(\bx)$ denote the set of probability measures on $\bx$. 
We adopt the equivalent formulation of a kernel as a measurable mapping from $\bx \times \bu$ to $\mathcal{P}(\bx)$.

\begin{assumption}[Wasserstein Regular MDPs]\label{myassump} Consider an MDP $(\bx,\bu,\ct,c)$ and a discount factor $\beta<1$. In addition to Assumption \ref{basicassump}, assume the following holds:

\noindent \textbf{(a).} For any $u\in\mathbb{U}$, $c:\bx\times\bu\to\bR$ is $\clip$-Lipschitz continuous as a function of $x$.

\noindent \textbf{(b).} For any $u\in\bu$, $\ct:\;\bx\times\bu\to\mathcal{P}(\bx)$ is $\tlip$-Lipschitz continuous as a function of $x$, with respect to Wasserstein-1 distance, i.e.
\begin{equation*}
\wc_{1}\left(\ct(\cdot|x,u),\ct(\cdot|y,u)\right)\leq \tlip d_{\bx}(x,y),\;\forall\;u\in\bu, x,y \in \bx.
\end{equation*}

 \noindent \textbf{(c).} $\beta\tlip<1$.
\end{assumption}

\begin{remark}
Such Lipschitz-type assumptions are commonly used in continuous-state MDPs and reinforcement learning. They are used to establish continuity and stability
of value functions and policies
\cite{SaLiYuSpringer,bozkurt2024modelapproximationmdpsunbounded}
and to obtain regret guarantees in weakly continuous MDPs
\cite{Sinclair2023,kar2025policyzoomingadaptivediscretizationbased,maran2024noregretreinforcementlearningsmooth}.
Early implications were studied in \cite{hinderer2005} and refined in
\cite{SaLiYuSpringer,pmlr-v80-asadi18a}. They also appear in belief-state
reductions of POMDPs under additional Dobrushin-type conditions
\cite{demirci2023average}, and in applications such as policy gradient
analysis and continuity of $Q$-functions \cite{Pirotta2015,Rachelson2010OnTL}.
\end{remark}

\section{Robustness to Model Approximation}\label{sec:approximation}

This section develops the main analytical tools used to relate model approximation to control performance. In Subsection \ref{subsec:OE} we recall the optimality equations and the structural properties of optimal policies. Subsection \ref{subsec:continuity} establishes continuity of the optimal value functions with respect to the model, which is then used in Subsection \ref{subsec:robust_performance_lass} to derive robustness bounds for policies computed from approximate MDPs. Finally, Subsection \ref{Lipschitzmdp} presents sufficient conditions ensuring Lipschitz regularity of the optimal (relative) value functions.

\subsection{Optimality Equations} \label{subsec:OE}
In this subsection, we introduce assumptions under which stochastic optimal control problems admit fixed-point characterizations and optimal policies can be restricted to deterministic stationary controls. 
 We denote by $C_b(\bx)$ the set of bounded continuous measurable functions from $(\bx,\mathcal{F}_{\bx})$ to $(\bR,\mathcal{B}(\bR))$.  For a function $c:\bx\times\bu\to\bR$, define the supremum norm
$\|c\|_{\infty}=\sup_{x\in\bx,\,u\in\bu}|c(x,u)|$.
This notation extends analogously to functions defined on other domains.

As noted above, Assumption~\ref{basicassump} ensures a unique fixed-point characterization of the optimal value function through the discounted cost optimality equation (DCOE), as stated in the following theorem.

\begin{theorem}[DCOE, see e.g. {\cite[Theorem 16.2]{Schal}}] \label{dcoetheorem} 
Consider an MDP $(\bx,\bu,\ct,c)$ and a discount factor $\beta\in(0,1)$. Suppose Assumption \ref{basicassump} holds.

 \noindent \textbf{(a).} The optimal discounted-cost function $J_{\beta}^{*}(c,\ct)\in C_{b}(\bx)$ is the unique solution to the following fixed point equation in $v\in C_{b}(\bx)$, which we refer to as DCOE:
\begin{equation*}
v(x)= \mathbb{T}_{\beta}(v)(x),\;\forall\;x\in\bx,\;\; \text{ where } \mathbb{T}_{\beta}(v)(x) := \inf_{u\in\bu}\left\{c(x,u)+\beta \int_{\bx}v(y)\ct(dy|x,u) \right\}.
\end{equation*}

\noindent \textbf{(b).} There exists a deterministic stationary policy $\gamma^{*}:\;\bx\to\bu$ such that
\begin{equation*}
v(x)=c(x,\gamma^{*}(x))+\beta \int_{\bx}v(y)\ct(dy|x,\gamma^{*}(x)) ,\;\;\forall\;x\in\bx,
\end{equation*}
and $J_{\beta}(c,\ct,\gamma^{*})(x)=J_{\beta}^{*}(c,\ct)(x),\;\;\forall\;x\in\bx$.
\end{theorem}

Part (a) follows from the contraction property of the Bellman operator
$\mathbb{T}_{\beta}:C_b(\bx)\to C_b(\bx)$ on
$\left(C_b(\bx),\|\cdot\|_\infty\right)$. Specifically, for any
$f_1,f_2\in C_b(\bx)$,
\begin{equation*}
\|\mathbb{T}_{\beta}f_1-\mathbb{T}_{\beta}f_2\|_\infty
\le \beta\,\|f_1-f_2\|_\infty.
\end{equation*}
Hence, by the Banach fixed-point theorem, the DCOE admits a unique solution
(see also~\cite[p.~19]{HernndezLerma1989}).
As discussed in Section \ref{Lipschitzmdp}, this uniqueness is crucial for
establishing Lipschitz continuity of the value function under suitable
assumptions.

\begin{remark}[Discounted-cost Bellman consistency equation]
Given a deterministic stationary policy $\gamma$, consider
\begin{equation*}
v(x)=\mathbb{T}_{\beta}^{\gamma}(v)(x), \quad \forall x\in\bx,
\quad \text{where }\mathbb{T}_{\beta}^{\gamma}(v)(x)
:= c(x,\gamma(x))+\beta \int_{\bx}v(y)\ct(dy\mid x,\gamma(x)).
\end{equation*}
The operator $\mathbb{T}_{\beta}^{\gamma}$ is a contraction and therefore
admits a unique solution in the space of bounded measurable functions
(since $\gamma$ need not be continuous). The solution is
$J_{\beta}(c,\ct,\gamma)$. We refer to this equation as the
\emph{discounted-cost Bellman consistency equation} for $\gamma$.
\end{remark}

For the average-cost criterion, we introduce two additional sets of
assumptions, each sufficient to ensure the validity of the Average Cost Optimality Equation (ACOE).

\begin{definition}[Minorization Condition for  Kernels]\label{minordef} 
The controlled transition kernel $\ct:\;\mathcal{F}_{\bx}\times\bx\times\bu\to[0,1]$ is said to satisfy the minorization condition with a probability measure $\rho$ and a constant $\epsilon > 0$ if for all $x\in\bx,\;u\in\bu$, and $A\in\mathcal{F}_{\bx}$, we have
\begin{align}\label{minorizationCond}
\ct(A|x,u)\geq \epsilon \rho(A).
\end{align}
\end{definition}

\begin{theorem}[ACOE]\label{minoracoe} 
 Consider an MDP $(\bx,\bu,\ct, c)$ such that $\ct$ satisfies the minorization condition (\ref{minorizationCond}) with some probability measure $\rho$ and constant $\epsilon > 0$, and suppose Assumption \ref{basicassump} holds. Then there exists   $g \in \mathbb{R}$ and $h\in C_{b}(\bx)$ such that the following fixed point equation, referred to as ACOE, holds:
\begin{equation*}
g+h(x)=\inf_{u\in\bu} \left\{c(x,u)+\int_{\bx}h(y)\ct(dy|x,u)\right\},\;\forall\;x\in\bx,
\end{equation*}
where $g=\inf_{\gamma\in\Gamma_{A}}J_{\infty}(c,\ct,\gamma)(x)$ for all $x \in \bx$. Furthermore, there exists a deterministic stationary policy $\gamma^{*}:\;\bx\to\bu$ such that
\begin{equation*}
g+h(x)=c(x,\gamma^{*}(x))+\int_{\bx}h(y)\ct(dy|x,\gamma^{*}(x)),\;\forall\;x\in\bx,
\end{equation*}
and $J_{\infty}(c,\ct,\gamma^{*})(x)=g,\;\forall\;x\in\bx$.
\end{theorem}

\begin{remark} \label{remark:g_c_T}
We refer to $(g^{*},h^{*},\gamma^{*})$ satisfying ACOE as a canonical triplet. To emphasize the dependence on the cost function $c$ and kernel $\ct$, we use the following notation for the canonical triplet:
$(g_{c,\ct}^{*}, h_{c,\ct}^{*}, \gamma^{*}_{c,\ct})$. Note that 
$g_{c,\ct}^{*} = J_{\infty}^{*}(c,\ct)$.

In general, the function $h_{c,\ct}^{*}$ solving the ACOE is not unique. However, under the minorization condition, such a function can be obtained constructively. Define $h_{c,\ct,\epsilon}^{*}$ as the unique fixed point of the contraction map on $C_b(\bx)$
\begin{equation}\label{equ:minor_operator}
\mathbb{T}_{c,\ct,\epsilon}v(x):=\inf_{u\in\bu}\left\{c(x,u)+\int_{\bx}v(y)\left(\ct(dy|x,u)-\epsilon\rho(dy)\right)\right\},\;\forall\;x\in\bx.
\end{equation}
Then $g_{c,\ct}^*:=\epsilon\int_{\bx}h_{c,\ct,\epsilon}^{*}\rho(dx)$, and the pair $(g_{c,\ct}^{*}, h_{c,\ct,\epsilon}^{*})$ satisfies the ACOE. Finally, via a verification theorem (e.g. \cite[Theorem 2.2 (a)+(b)]{HernandezLermaMCP}), the above construction validates the theorem.
\end{remark}

\begin{remark}[Average-cost Bellman consistency equation]\label{remakr:avg_cost_bell}
Under Assumption \ref{basicassump} and the minorization condition (\ref{minorizationCond}), for a deterministic stationary policy $\gamma$, the operator
\begin{equation*}
\mathbb{T}^{\gamma}v(x)
:=c(x,\gamma(x))
+\int_{\bx}v(y)\bigl(\ct(dy\mid x,\gamma(x))-\epsilon\rho(dy)\bigr),
\quad \forall x\in\bx,
\end{equation*}
is a contraction under the supremum norm. Consequently, there exist a constant $g^{\gamma} \in \mathbb{R}$ and a
bounded measurable function $h^{\gamma}$ such that
\begin{equation*}
g^{\gamma}+h^{\gamma}(x)
=c(x,\gamma(x))
+\int_{\bx}h^{\gamma}(y)\ct(dy\mid x,\gamma(x)),
\end{equation*}
where $g^{\gamma}=J_{\infty}(c,\ct,\gamma)$. We refer to this equation as the \emph{average-cost Bellman consistency equation}.
\end{remark}

Using the ACOE and the minorized fixed-point equation \eqref{equ:minor_operator},
we state several properties of $h_{c,\ct,\epsilon}^{*}$. These results are
classical (see, e.g., \cite[p.~61]{HernndezLerma1989}); for completeness,
we provide a proof in Appendix \ref{sec:hpropertiesproof}.

\begin{lemma}[Properties of $h_{c,\ct,\epsilon}^{*}$]\label{lemma:hproperties}
Consider an MDP $(\bx,\bu,\ct, c)$ such that $\ct$ satisfies the minorization condition with some probability measure $\rho$ and constant $\epsilon_{1} > 0$, and suppose Assumption \ref{basicassump} holds.  

\noindent \textbf{(a).} Let $\epsilon_{2}$ be such that $\epsilon_{1}>\epsilon_{2}>0$. The difference between $h_{c,\ct,\epsilon_{1}}^{*}$ and $h_{c,\ct,\epsilon_{2}}^{*}$ is a constant function. 

 \noindent \textbf{(b).} For each $\epsilon\in (0,\epsilon_{1}]$, 
we have $\left\Vert h_{c,\ct,\epsilon}^{*}\right\Vert_{\mathrm{Lip}} = \left\Vert h_{c,\ct,\epsilon_1}^{*}\right\Vert_{\mathrm{Lip}}$.

\noindent \textbf{(c).} For each $\epsilon \in (0,\epsilon_{1}]$, $h_{c,\ct,\epsilon}^{*}$ is a bounded function. In particular,
$\left\Vert h_{c,\ct,\epsilon}^{*} \right\Vert_{\infty}\leq \Vert c\Vert_{\infty}/\epsilon$.
\end{lemma}

\begin{remark}
    \label{rk:small_epsilon}
Lemma \ref{lemma:hproperties}(b) allows us to denote
$\|h_{c,\ct,\epsilon'}^{*}\|_{\mathrm{Lip}}$ for $\epsilon' \in (0,\epsilon_1]$ by $\hlip$.
\end{remark}

We now introduce a second set of assumptions sufficient for existence of a solution to the ACOE, based on the transition kernel being a Wasserstein contraction.

\begin{assumption} \label{equicontinuousassump} Let $L, \widetilde{L} > 0$ be constants. Assume the following hold for an MDP $(\bx,\bu,\ct,c)$:

\noindent \textbf{(a).}  $\bx$ is compact.

\noindent \textbf{(b).} There exists some $\beta^* \in (0,1)$ such that
$\|\jsc\|_{\mathrm{Lip}} \leq L, \text{ for all } \beta \in [\beta^*,1)
$.

\noindent \textbf{(c).} The map $\ct:\;\bx\times\bu\to\mathcal{P}(\bx)$ is $\widetilde{L}$-Lipschitz continuous with respect to Wasserstein-1 distance, in the sense that
\begin{equation*}
\wc_{1}\left(\ct(\cdot|x,u),\ct(\cdot|y,u')\right)\leq \widetilde{L}\left(d_{\bx}(x,y) + d_{\bu}(u,u')\right),\;\forall\;u,u'\in\bu,\; x,y \in \bx.
\end{equation*}
\end{assumption}

Under Assumption \ref{equicontinuousassump}, the average-cost setting can be analyzed via the vanishing discount factor approximation.

\begin{theorem}\label{thrm:vanish_main}
Suppose Assumption \ref{basicassump} and   Assumption \ref{equicontinuousassump}(a)-(b) hold for an MDP $(\bx,\bu,\ct, c)$. Then there exist a \textit{canonical triplet} $(g_{c,\ct}^{*}, h_{c,\ct}^{*}, \gamma^{*}_{c,\ct})$, a state $z \in \bx$, and an increasing subsequence  $\{\beta(n), n \geq 1\} \subset (0,1)$ such that $\lim_{n \to \infty}\beta(n) = 1$,  and for all $x \in \bx$,
\begin{align*}
\lim_{n\to\infty}\left(1-\beta(n)\right)J_{\beta(n)}^{*}(c,\ct)(x)=g_{c,\ct}^{*}=J_{\infty}^{*}(c,\ct), \quad
\lim_{n\to\infty} h_{c,\ct,\beta(n)}(x) = h_{c,\ct}^{*}(x).
\end{align*}
where $h_{c,\ct,\beta(n)}(x) := J_{\beta(n)}^{*}(c,\ct)(x) - J_{\beta(n)}^{*}(c,\ct)(z)$, and $h_{c,\ct}^{*}$ is $L$-Lipschitz with respect to the state variable.

Suppose, in addition, Assumption \ref{equicontinuousassump}(c) holds. Define for each $n\geq 1$, and $(x,u) \in \bx \times \bu$
$$
\mathcal{I}_n(x,u) = \int h_{c,\ct,\beta(n)}(x') \ct(dx'|x,u), \quad
\mathcal{I}_{\infty}(x,u) = \int h_{c,\ct}^*(x') \ct(dx'|x,u).
$$
Then $\{\mathcal{I}_n: n \geq 1\}$ and $\mathcal{I}_{\infty}$ are $L \times \widetilde{L}$-Lipschitz functions on $\bx \times \bu$, and $\lim_{n \to \infty} \mathcal{I}_n(x,u) = \mathcal{I}_{\infty}(x,u)$ for each $(x,u) \in \bx \times \bu$.
\end{theorem}

\begin{proof} The first claim follows from the proof of Lemma 2.5 in \cite{demirci2023average}. We now focus on the second claim. Note that for each $n \geq 1$, $h_{c,\ct,\beta(n)}(z) = 0$ and $h_{c,\ct,\beta(n)}$ is ${L}$-Lipschitz on $\bx$. Since $\bx$ is compact and $\lim_{n\to\infty} h_{c,\ct,\beta(n)}(x) = h^{*}(x)$ for each $x \in \bx$, by bounded convergence theorem, we have $\lim_{n \to \infty} \mathcal{I}_n(x,u) = \mathcal{I}_{\infty}(x,u)$ for each $(x,u) \in \bx \times \bu$. Finally, for each $n \geq 1$, $x,y \in \bx$ and $u,u' \in \bu$, note that 
\begin{align*}
&\left|\mathcal{I}_n(x,u) - \mathcal{I}_{n}(y,u')\right| = \left|\int h_{c,\ct,\beta(n)}(x') \ct(dx'|x,u) - \int h_{c,\ct,\beta(n)}(x') \ct(dx'|y,u')\right| \\
\leq &{L} \wc_{1}\left(\ct(\cdot|x,u),\ct(\cdot|y,u')\right) \leq L \widetilde{L}\left(d_{\bx}(x,y) + d_{\bu}(u,u')\right),
\end{align*}
where we use the fact that $\| h_{c,\ct,\beta(n)}\|_{\mathrm{Lip}}$ is bounded by $L$ and Assumption \ref{equicontinuousassump}(c). Thus, $\mathcal{I}_n: \bx \times \bu \to \mathbb{R}$ is Lipschitz with a constant $L \times \widetilde{L}$. By the same argument, we have $\mathcal{I}_\infty: \bx \times \bu \to \mathbb{R}$ is also $L\times \widetilde{L}$-Lipschitz. The proof is complete.
\end{proof}

\begin{remark}\label{remark:further_subseq}
By the proof of Lemma 2.5 in \cite{demirci2023average}, in fact, we can show that for any increasing sequence $\{\beta(n); n \geq 1\} \subset (0,1)$ such that $\lim_{n \to \infty}\beta(n) = 1$, we can extract a further subsequence $\{\beta(n_k); n_k \geq 1\} \subset (0,1)$ such that the conclusions in Theorem \ref{thrm:vanish_main} continue to hold for this subsequence. Further, the state $z \in \bx$ may be chosen arbitrarily.
\end{remark}

\revise{
We conclude this subsection with a comparison lemma that converts a one-step Bellman-type inequality into an upper bound on the average cost of a fixed stationary policy. This result will be useful in deriving average-cost robustness bounds, especially when we do not impose the preceding regularity assumptions on the model under which the policy is evaluated.

\begin{lemma}[{\cite[Lemma 5.2.5 (a)]{HernandezLermaMCP}}]\label{lemma:comparison}
Let $g\in\mathbb{R}$, let $h: \mathbb{X}\to\mathbb{R}$ be a bounded measurable function, and let $\gamma: \mathbb{X}\to \bu$ be a measurable selector. If
\begin{equation*}
g + h(x) \geq c(x, \gamma(x)) + \int_{\mathbb{X}} h(y)\, \mathcal{T}(dy \mid x, \gamma(x)), \quad \forall\, x \in \mathbb{X},
\end{equation*}
then $g \geq J_\infty(c, \mathcal{T}, \gamma)(x)$ for all $x \in \mathbb{X}$. 
\end{lemma}}

\subsection{Continuity of Optimal Value Functions in Models}\label{subsec:continuity}
We begin by establishing the continuity of the optimal value functions with respect to the model components, including the cost function and the transition kernel. These results serve as key intermediate steps in our robustness analysis in Subsection \ref{subsec:robust_performance_lass}. The analysis is conducted separately for the discounted-cost criterion in Subsection \ref{discountedcontinuitysection} and the average-cost criterion in Subsection \ref{averagecontinuitysection}. Related continuity results can be found in \cite{kara2022robustness, kara2020robustness}.

\subsubsection{Discounted-cost Criterion}\label{discountedcontinuitysection}
\begin{theorem}\label{upperbound1proof} Suppose that Assumption \ref{basicassump} holds for two MDPs, (reference) $(\bx,\bu,\ct,c)$ and (approximation) $(\bx,\bu,\ctp,\altc)$. Then
\begin{equation*}
\left\Vert\jsc-\jscp\right\Vert_{\infty}\leq \frac{1}{1-\beta}\left\Vert c-\altc\right\Vert_{\infty}+\frac{\beta}{1-\beta} d_{\jsc}(\ct,\cs).
\end{equation*}
If in addition, $\jsc$ is Lipschitz continuous with respect to the state variable, then
\begin{equation*}
\left\Vert\jsc-\jscp\right\Vert_{\infty}\leq \frac{1}{1-\beta}\left\Vert c-\altc\right\Vert_{\infty}+\frac{\beta}{1-\beta} \left\Vert J_{\beta}^{*}(c,\ct)\right\Vert_{\mathrm{Lip}} d_{\wc_{1}}(\ct,\ctp).
\end{equation*}
\end{theorem}
\begin{proof} For ease of notation, we use the following abbreviation:
\begin{equation}\label{ease_of_notation}
v_{c,\ct}:=J_{\beta}^{*}(c,\ct),\quad v_{\altc,\ctp}:=J_{\beta}^{*}(\altc,\ctp).
\end{equation}
By Theorem \ref{dcoetheorem}, we have that for any $x\in\bx$:
\begin{align*}
&v_{c,\ct}(x)=\inf_{u\in\bu}\left(c(x,u)+\beta\int v_{\ct}(x')\ct(dx'|x,u) \right),\;\;
v_{\altc,\cs}(x)=\inf_{u\in\bu}\left(\altc(x,u)+\beta\int v_{\cs}(x')\ctp(dx'|x,u) \right).
\end{align*}
Therefore, by the triangular inequality, we have that for any $x\in\bx$:
\begin{align*}
&|v_{c,\ct}(x)-v_{\altc,\cs}(x)| \leq \sup_{u\in\bu}\left|c(x,u)-\altc(x,u)\right|+\beta\sup_{u\in\bu}\left|\int v_{c,\ct}(x')\ct(dx'|x,u)-\int v_{\altc,\cs}(x')\ctp(dx'|x,u) \right|\\
&\leq \left\Vert c-\altc\right\Vert_{\infty}+\beta\sup_{u\in\bu}\left| \int v_{c,\ct}(x')\ct(dx'|x,u)-\int v_{c,\ct}(x')\ctp(dx'|x,u)\right|\\
&\qquad \qquad \quad \;+ \beta\sup_{u\in\bu} \left|\int v_{c,\ct}(x')\ctp(dx'|x,u)-\int v_{\altc,\ctp}(x')\ctp(dx'|x,u) \right|\\
&\leq \left\Vert c-\altc\right\Vert_{\infty}+\beta d_{\jsc}(\ct,\cs) +\beta\left\Vert v_{c,\ct}-v_{\altc,\ctp}\right\Vert_{\infty}.
\end{align*}
Taking the supremum over all $x\in\bx$, and using the boundedness of $v_{c,\ct}$ and $v_{\altc,\ctp}$, we can rearrange terms to establish the first claim. The second claim then follows from Lemma \ref{trivialinequality}. 
\end{proof}

\subsubsection{Average-cost criterion}\label{averagecontinuitysection}
In this subsection, we first establish an analogous result to Theorem \ref{upperbound1proof} under the average-cost criterion, assuming the minorization condition (Definition \ref{minordef}). Next, we present an alternative method---based on a vanishing discount factor and Theorem \ref{upperbound1proof}---that yields a similar result without requiring the minorization condition.

Recall that $h_{c,\ct,\epsilon}^{*}$ is the unique fixed-point of $\mathbb{T}_{c,\ct,\epsilon}$ defined in \eqref{equ:minor_operator}. Further, recall the definition of $\hlip$ in Remark \ref{rk:small_epsilon}.

\begin{theorem}\label{upperboundacoe1}
 Suppose that Assumption \ref{basicassump} holds for two MDPs, (reference) $(\bx,\bu,\ct,c)$ and (approximation) $(\bx,\bu,\ctp,\altc)$. Further, assume $\ct$ (resp.~$\ctp$) satisfies the minorization condition with a probability measure $\rho$ (resp.~$\tau$) and a constant $\epsilon > 0$. Then
 \begin{align*}
     \left\Vert J_{\infty}^{*}(c,\mathcal{T})-J_{\infty}^{*}(\altc,\cs)\right\Vert_{\infty} \leq 2\epsilon \Delta_{c,\ct,\epsilon}(\rho,\tau) +
          \|c - \altc\|_{\infty} \; + \; d_{h_{c,\ct,\epsilon}^{*}}(\ct, \ctp),
 \end{align*}
where we define
\begin{equation}
    \label{def:Delta_epsilon_rhotau}
    \Delta_{c,\ct,\epsilon}(\rho,\tau) := \left|\int_{\bx}  h_{c,\ct,\epsilon}^{*}(x)(\rho-\tau)(dx)\right|.
\end{equation}
 If in addition, $h_{c,\ct,\epsilon}^{*}$ is Lipschitz continuous with respect to the state variable, we have
\begin{align*}
\left\Vert J_{\infty}^{*}(c,\mathcal{T})-J_{\infty}^{*}(\altc,\cs)\right\Vert_{\infty}  \leq \left\Vert c-\altc\right\Vert_{\infty} + \hlip d_{\wc_{1}}(\ct,\ctp).
\end{align*}
 \end{theorem}
\begin{proof} 
Recall that $J_{\infty}^{*}(c,\mathcal{T})(x) = g_{c,\ct}^{*}$ and $J_{\infty}^{*}(\altc,\cs)(x) = g_{\altc,\ctp}^{*}$ for $x \in \bx$. 
By Remark  \ref{remark:g_c_T}, we have for $x \in \bx$,
\begin{align}\label{aux:ACOE_eq}
\begin{split}
&h_{c,\ct,\epsilon}^{*}(x)=\inf_{u\in\bu}\left(c(x,u)+\int_{\bx} h_{c,\ct,\epsilon}^{*}(x')\left(\mathcal{T}(dx'|x,u)-\epsilon\rho(dx')\right)\right),\quad
g_{c,\ct}^{*}=\epsilon\int_{\bx}h_{c,\ct,\epsilon}^{*}(x')\rho(dx'),\\
& h_{\altc,\cs,\epsilon}^{*}(x)=\inf_{u\in\bu}\left(\altc(x,u)+\int_{\bx} h_{\altc,\cs,\epsilon}^{*}(x')\left(\mathcal{S}(dx'|x,u)-\epsilon\tau(dx')\right)\right), \quad
g_{\altc,\cs}^{*}=\epsilon\int_{\bx} h_{\altc,\cs,\epsilon}^{*}(x')\tau(dx').
\end{split}
\end{align} 
By subtracting the second equation from the first, we have that for each $x \in \bx$,
\begin{align*}
&\left|h_{c,\ct,\epsilon}^{*}(x)-h_{\altc,\ctp,\epsilon}^{*}(x)\right| \leq \sup_{u\in\bu}\left|c(x,u)-\altc(x,u)\right| + (I) + (II)
\end{align*}
where we define
\begin{align*}
&(I) := \sup_{u\in\bu}\left| \int_{\bx} h_{c,\ct,\epsilon}^{*}(x')\left(\ct(dx'|x,u)-\epsilon\rho(dx')\right)\right.
\left.-\int_{\bx} h_{c,\ct,\epsilon}^{*}(x')\left(\ctp(dx'|x,u)-\epsilon\tau(dx')\right) \right|,\\
& (II) := \sup_{u\in\bu}  \int_{\bx} \left|h_{c,\ct,\epsilon}^{*}(x') - h_{\altc,\ctp,\epsilon}^{*} (x')\right| \left(\ctp(dx'|x,u)-\epsilon\tau(dx')\right).
\end{align*}
For the two terms, by definition,
\begin{align*}
    (I) \leq  d_{h_{c,\ct,\epsilon}^{*}}(\ct, \ctp) \; + \;
     \epsilon \Delta_{c,\ct,\epsilon}(\rho,\tau),\quad
     (II) \leq \|h_{c,\ct,\epsilon}^{*} - h_{\altc,\ctp,\epsilon}^{*}\|_{\infty} (1-\epsilon).
\end{align*}
As a result, by rearranging terms and taking the supremum over $x \in \bx$, we have
\begin{align}\label{aux:avg_minor_cont}
    \|h_{c,\ct,\epsilon}^{*} - h_{\altc,\ctp,\epsilon}^{*}\|_{\infty} \leq \frac{1}{\epsilon} \left(\|c - \altc\|_{\infty} \; + \; d_{h_{c,\ct,\epsilon}^{*}}(\ct, \ctp)\right) \; + \;\Delta_{c,\ct,\epsilon}(\rho,\tau).
\end{align}
Finally, by the relationships between $h_{c,\ct,\epsilon}^{*}$ and $g_{c,\ct,\epsilon}^{*}$, and between $h_{\altc,\ctp,\epsilon}^{*}$ and $g_{\altc,\ctp,\epsilon}^{*}$, and due to the triangle inequality, we have
\begin{align*}
    |g_{c,\ct}^{*} - g_{\altc,\ctp}^{*}| & \leq \epsilon 
    \left|\int_{\bx}h_{c,\ct,\epsilon}^{*}(x')\rho(dx') - \int_{\bx}h_{c,\ct,\epsilon}^{*}(x')\tau(dx')\right|
    +
    \epsilon \int_{\bx} |h_{c,\ct,\epsilon}^{*} - h_{\altc,\ctp,\epsilon}^{*}| \tau(dx') \\
    &\leq \epsilon \Delta_{c,\ct,\epsilon}(\rho,\tau) +
      \epsilon\|h_{c,\ct,\epsilon}^{*} - h_{\altc,\ctp,\epsilon}^{*}\|_{\infty}  \leq 2\epsilon \Delta_{c,\ct,\epsilon}(\rho,\tau) +
          \|c - \altc\|_{\infty} \; + \; d_{h_{c,\ct,\epsilon}^{*}}(\ct, \ctp).
\end{align*}
This proves the first claim. For the second part, we note by Lemma \ref{trivialinequality}, we have
$$
|g_{c,\ct}^{*} - g_{\altc,\ctp}^{*}|\leq 2\epsilon \hlip d_{\wc_{1}}(\rho, \tau)+\|c - \altc\|_{\infty} \; + \; \hlip d_{\wc_{1}}(\ct, \ctp),
$$
where, due to Lemma \ref{lemma:hproperties} \textbf{(b)}, we use the fact that $\left\Vert h_{c,\ct,\epsilon}^{*}\right\Vert_{\mathrm{Lip}}$ is independent of $\epsilon$ (see Remark \ref{rk:small_epsilon}). Finally, we send $\epsilon \to 0$, and the proof of the second claim is complete.
\end{proof}

\begin{remark} 
Above we assumed that the minorants $\rho$ and $\tau$ share the same scaling constant $\epsilon$. This entails no loss of generality: if a minorant $\rho$ admits a scaling constant $\epsilon_{1}$, then any $\epsilon_{2}\in (0, \epsilon_{1}]$ is also valid. Hence the assumption can always be satisfied by choosing the smaller constant.
\end{remark}

Next, we present an approach based on the vanishing-discount method to achieve similar results under different assumptions, utilizing Theorem \ref{thrm:vanish_main}. 

\begin{theorem}\label{vanishingproof1}
Suppose Assumption \ref{basicassump} and   Assumption \ref{equicontinuousassump}(a)-(b) hold for two MDPs, (reference) $(\bx,\bu,\ct,c)$ and (approximation) $(\bx,\bu,\ctp,\altc)$. Then
\begin{equation*}
\left\Vert J_{\infty}^{*}(c,\mathcal{T})-J_{\infty}^{*}(\altc,\cs)\right\Vert_{\infty}\leq \left\Vert c-\altc\right\Vert_{\infty}+ L d_{\wc_{1}}(\ct,\ctp),
\end{equation*}
where $L$ is the constant in Assumption \ref{equicontinuousassump}. If additionally Assumption \ref{equicontinuousassump}(c) holds for both MDPs, then 
\begin{equation*}
\left\Vert J_{\infty}^{*}(c,\mathcal{T})-J_{\infty}^{*}(\altc,\cs)\right\Vert_{\infty}\leq \left\Vert c-\altc\right\Vert_{\infty}\;+\; d_{h_{c,\ct}^{*}}(\ct, \ctp),
\end{equation*}
where $h_{c,\ct}^{*}$ appears in Theorem \ref{thrm:vanish_main}.
\end{theorem}
\begin{proof}
By Theorem \ref{upperbound1proof} and due to Assumption \ref{equicontinuousassump}(b), for any $\beta\in [\beta^*,1)$, we have
\begin{align}\label{aux_ys}
\begin{split}
    &\left\Vert(1-\beta)\left(\jsc-\jscp\right)\right\Vert_{\infty}\leq \left\Vert c-\altc\right\Vert_{\infty}+\beta L d_{\wc_{1}}(\ct,\ctp), \\
&\left\Vert(1-\beta)\left(\jsc-\jscp\right)\right\Vert_{\infty}\leq \left\Vert c-\altc\right\Vert_{\infty}+ \beta d_{\jsc}(\ct,\cs).
\end{split}
\end{align}
By Theorem \ref{thrm:vanish_main} and Remark \ref{remark:further_subseq}, there exist  \textit{canonical triplets} $(g_{c,\ct}^{*}, h_{c,\ct}^{*}, \gamma^{*}_{c,\ct})$ and 
 $ (g_{\altc,\cs}^{*}, h_{\altc,\cs}^{*}, \gamma^{*}_{\altc,\cs})$, a state $z \in \bx$, and an increasing subsequence  $\{\beta(n); n \geq 1\} \subset (0,1)$   such that $\lim_{n \to \infty}\beta(n) = 1$,  and for all $x \in \bx$,
\begin{align*}
&\lim_{n\to\infty}\left(1-\beta(n)\right)J_{\beta(n)}^{*}(c,\ct)(x)=g_{c,\ct}^{*}=J_{\infty}^{*}(c,\ct), \quad
\lim_{n\to\infty} h_{c,\ct,\beta(n)}(x) = h_{c,\ct}^{*}(x),\\
&\lim_{n\to\infty}\left(1-\beta(n)\right)J_{\beta(n)}^{*}(\altc,\cs)(x)=g_{\altc,\cs}^{*}=J_{\infty}^{*}(\altc,\cs), \quad
\lim_{n\to\infty} h_{\altc,\cs,\beta(n)}(x) = h_{\altc,\cs}^{*}(x).
\end{align*}
where $h_{c,\ct,\beta(n)}(x) := J_{\beta(n)}^{*}(c,\ct)(x) - J_{\beta(n)}^{*}(c,\ct)(z)$ and $h_{\altc,\cs,\beta(n)}(x) := J_{\beta(n)}^{*}(\altc,\cs)(x) - J_{\beta(n)}^{*}(\altc,\cs)(z)$.

Thus, by taking the limit of the first equation in \eqref{aux_ys} along the subsequence  $\{\beta(n); n \geq 1\}$, we immediately obtain the first claim.

Further,  by taking the limit of the second equation in \eqref{aux_ys} along the subsequence, we have
\begin{align*}
 \left\Vert J_{\infty}^{*}(c,\mathcal{T})-J_{\infty}^{*}(\altc,\cs)\right\Vert_{\infty}   \leq  \left\Vert c-\altc\right\Vert_{\infty}+ \limsup_{n \to \infty}  d_{J_{\beta(n)}^{*}(c,\ct)}(\ct,\ctp)  
 = \left\Vert c-\altc\right\Vert_{\infty}+ \limsup_{n \to \infty} \sup_{(x,u) \in \bx \times \bu} \left|\widetilde{\mathcal{I}}_{n}(x,u)\right|,
\end{align*}
where we define
\begin{align*}
   &\widetilde{\mathcal{I}_{n}}(x,u) :=  \int h_{c,\ct,\beta(n)}(x')\ct(dx'|x,u)-\int h_{c,\ct,\beta(n)}(x')\ctp(dx'|x,u), \\
 &\widetilde{\mathcal{I}}_{\infty}(x,u) :=  \int h_{c,\ct}^{*}(x')\ct(dx'|x,u)-\int h_{c,\ct}^{*}(x')\ctp(dx'|x,u).
\end{align*}
By Theorem \ref{thrm:vanish_main}, since Assumption \ref{equicontinuousassump}(c) is imposed for the second claim, we have: (i). $\bx$ is compact; (ii). $\lim_{n \to \infty} \widetilde{\mathcal{I}}_{n} (x,u) =  \widetilde{\mathcal{I}}_{\infty}(x,u)$ for each $(x,u) \in \bx \times \ \bu$; (iii).   $\{|\widetilde{\mathcal{I}}_n|: n \geq 1\}, |\widetilde{\mathcal{I}}_{\infty}|$ are $L \times \widetilde{L}$-Lipschitz functions on $\bx \times \bu$, and thus equicontinuous. Now, by the definition of $\limsup$, we can find a subsequence $\{\kappa_n; n \geq 1\}$ such that 
\begin{align*}
  \limsup_{n \to \infty} \sup_{(x,u) \in \bx \times \bu} \left|\widetilde{\mathcal{I}}_{n}(x,u)\right| = \lim_{n \to \infty} \sup_{(x,u) \in \bx \times \bu} \left|\widetilde{\mathcal{I}}_{\kappa_n}(x,u)\right|.
\end{align*}
Due to (i) and (iii) (which implies uniform boundedness), by Arzel\`a-Ascoli Theorem, we can find a further subsequence $\{\kappa_{\ell_n}; n \geq 1\}$ of $\{\kappa_n: n\geq 1\}$ such that $\widetilde{\mathcal{I}}_{\kappa_{\ell_n}}$ converges to $\widetilde{\mathcal{I}}_{\infty}$ uniformly on $\bx \times \bu$ as $n \to \infty$, which implies
     \begin{align*}
    \lim_{n \to \infty} \sup_{(x,u) \in \bx \times \bu} \left|\widetilde{\mathcal{I}}_{\kappa_n}(x,u)\right| =\lim_{n \to \infty} \sup_{(x,u) \in \bx \times \bu} \left|\widetilde{\mathcal{I}}_{\kappa_{\ell_n}}(x,u)\right| 
= \sup_{(x,u) \in \bx \times \bu}\left|\widetilde{\mathcal{I}}_{\infty}(x,u)\right|
= d_{h_{c,\ct}^{*}}(\ct, \ctp).
\end{align*}
The proof is then complete.
\end{proof}
\begin{remark}
Note that the first claim in Theorem \ref{vanishingproof1} does not require Assumption \ref{equicontinuousassump}(c), while the second claim imposes it for \textit{both} MDPs.
\end{remark}

\subsection{Bounds on Robustness Errors due to Model Misspecification}\label{subsec:robust_performance_lass}
We now use the continuity results in Section \ref{subsec:continuity} to establish our main results: an upper bound on the performance loss incurred by applying a policy that is optimal with respect to an \textit{approximate} MDP. This bound is expressed in terms of the supremum norm between cost functions and two types of distances between the transition kernels: one based on their discrepancy with respect to optimal value functions, and the other given by the uniform Wasserstein-1 distance.

\subsubsection{Discounted-cost Criterion}\label{discountrobustnesssection}

We denote by $\gamma_{\altc,\cs,\beta}^{*}$ the optimal policy for the approximate MDP $(\bx,\bu,\ctp,\altc)$ under the \textit{$\beta$-discounted cost} criterion in this subsection.

\begin{theorem}\label{thrm:discounted_robust}\label{upperbound3proof}
 Suppose that Assumption \ref{basicassump} holds for two MDPs, (reference) $(\bx,\bu,\ct,c)$ and (approximation) $(\bx,\bu,\ctp,\altc)$. Then
\begin{align}
&\left\Vert\jb(c,\ct,\gamma_{\altc,\cs,\beta}^{*})-\jsc\right\Vert_{\infty} \leq \frac{2}{(1-\beta)^{2}}\left\Vert c-\altc\right\Vert_{\infty} +\frac{2\beta }{(1-\beta)^{2}}d_{\jsc}(\ct,\cs),\label{eq:discountedrobust1}\\
&\left\Vert\jb(c,\ct,\gamma_{\altc,\cs,\beta}^{*})-\jsc\right\Vert_{\infty} \leq \frac{2}{1-\beta}\left\Vert c-\altc\right\Vert_{\infty} +\frac{\beta }{1-\beta}\left(d_{\jsc}(\ct,\cs) 
+  d_{J_{\beta}^*(\altc, \ctp)}(\ct,\cs) 
\right),\label{eq:discountedrobust2}\\
&\left\Vert\jb(c,\ct,\gamma_{\altc,\cs,\beta}^{*})-\jsc\right\Vert_{\infty} \leq \frac{2}{1-\beta}\left\Vert c-\altc\right\Vert_{\infty} +\frac{2\beta }{1-\beta}  d_{J_{\beta}^*(\altc, \ctp)}(\ct,\cs). \label{eq:discountedrobust3}
\end{align}
\end{theorem}

\begin{proof}
For ease of notation, we adopt the abbreviations in \eqref{ease_of_notation}, and 
let
$$
\tilde{v}(x):=\jb(c,\ct,\gamma_{\altc,\cs,\beta}^{*})(x) \quad \text{ for } \quad x \in \bx.
$$
By the triangle inequality, we have $\left\Vert \tilde{v}-v_{c,\ct}\right\Vert_{\infty}  \leq \left\Vert \tilde{v}-v_{\altc,\cs}\right\Vert_{\infty}+\left\Vert v_{\altc,\cs}-v_{c,\ct}\right\Vert_{\infty}$. Then, by Theorem \ref{upperbound1proof} and the symmetry between $\ct$ and $\ctp$, we have
\begin{align}
   \label{robust_disc_aux}
   \begin{split}
       &\left\Vert \tilde{v}-v_{c,\ct}\right\Vert_{\infty}  \leq \left\Vert \tilde{v}-v_{\altc,\cs}\right\Vert_{\infty} + \frac{1}{1-\beta}\left\Vert c-\altc\right\Vert_{\infty}+\frac{\beta}{1-\beta}d_{\jsc}(\ct,\cs) \\
    &\left\Vert \tilde{v}-v_{c,\ct}\right\Vert_{\infty}  \leq \left\Vert \tilde{v}-v_{\altc,\cs}\right\Vert_{\infty} + \frac{1}{1-\beta}\left\Vert c-\altc\right\Vert_{\infty}+\frac{\beta}{1-\beta}d_{J_{\beta}^*(\altc, \ctp)}(\ct,\cs).
       \end{split}
\end{align}

Next, we bound $\left\Vert \tilde{v}-v_{\altc,\cs}\right\Vert_{\infty}$. By the discounted-cost Bellman consistency equation, we have:
\begin{align*}
&\tilde{v}(x)=c(x,\gamma_{\altc,\cs,\beta}^{*}(x))+\beta\int \tilde{v}(x')\ct(dx'|x,\gamma_{\altc,\cs,\beta}^{*}(x)),  \text{ for } x \in \bx, \\
&v_{\altc,\cs}(x)=\altc(x,\gamma_{\altc,\cs,\beta}^{*}(x))+\beta\int v_{\altc,\cs}(x')\ctp(dx'|x,\gamma_{\altc,\cs,\beta}^{*}(x)), \text{ for } x \in \bx.
\end{align*}
Taking the difference between them, we obtain
\begin{align}
\begin{split}
|\tilde{v}(x)-\vcn(x)|
&\leq \| c - \hat{c} \|_{\infty} + \beta \left| \int \left(\tilde{v}(x')-v_{\altc,\cs}(x')\right)\ct(dx'|x,\gamma_{\altc,\cs,\beta}^{*}(x))\right| + \beta (I) \\
&\leq \| c - \hat{c} \|_{\infty} + \revise{\beta \left\Vert  \tilde{v}-v_{\hat{c},\cs}\right\Vert_{\infty}}+ \beta (I)
\end{split}
\label{aux:discounted_robustness}
\end{align}
where we define
\begin{align*}
(I) := \left|\int v_{\altc,\cs}(x')\left(\ct(dx'|x,\gamma_{\altc,\cs,\beta}^{*}(x))-\ctp(dx'|x,\gamma_{\altc,\cs,\beta}^{*}(x))\right)\right|.
\end{align*}

\noindent \textbf{Approach 1.} We apply the following bound:
$(I) \leq d_{J_{\beta}^*(\altc, \ctp)}(\ct,\cs)$.
By taking the supremum over $x \in \bx$ \revise{of the inequality \eqref{aux:discounted_robustness}} and  rearranging the above terms, we have
$$
\revise{\left\Vert  \tilde{v}-v_{\hat{c},\cs}\right\Vert_{\infty}} \leq \frac{1}{1-\beta} \|c-\hat{c}\|_{\infty} + \frac{\beta}{1-\beta} d_{J_{\beta}^*(\altc, \ctp)}(\ct,\cs).
$$
Then, in view of \eqref{robust_disc_aux}, the second and third claims follow immediately.

\noindent \textbf{Approach 2.}  Note that by the triangle inequality,
\begin{align*}
    (I) \leq  &\left|\int v_{c,\ct}(x')\left(\ct(dx'|x,\gamma_{\altc,\cs,\beta}^{*}(x))-\ctp(dx'|x,\gamma_{\altc,\cs,\beta}^{*}(x))\right)\right|  \\
&+    \biggl|\int \left(v_{\altc,\cs}(x')-v_{c,\ct}(x')\right)\left(\ct(dx'|x,\gamma_{\altc,\cs,\beta}^{*}(x))-\ctp(dx'|x,\gamma_{\altc,\cs,\beta}^{*}(x))\right)\biggr| \\
\leq & d_{\jsc}(\ct,\cs) + 2 \left\Vert v_{\altc,\cs}  -v_{c,\ct}\right\Vert_{\infty}.
\end{align*}
Now, by applying Theorem \ref{upperbound1proof} to the last term above,  taking the supremum over $x \in \bx$ \revise{of the inequality \eqref{aux:discounted_robustness}}, and  rearranging the above terms, we have
\begin{align*}
   \revise{\left\Vert  \tilde{v}-v_{\hat{c},\cs}\right\Vert_{\infty}} \leq    \frac{1+ \beta}{(1-\beta)^{2}}\left\Vert c-\altc\right\Vert_{\infty}+\frac{\beta + \beta^{2}}{(1-\beta)^{2}}d_{\jsc}(\ct,\cs).
\end{align*}
Then, in view of \revise{the first inequality} of \eqref{robust_disc_aux}, the first claim follows.
\end{proof}

\begin{remark}[Comparison and applicability of the three bounds]\label{rmk:applicability}
The upper bound in \eqref{eq:discountedrobust1} only depends on the optimal value function $\jsc$ under the reference MDP, while \eqref{eq:discountedrobust2} and \eqref{eq:discountedrobust3} also involve the optimal value function $J_{\beta}^*(\altc, \ctp)$ under the approximating MDP. However, \eqref{eq:discountedrobust1} exhibits a worse dependence on $\beta$ compared to the others. In particular, to apply the vanishing discounted factor approach to the average-cost setting, as in Theorem \ref{vanishingproof1}, we consider the following quantity and let $\beta \to 1$ along a subsequence:
\[
(1-\beta)\bigl(\jb(c,\ct,\gamma_{\altc,\cs,\beta}^{*}) - \jsc \bigr).
\]
Therefore, in order for the vanishing discounted factor approach to succeed, the upper bound on 
$\left\Vert \jb(c,\ct,\gamma_{\altc,\cs,\beta}^{*}) - \jsc \right\Vert_{\infty}$
must scale as $1/(1-\beta)$ rather than $1/(1-\beta)^2$.

We require three bounds to cover different estimation settings. In
Section \ref{secSampleCo}, we estimate a general continuous MDP via
state-space quantization and an empirical occupancy measure. In this setting, as we will discuss in detail in Section \ref{finiteMDPApp}, \eqref{eq:discountedrobust1} remains valid when the approximate model is a quantized model, even though a quantized model does not satisfy weak continuity imposed by Assumption \ref{basicassump}. In Section \ref{secDistApprx}, we study the sample complexity of disturbance estimation; there, the estimated transition kernel is weakly continuous. \eqref{eq:discountedrobust2} and \eqref{eq:discountedrobust3} are therefore useful: they have better dependence on the discount factor and, as we show later, allow the vanishing-discount approach to be applied to the average-cost case.
  
Comparing \eqref{eq:discountedrobust2} and \eqref{eq:discountedrobust3}, the regularity of $J_{\beta}^{*}(c,\ct)$ is often stronger than that of $J_{\beta}^{*}(\altc,\ctp)$. Consequently, when the reference model is sufficiently regular, the second bound may be tighter. On the other hand, the third bound can accommodate, for example, a Lipschitz approximation of a non-Lipschitz model.

As mentioned earlier in the review section, compared to some (but not all) of the prior literature, an important attribute of the above bounds is that an optimal control policy can be \textit{discontinuous}. The implication is that robustness analysis only requires (weak) continuity of the model and not that of the policies.
\end{remark}

\begin{remark}
If $\jsc$ and $J_{\beta}^*(\altc, \ctp)$ are Lipschitz with respect to the state variable, by Lemma \ref{trivialinequality}, we have
$$
d_{\jsc}(\ct,\cs) \leq \|\jsc\|_{\mathrm{Lip}} d_{\wc_1}(\ct,\ctp), \quad
d_{J_{\beta}^*(\altc, \ctp)}(\ct,\cs) \leq \left\Vert J_{\beta}^*(\altc, \ctp)\right\Vert_{\mathrm{Lip}} d_{\wc_1}(\ct,\ctp).
$$
Using the loose bounds, the third inequality becomes:
\begin{align*}
    \left\Vert\jb(c,\ct,\gamma_{\altc,\cs,\beta}^{*})-\jsc\right\Vert_{\infty} \leq \frac{2}{1-\beta}\left\Vert c-\altc\right\Vert_{\infty} +\frac{2\beta }{1-\beta}  \left\Vert J_{\beta}^*(\altc, \ctp)\right\Vert_{\mathrm{Lip}} d_{\wc_1}(\ct,\ctp).
\end{align*}
This recovers part (2) of Corollary 1 in \cite{bozkurt2024modelapproximationmdpsunbounded}.
 \end{remark}
  
\subsubsection{Average-cost criterion}\label{averagerobustnesssection}

In this subsection, we consider the average-cost criterion and bound the performance loss incurred by applying a policy that is optimal with respect to an \textit{approximate} MDP. 
 Theorem \ref{upperbound3acoe} follows the minorization approach, assuming both kernels satisfy a minorization condition. Theorem \ref{vanishingproof2} builds on the vanishing-discount approach in Theorem \ref{vanishingproof1}, assuming compactness and Lipschitz regularity conditions. 

We denote by $\gamma_{\altc,\cs}^{*}$ the optimal policy learned under an approximate MDP $(\bx,\bu,\ctp,\altc)$ under the \textit{average}-cost criterion in this subsection. 
We also recall that $h_{c,\ct,\epsilon}^{*}$ is the unique fixed-point of $\mathbb{T}_{c,\ct,\epsilon}$ defined in \eqref{equ:minor_operator}. \revise{
We first state the minorization-based robustness bound.}

\begin{theorem}\label{upperbound3acoe}
 Suppose that Assumption \ref{basicassump} holds for two MDPs, (reference) $(\bx,\bu,\ct,c)$ and (approximation) $(\bx,\bu,\ctp,\altc)$. Further, assume $\ct$ (resp.~$\ctp$) satisfies the minorization condition with a probability measure $\rho$ (resp.~$\tau$) and a constant $\epsilon > 0$. Then
 \begin{align*}
\left\Vert J_{\infty} (c,\ct,\gamma_{\altc,\cs}^{*}) - J_{\infty}^{*}(c,\mathcal{T})\right\Vert_{\infty} 
 \leq  \revise{   \frac{2}{\epsilon} \left(\|c - \altc\|_{\infty} \; + \; d_{h_{c,\ct,\epsilon}^{*}}(\ct, \ctp)\right) \; + \;2 \Delta_{c,\ct,\epsilon}(\rho,\tau)},
 \end{align*}
 where we recall that $\Delta_{c,\ct,\epsilon}(\rho,\tau)$ is defined in \eqref{def:Delta_epsilon_rhotau}. In addition,
  \begin{align*}
\left\Vert J_{\infty} (c,\ct,\gamma_{\altc,\cs}^{*}) - J_{\infty}^{*}(c,\mathcal{T})\right\Vert_{\infty} 
 \leq  \revise{   2\|c - \altc\|_{\infty} \; + \; d_{h_{c,\ct,\epsilon}^{*}}(\ct, \ctp)+ \;d_{h_{\altc,\ctp,\epsilon}^{*}}(\ct, \ctp) + \;2 \epsilon\Delta_{c,\ct,\epsilon}(\rho,\tau)}.
 \end{align*}
\end{theorem}

\begin{proof}
By Theorem \ref{minoracoe} and Remark \ref{remark:g_c_T}, under the minorization condition for both MDPs, there exist canonical triplets $(g_{c,\ct}^{*}, h_{c,\ct,\epsilon}^{*}, \gamma^{*}_{c,\ct})$ and 
 $ (g_{\altc,\cs}^{*}, h_{\altc,\cs,\epsilon}^{*}, \gamma^{*}_{\altc,\cs})$ such that for each $x \in \bx$,
 \begin{align}\label{aux_eq:acoe_reference_approx}
 \begin{split}
 &g^*_{c,\ct}
 +
 h^*_{c,\ct,\epsilon}(x)
 =
 c(x,\gamma^*_{c,\ct}(x))
 +
 \int_{\bx} h^*_{c,\ct,\epsilon}(y)
 \ct(dy\mid x,\gamma^*_{c,\ct}(x)),\\
&
 g^*_{\hat c,\ctp}
 +
 h^*_{\hat c,\ctp,\epsilon}(x)
 =
 \hat c(x,\gamma^*_{\hat c,\ctp}(x))
 +
 \int_{\bx} h^*_{\hat c,\ctp,\epsilon}(y)
 \ctp(dy\mid x,\gamma^*_{\hat c,\ctp}(x)),
  \end{split}
\end{align}
and further
 \begin{align}\label{aux:acoe_minor1}
      h_{\altc,\cs,\epsilon}^{*}(x)=\altc(x,\gamma^{*}_{\revise{\altc,\cs}}(x))+\int_{\bx}h_{\altc,\cs,\epsilon}^{*}(x') \left(\ctp(dx'|x,\gamma^{*}_{\revise{\altc,\cs}}(x)) -\epsilon \tau(dx')\right),\;\forall\;x\in\bx.
 \end{align}
In addition, by Remark \ref{remakr:avg_cost_bell}, there exists a triplet 
 $\left(\tilde{g}, \tilde{h},\gamma_{\altc,\cs}^{*}\right)$ such that for each $x \in \bx$,
\begin{align}\label{aux:acoe_minor2}
\begin{split}
&\tilde{h}(x) =c(x,\gamma_{\altc,\cs}^{*}(x))+\int \tilde{h}(x')\left(\mathcal{T}(dx'|x,\gamma_{\altc,\cs}^{*}(x))-\epsilon\rho(dx')\right),\\
&\tilde{g}=\epsilon\int_{\bx}\tilde{h}(x')\rho(dx') =  J_{\infty} (c,\ct,\gamma_{\altc,\cs}^{*})(x).
\end{split}
\end{align}
Since $g_{c,\ct}^{*}$ is the optimal cost,   we have that for any $x \in \bx$,
\begin{align}\label{aux:acoe_v1}
\begin{split}
        0 \leq \tilde{g}-g_{c,\ct}^{*} &= J_{\infty} (c,\ct,\gamma_{\altc,\cs}^{*})(x)- J_{\infty}^{*}(c,\mathcal{T})(x).
\end{split}
\end{align}

 \revise{
\noindent \textbf{Proof for the first bound.}
By the representations for $\tilde{g}$ and $g_{c,\ct}^{*}$, we have
\begin{align*}
    0 \leq \tilde{g}-g_{c,\ct}^{*} = \epsilon\int_{\bx}\left(\tilde{h}(x')-h_{c,\ct,\epsilon}^{*}(x')\right)\rho(dx'),
\end{align*}
which, by the triangle inequality, implies that
\begin{align}\label{aux:average_a1}
 |\tilde{g}-g_{c,\ct}^{*}| \leq \epsilon \|\tilde{h} - h_{\altc,\ctp,\epsilon}^{*}\|_{\infty} +    \epsilon\|h_{\altc,\ctp,\epsilon}^{*}  -  h_{c,\ct,\epsilon}^{*}\|_{\infty}.
\end{align}
}
  
\revise{We start by focusing on the first term in \eqref{aux:average_a1}}. 
By subtracting the  fixed-point equation in \eqref{aux:acoe_minor2} from \eqref{aux:acoe_minor1}, we obtain  that for $x \in \bx$,
\begin{align*}
&|\tilde{h}(x)-h_{\altc,\cs,\epsilon}^{*}(x)|\leq \left\|c -\altc \right\|_{\infty} + (I) + (II) + (III),
\end{align*}
where we define
\begin{align*}
    &(I) := \left| \int \left(\tilde{h}(x')-h_{\altc,\cs,\epsilon}^{*}(x')\right)\left(\ct(dx'|x,\gsp(x))-\epsilon\rho(dx')\right)\right|, \\
    &(II) := \left|\int (h_{\altc,\cs,\epsilon}^{*}(x') - h_{c,\ct,\epsilon}^{*}(x'))\left(\ct(dx'|x,\gsp(x))-\ctp(dx'|x,\gsp(x))+\epsilon\tau(dx')-\epsilon\rho(dx')\right)\right|,\\
    &(III) := \left|\int h_{c,\ct,\epsilon}^{*}(x')\left(\ct(dx'|x,\gsp(x))-\ctp(dx'|x,\gsp(x))+\epsilon\tau(dx')-\epsilon\rho(dx')\right)\right|.
\end{align*}
By definition, we have
\begin{align*}
 &   (I) \leq (1-\epsilon)\left\Vert \tilde{h}-h_{\altc,\ctp,\epsilon}^{*}\right\Vert_{\infty}, \quad (II) \leq (2-2\epsilon) 
     \|h_{c,\ct,\epsilon}^{*} - h_{\altc,\ctp,\epsilon}^{*}\|_{\infty},\\ 
&    (III) \leq  d_{h_{c,\ct,\epsilon}^{*}}(\ct, \ctp) + \epsilon \Delta_{c,\ct,\epsilon}(\rho,\tau).
\end{align*}
Taking the supremum over all $x\in\bx$, and rearranging terms, we have
\begin{align*}
  &\epsilon \left\Vert \tilde{h}-h_{\altc,\ctp,\epsilon}^{*}\right\Vert_{\infty} \leq   \left\|c -\altc \right\|_{\infty}+(2-2\epsilon) 
     \|h_{c,\ct,\epsilon}^{*} - h_{\altc,\ctp,\epsilon}^{*}\|_{\infty}+d_{h_{c,\ct,\epsilon}^{*}}(\ct, \ctp) + \epsilon  \Delta_{c,\ct,\epsilon}(\rho,\tau).
\end{align*}
\revise{
Applying  \eqref{aux:avg_minor_cont} to bound $\|h_{c,\ct,\epsilon}^{*} - h_{\altc,\ctp,\epsilon}^{*}\|_{\infty}$ and due to \eqref{aux:average_a1}, we have
\begin{align*}
    \revise{|\tilde{g}-g_{c,\ct}^{*}| \leq  \frac{2}{\epsilon} \left(\|c - \altc\|_{\infty} \; + \; d_{h_{c,\ct,\epsilon}^{*}}(\ct, \ctp)\right) \; + \;2 \Delta_{c,\ct,\epsilon}(\rho,\tau)},
\end{align*}
which completes the proof of the first bound.
}

\smallskip

\revise{\noindent \textbf{Proof of the second bound.}
Due to \eqref{aux:acoe_v1}, note  the following decomposition: for each $x \in \bx$,
\begin{align}\label{aux:acoe_v2}
0\leq  \tilde{g}-g_{c,\ct}^{*} = J_{\infty} (c,\ct,\gamma_{\altc,\cs}^{*})(x) -
J_{\infty} (\altc,\ctp,\gamma_{\altc,\cs}^{*})(x) +
J_{\infty}^* (\altc,\ctp)(x)
- J_{\infty}^{*}(c,\mathcal{T})(x),
\end{align}
For the second term in \eqref{aux:acoe_v2}, by Theorem \ref{upperboundacoe1},
\begin{align}\label{aux:from_Theorem25}
   \| J_{\infty}^* (\altc,\ctp)
- J_{\infty}^{*}(c,\mathcal{T})\|_{\infty} \leq 2\epsilon \Delta_{c,\ct,\epsilon}(\rho,\tau) +
          \|c - \altc\|_{\infty} \; + \; d_{h_{c,\ct,\epsilon}^{*}}(\ct, \ctp).
\end{align}
It remains to  bound the first term in \eqref{aux:acoe_v2}. Due to the second equation in \eqref{aux_eq:acoe_reference_approx}, 
\begin{align*}
g^*_{\hat c,\ctp}
 +
 h^*_{\hat c,\ctp,\epsilon}(x)
 =&  c(x,\gamma^*_{\hat c,\ctp}(x))
 +
 \int_{\bx} h^*_{\hat c,\ctp,\epsilon}(y)
 \ct(dy\mid x,\gamma^*_{\hat c,\ctp}(x)) \\
 & +
 (\hat c(x,\gamma^*_{\hat c,\ctp}(x))- c(x,\gamma^*_{\hat c,\ctp}(x)))
 +
 \int_{\bx} h^*_{\hat c,\ctp,\epsilon}(y)
 (\ctp-\ct)(dy\mid x,\gamma^*_{\hat c,\ctp}(x)),
\end{align*}
which implies that
\begin{align*}
   g^*_{\hat c,\ctp}
 +
 h^*_{\hat c,\ctp,\epsilon}(x)
 \geq c(x,\gamma^*_{\hat c,\ctp}(x))
 +
 \int_{\bx} h^*_{\hat c,\ctp,\epsilon}(y)
 \ct(dy\mid x,\gamma^*_{\hat c,\ctp}(x)) - \|c-\hat{c}\|_{\infty} - d_{h_{\hat c,\ctp,\epsilon}^{*}}(\ct, \ctp).
\end{align*}
Now we apply the inequality in Lemma \ref{lemma:comparison} with  
$$\gamma = \gamma_{\altc,\cs}^{*},\quad h =h_{\altc,\ctp,\epsilon}^{*},\quad g = g^*_{\hat c,\ctp} + \|c-\hat{c}\|_{\infty} + d_{h_{\altc,\ctp,\epsilon}^{*}}(\ct, \ctp),$$
and obtain for each $x \in \bx$,
\begin{align*}
    g^*_{\hat c,\ctp} + \|c-\hat{c}\|_{\infty} + d_{h_{\altc,\ctp,\epsilon}^{*}}(\ct, \ctp) \geq J_{\infty} (c,\ct,\gamma_{\altc,\cs}^{*})(x).
\end{align*}
As a result, the first term in \eqref{aux:acoe_v2} can be upper bounded as follows:
\begin{align*}
  J_{\infty} (c,\ct,\gamma_{\altc,\cs}^{*})(x) -
J_{\infty} (\altc,\ctp,\gamma_{\altc,\cs}^{*})(x) \leq   \|c-\hat{c}\|_{\infty} + d_{h_{\altc,\ctp,\epsilon}^{*}}(\ct, \ctp).
\end{align*}
Now combining the bounds for the two terms in  \eqref{aux:acoe_v2}, we have that for each $x \in \bx$,
\begin{align*}
    0 \leq   J_{\infty} (c,\ct,\gamma_{\altc,\cs}^{*})(x)- J_{\infty}^{*}(c,\mathcal{T})(x) 
    \leq 2\epsilon \Delta_{c,\ct,\epsilon}(\rho,\tau) +
          2\|c - \altc\|_{\infty}  + d_{h_{c,\ct,\epsilon}^{*}}(\ct, \ctp) + 
          d_{h_{\altc,\ctp,\epsilon}^{*}}(\ct, \ctp),
\end{align*}
which completes the proof.
}
\end{proof}

\revise{
\begin{remark}
Note that the first bound in Theorem \ref{upperbound3acoe} only involves the relative value function of the reference MDP $(c,\ct)$, while the second bound additionally involves that of the approximating MDP $(\altc,\ctp)$. The advantage of the second bound is that the cost and kernel discrepancy terms no longer carry the
explicit $1/\epsilon$ factor.
\end{remark}
}

\revise{Next, we give the corresponding robustness bound under the vanishing-discount approach. Unlike Theorem \ref{upperbound3acoe}, this result does not require a minorization condition, but instead relies on the compactness and Lipschitz regularity assumptions used in Theorem \ref{vanishingproof1}.
}

\revise{
\begin{theorem}\label{vanishingproof2}
Suppose Assumption \ref{basicassump} and   Assumption \ref{equicontinuousassump}(a)-(b) hold  with the same constants for two MDPs, (reference) $(\bx,\bu,\ct,c)$ and (approximation) $(\bx,\bu,\ctp,\altc)$. Then
\begin{equation*}
\left\Vert J_{\infty}(c,\ct,\gamma_{\altc,\cs}^{*}) - J_{\infty}^{*}(c,\mathcal{T})\right\Vert_{\infty} \leq 2\|c-\altc\|_{\infty}+ 2L d_{\wc_{1}}(\ct,\ctp),
\end{equation*}
where $L$ is the constant in Assumption \ref{equicontinuousassump}. If additionally Assumption \ref{equicontinuousassump}(c) holds for both MDPs, then 
\begin{equation*}
\left\Vert J_{\infty}(c,\ct,\gamma_{\altc,\cs}^{*}) - J_{\infty}^{*}(c,\mathcal{T})\right\Vert_{\infty} \leq 2\|c-\altc\|_{\infty}+ \left(d_{h_{c,\ct}^{*}}(\ct, \ctp) 
+  d_{h_{\altc, \ctp}^{*}}(\ct, \ctp)  
\right),
\end{equation*}
where $h_{c,\ct}^{*}$ and $h_{\altc, \ctp}^{*}$ are defined as in Theorem \ref{thrm:vanish_main}.
\end{theorem}
}
\revise{
\begin{proof}
By Theorem \ref{thrm:vanish_main} and the ACOE, we have the following identity: for $x \in \bx$,
\begin{align*}
&\quad \;c(x,\gamma_{\hat{c},\ctp}^{*}(x))+\int_{\bx}h_{\hat{c},\ctp}^{*}(y)\ct(dy|x,\gamma_{\hat{c},\ctp}^{*}(x))\\&=\hat{c}(x,\gamma_{\hat{c},\ctp}^{*}(x))+\int_{\bx}h_{\hat{c},\ctp}^{*}(y)\ctp(dy|x,\gamma_{\hat{c},\ctp}^{*}(x))\\
&\quad +(c(x,\gamma_{\hat{c},\ctp}^{*}(x))-\hat{c}(x,\gamma_{\hat{c},\ctp}^{*}(x)))+\int_{\bx}h_{\hat{c},\ctp}^{*}(y)(\ct-\ctp)(dy|x,\gamma_{\hat{c},\ctp}^{*}(x))\\
&= g_{\hat{c},\ctp}^{*}+ h_{\hat{c},\ctp}^{*}(x)+(c(x,\gamma_{\hat{c},\ctp}^{*}(x))-\hat{c}(x,\gamma_{\hat{c},\ctp}^{*}(x)))+\int_{\bx}h_{\hat{c},\ctp}^{*}(y)(\ct-\ctp)(dy|x,\gamma_{\hat{c},\ctp}^{*}(x)).
\end{align*}
Then by triangle inequality, we have that for $x \in \bx$,
\begin{align*}
   &c(x,\gamma_{\hat{c},\ctp}^{*}(x))+\int_{\bx}h_{\hat{c},\ctp}^{*}(y)\ct(dy|x,\gamma_{\hat{c},\ctp}^{*}(x)) \\
   \leq  & g_{c,\ct}^{*} + h_{\hat{c},\ctp}^{*}(x)+ |g_{\hat{c},\ctp}^{*}-g_{c,\ct}^{*}| + \left\Vert c-\hat{c}\right\Vert_{\infty}+d_{h_{\hat{c},\ctp}^{*}}(\ct,\ctp).
\end{align*}
If we define
$g_u:=g_{c,\ct}^{*} + |g_{\hat{c},\ctp}^{*}-g_{c,\ct}^{*}| + \left\Vert c-\hat{c}\right\Vert_{\infty}+d_{h_{\hat{c},\ctp}^{*}}(\ct,\ctp)$ and $\tilde{h}(x) := h_{\hat{c},\ctp}^{*}(x)$,
then we have that for $x \in \bx$,
\begin{align*}
    g_u+\tilde{h}(x)\geq c(x,\gamma_{\hat{c},\ctp}^{*}(x))+\int_{\bx}\tilde{h}(y)\ct(dy|x,\gamma_{\hat{c},\ctp}^{*}(x)).
\end{align*}
Furthermore, by Theorem \ref{thrm:vanish_main}, $\tilde{h} := h_{\hat{c},\ctp}^{*}$ is Lipschitz on a compact domain, which implies $\left\Vert \tilde{h}\right\Vert_{\infty}<\infty$. We can therefore invoke Lemma \ref{lemma:comparison} and obtain  for $x \in \bx$,
$$
g_u=g_{c,\ct}^{*}+| g_{\hat{c},\ctp}^{*}-g_{c,\ct}^{*}|+\left\Vert c-\hat{c}\right\Vert_{\infty}+d_{h_{\hat{c},\ctp}^{*}}(\ct,\ctp)\geq J_{\infty}(c,\ct,\gamma_{\hat{c},\ctp}^{*})(x),
$$
which implies that for $x \in \bx$,
\begin{align*}
 J_{\infty}(c,\ct,\gamma_{\altc,\cs}^{*})(x) - J_{\infty}^{*}(c,\mathcal{T})(x) \leq | g_{\hat{c},\ctp}^{*}-g_{c,\ct}^{*}|+\left\Vert c-\hat{c}\right\Vert_{\infty}+d_{h_{\hat{c},\ctp}^{*}}(\ct,\ctp).
\end{align*}

By the definition of optimal value, we have that for $x \in \bx$,
$J_{\infty}(c,\ct,\gamma_{\altc,\cs}^{*})(x) - J_{\infty}^{*}(c,\mathcal{T})(x) \geq 0$.
As a result,
$$
\left\Vert J_{\infty}(c,\ct,\gamma_{\altc,\cs}^{*}) - J_{\infty}^{*}(c,\mathcal{T})\right\Vert_{\infty} \leq | g_{\hat{c},\ctp}^{*}-g_{c,\ct}^{*}|+\left\Vert c-\hat{c}\right\Vert_{\infty}+d_{h_{\hat{c},\ctp}^{*}}(\ct,\ctp).
$$

 By the first upper bound in Theorem \ref{vanishingproof1},
$| g_{\hat{c},\ctp}^{*}-g_{c,\ct}^{*}| \leq \Vert c-\altc\Vert_{\infty}+ L d_{\wc_{1}}(\ct,\ctp)$. By Theorem \ref{thrm:vanish_main}, $h_{\hat{c},\ctp}^{*}$ is $L$-Lipschitz, which, due to Lemma \ref{trivialinequality}, implies that $d_{h_{\hat{c},\ctp}^{*}}(\ct,\ctp) \leq Ld_{\wc_{1}}(\ct,\ctp)$. Thus, we have
$$
\left\Vert J_{\infty}(c,\ct,\gamma_{\altc,\cs}^{*}) - J_{\infty}^{*}(c,\mathcal{T})\right\Vert_{\infty} \leq 2\left\Vert c-\altc\right\Vert_{\infty}+ 2L d_{\wc_{1}}(\ct,\ctp),
$$
which completes the proof of the first claim. Finally, the second claim follows from the second upper bound of $| g_{\hat{c},\ctp}^{*}-g_{c,\ct}^{*}|$ in Theorem \ref{vanishingproof1}.
\end{proof}
}

\subsection{Lipschitz Regularity of optimal value functions}\label{Lipschitzmdp}

When applying the previous results, it is often necessary to show that the optimal (relative) value functions of the reference MDP $(\bx,\bu,\ct,c)$, namely $J_{\beta}^{*}(c,\ct)$ and $h_{c,\ct,\epsilon}^{*}$ defined in \eqref{equ:minor_operator}, are Lipschitz continuous and to bound their Lipschitz constants. In this section, we present sufficient conditions for this property. To our knowledge, these results were first established in \cite{hinderer2005}, with a main result also appearing in \cite{SaLiYuSpringer}; for the average-cost case, see \cite{demirci2023average}.

\begin{lemma}[Theorem 4.37 of \cite{SaLiYuSpringer}]\label{lipbellman}
For an MDP $(\bx,\bu,\ct,c)$ and a discount factor $\beta < 1$ satisfying Assumption \ref{myassump}, $J_{\beta}^{*}(c,\ct): \bx \to \mathbb{R}$ is $\frac{\clip}{1-\beta\tlip}$-Lipschitz with respect to the state variable.
\end{lemma}

By combining Lemma \ref{lipbellman} with Theorem \ref{upperbound1proof} and Theorem \ref{thrm:discounted_robust}, we immediately have the following corollary for the $\beta$-discounted cost case.

\begin{corollary}\label{upperbound3cor}
Consider two MDPs, (reference) $(\bx,\bu,\ct,c)$ and (approximation) $(\bx,\bu,\ctp,\altc)$. Suppose that the reference MDP satisfies Assumption \ref{myassump}, and that the approximate MDP satisfies Assumption \ref{basicassump}. Then
\begin{equation*}
\begin{aligned}
&\left\Vert\jsc-\jscp\right\Vert_{\infty}\leq \frac{1}{1-\beta}\left\Vert c-\altc\right\Vert_{\infty}+\frac{\beta\clip}{(1-\beta)(1-\beta\tlip)}d_{\wc_{1}}(\ct,\ctp),\\
&\left\Vert\jb(c,\ct,\gamma_{\altc,\cs,\beta}^{*})-\jsc\right\Vert_{\infty} \leq  \frac{2}{(1-\beta)^{2}}\left\Vert c-\altc\right\Vert_{\infty}+\frac{2\beta\clip}{(1-\beta)^{2}(1-\beta\tlip)}d_{\wc_{1}}(\ct,\ctp).
\end{aligned}
\end{equation*}
If, in addition, the approximate MDP satisfies Assumption \ref{myassump}, then
     \begin{equation*}
\left\Vert\jb(c,\ct,\gamma_{\altc,\cs,\beta}^{*})-\jsc\right\Vert_{\infty} \leq  \frac{2}{1-\beta}\left\Vert c-\altc\right\Vert_{\infty}+\frac{\beta}{1-\beta}\left(\frac{\clip}{1-\beta\Vert \ct\Vert_{\mathrm{Lip}}}+\frac{\Vert\altc\Vert_{\mathrm{Lip}(\bx)}}{1-\beta\Vert\ctp\Vert_{\mathrm{Lip}}}\right)d_{\wc_{1}}(\ct,\ctp).
\end{equation*}
\end{corollary}

We now present results for the average-cost criterion case.

\begin{assumption}\label{acoelipassump}
In addition to Assumption \ref{myassump}, assume $\tlip<1$.
\end{assumption}

The following is an analog of Lemma \ref{lipbellman} for the average-cost criterion; see  \cite[Lemma 2.4]{demirci2023average} or \cite[Theorem 3.5]{ky2023qaverage}. Notably, the Lipschitz constant upper bound is independent of the scaling constant $\epsilon$ of the minorant.

\begin{lemma}\label{acoelip}
Consider an MDP $(\bx,\bu,\ct,c)$ satisfying Assumption \ref{acoelipassump}. Assume that $\ct$ satisfies the minorization condition with a probability measure $\rho$ and a constant $\epsilon > 0$. Then $h_{c,\ct,\epsilon}^{*}:\bx \to \mathbb{R}$ as defined in \eqref{equ:minor_operator} is $\frac{\clip}{1-\tlip}$-Lipschitz  with respect to the state variable.
\end{lemma}

By combining Lemma \ref{acoelip} with Theorem \ref{upperboundacoe1} and Theorem \ref{upperbound3acoe}, we immediately have the following corollaries for the average-cost case under minorization conditions.

\begin{corollary}\label{cor:average_cost_final}
Consider two MDPs, (reference) $(\bx,\bu,\ct,c)$ and (approximation) $(\bx,\bu,\ctp,\altc)$. Suppose that the reference satisfies Assumption \ref{acoelipassump}, and the approximation satisfies Assumption \ref{basicassump}. Further, assume $\ct$ (resp.~$\ctp$) satisfies the minorization condition with a probability measure $\rho$ (resp.~$\tau$) and a constant $\epsilon > 0$. Then
\begin{align*}
 &       \left\Vert J_{\infty}^{*}(c,\mathcal{T})-J_{\infty}^{*}(\altc,\cs)\right\Vert_{\infty} \leq \left\Vert c-\altc\right\Vert_{\infty} \;+\;\;  \frac{\clip}{1-\tlip}  d_{\wc_{1}}(\ct,\ctp). \\
&\left\Vert J_{\infty}(c,\ct,\gamma_{\altc,\cs}^{*}) - J_{\infty}^{*}(c,\mathcal{T})\right\Vert_{\infty} 
 \leq \revise{\frac{2}{\epsilon}} \left(\|c - \altc\|_{\infty} \; + \;  \frac{\clip}{1-\tlip}  d_{\wc_{1}}(\ct,\ctp) \right)+  \frac{ \revise{2} \clip}{1-\tlip}  d_{\wc_{1}}(\rho, \tau) .
\end{align*}
\end{corollary}

Now, we consider the vanishing discounted factor approach under the average-cost criterion. By Lemma \ref{lipbellman} and due to Assumption \ref{acoelipassump}, the Assumption \ref{equicontinuousassump}(b) holds with $L = \frac{\clip}{1-\tlip}$. The same discussion applies 
to the approximate kernel $\ctp$.
Then, due to  Theorem \ref{vanishingproof1}  and Theorem \ref{vanishingproof2}, we have the following corollary.

\begin{corollary}\label{drivingnoisecorollary2}
Suppose Assumption \ref{acoelipassump} holds for two MDPs, (reference) $(\bx,\bu,\ct,c)$ and (approximation) $(\bx,\bu,\ctp,\altc)$. 
In addition, assume $\bx$ is compact. Then
\begin{align*}
&\left\Vert J_{\infty}^{*}(c,\mathcal{T})-J_{\infty}^{*}(\altc,\cs)\right\Vert_{\infty}\leq \left\Vert c-\altc\right\Vert_{\infty}+ \frac{\clip}{1-\tlip} d_{\wc_{1}}(\ct,\ctp); \\
&\revise{\left\Vert J_{\infty}(c,\ct,\gamma_{\altc,\cs}^{*}) - J_{\infty}^{*}(c,\mathcal{T})\right\Vert_{\infty} \leq 2\|c-\altc\|_{\infty} +\left(\frac{\clip}{1-\tlip} +\frac{\Vert\altc\Vert_{\mathrm{Lip}(\bx)}}{1-\Vert\ctp\Vert_{\mathrm{Lip}}}\right)d_{\wc_{1}}(\ct,\ctp).}
\end{align*}
\end{corollary}

\section{Application to Model Learning From Data and Sample Complexity}\label{secSampleCo}

In this section, we introduce a general learning framework in which the model itself is learned from data. We begin by reviewing model approximation with finite model representations, noting that some cases fall within the scope of our general robustness results. We then propose learning algorithms to estimate the quantized model from data, establishing sample complexity bounds. The generality of these results appears to be novel in the literature.

\subsection{Quantized Approximations}\label{finiteMDPApp}

In this subsection, we briefly review the state quantization scheme introduced, to the best of our knowledge, in \cite[Section 3]{saldi2017} (see also \cite[Section 2.3]{KSYContQLearning}). These results will also be useful when developing two general learning algorithms in the following subsections. We demonstrate that the near-optimality of such quantization,  as shown for the discounted-cost problem in Theorem 6 of \cite{KSYContQLearning} (which slightly refines \cite[Theorem 4.38]{SaLiYuSpringer}) and for the average-cost problem in Theorem 3.5 of \cite{ky2023qaverage}, can be considered as special cases of Corollary \ref{upperbound3cor}  and Corollary  \ref{cor:average_cost_final}, respectively. For simplicity, we assume that $\bu$ is a finite set. Extensions to compact action spaces can be
obtained via action-space quantization under additional regularity conditions; see Chapter 3 of \cite{SaLiYuSpringer}.

Consider a state space $\bx$ and an $M$-partition of it, $\{B_{i}\}_{i=1}^{M}$. For each partition $B_{i}$, we pick some representative element $y_{i}\in B_{i}$. A quantizer is then a map $q:\;\bx\to \{y_{i}\}_{i=1}^{M}$, where 
$$q(x)=y_{i} \quad \text{ if and only if } \quad 
x\in B_{i}.
$$

Given a weighting measure $\pi\in\mathcal{P}(\bx)$,  assuming $\pi(B_i) > 0$ for $i \in [M] = \{1,\dots,M\}$, 
we define a new finite state MDP, $(\{y_{i}\}_{i=1}^{M},\bu,\ct^{M,\pi},c^{\pi,M})$, where $c^{\pi,M}$ and $\ct^{M,\pi}$ are defined as follows: for $i,j \in [M]$ and $u \in \bu$,
\begin{align}\label{finitemodeldef}
c^{\pi,M}(y_{i},u):=   \frac{\int_{B_i} c(x,u) \pi(dx) }{    \int_{B_i} \pi(dx)}, \quad
\ct^{\pi,M}(y_{j}|y_{i},u):=\frac{\int_{B_{i}}\ct(B_{j}|x,u){\pi} (dx)}{ \int_{B_i} \pi(dx)}.
\end{align}
We denote by $\gamma^{\pi, M}_{\beta}$ and $\gamma^{\pi, M}_{\infty}$
the optimal policy for the $\beta$-discounted  and the average-cost MDP $(\{y_{i}\}_{i=1}^{M},\bu, {\ct}^{\pi,M},{c}^{\pi,M})$, respectively. Next, we define the piecewise constant extension of such a finite model to the original state space as in \cite[Section 3]{saldi2017}.

\begin{definition}[piecewise constant extension I]\label{def:piece-wise-extension-I}
For $1 \leq i \leq M$,  $x \in B_i$, $u \in \bu$,
and $A \in  \cF_{\bx}$, we define
$$
\overline{c^{\pi,M}}(x,u) := c^{\pi,M}(y_{i},u), \quad
\overline{\ct^{\pi,M}}(A|x,u) := \int_{B_i} \ct(A|z,u) {\pi}(dz)/\pi(B_i),
$$
and further
$\overline{\gamma^{\pi, M}_{\beta}}(x) := \gamma^{\pi, M}_{\beta}(y_i), \quad
\overline{\gamma^{\pi, M}_{\infty}}(x) := \gamma^{\pi, M}_{\infty}(y_i)$.
\end{definition}

As discussed in \cite[Section 3]{saldi2017}, $\overline{\gamma^{\pi, M}_{\beta}}$ and $\overline{\gamma^{\pi, M}_{\infty}}$ are 
the optimal policies for the $\beta$-discounted  and the average-cost MDP $(\bx,\bu,\overline{\ct^{\pi,M}},\overline{c^{\pi,M}})$, respectively.

We note that the transition kernel under piecewise constant extension is, in general, not weakly continuous, so Assumption \ref{basicassump} does not apply to the approximate model. As a remedy, for the discounted problem, we adopt Assumption 2.1 of \cite{HernndezLerma1989}, which automatically holds for the piecewise constant model under Assumption \ref{basicassump} with a finite action set $\bu$. For the average-cost problem, in addition to Assumption 2.1 of \cite{HernndezLerma1989}, Assumption  3.1(4) of \cite{HernndezLerma1989} also holds if the original kernel $\ct$ is minorized as in Definition \ref{minordef}. 
Therefore, we can invoke optimality equations for the piecewise constant model for both problems, which means that all our results hold when using the piecewise constant model as an approximation. We further define the  quantization error:
\begin{align}\label{diameter_def}
    \delta_M := \max_{1 \leq i \leq M} \sup_{x,x' \in B_i} d_{\bx}(x,x').
\end{align}
The following approximation error is noted in \cite{ky2023qaverage}:
\begin{lemma}[Lemma 3.3, \cite{ky2023qaverage}]\label{piecewiseconstantextensionerror}
Under Assumption \ref{myassump}(a) and (b), 
\begin{align*}
\left\Vert c-\overline{c^{\pi,M}}\right\Vert_{\infty}\leq \clip \delta_{M},\quad
d_{\wc_{1}}(\ct,\overline{\ct^{\pi,M}})\leq \tlip \delta_{M}.
\end{align*}
\end{lemma}

 Combining Corollary \ref{upperbound3cor} and Lemma \ref{piecewiseconstantextensionerror}, we immediately recover the following result for the $\beta$-discounted cost case,  as stated in Theorem 6 of \cite{KSYContQLearning}.

\begin{corollary}[Theorem 6, \cite{KSYContQLearning}]
Under Assumption   \ref{myassump}, 
\begin{equation*}
\left\Vert\jb(c,\ct,\overline{\gamma^{\pi,M}_{\beta}})-\jsc\right\Vert_{\infty}\leq \frac{2\clip}{(1-\beta)^{2}(1-\beta\tlip)}\delta_{M}.
\end{equation*}
\end{corollary}

Further, combining Corollary \ref{cor:average_cost_final} and Lemma \ref{piecewiseconstantextensionerror}, we also recover the following result for the average-cost case,  as stated in Theorem 3.5 of \cite{ky2023qaverage}.

\begin{corollary}[Theorem 3.5, \cite{ky2023qaverage}]
Suppose that $\ct$ satisfies the minorization condition with a probability measure $\rho$ and a constant $\epsilon > 0$  as in Definition \ref{minordef}. 
Under Assumption   \ref{acoelipassump},
\begin{equation*}
\left\Vert J_{\infty}(c,\ct,\overline{\gamma^{\pi,M}_{\infty}}) -J_{\infty}^{*}(c,\mathcal{T})\right\Vert_{\infty}\leq \left(1+\frac{2}{\epsilon}\right)\left(\frac{\clip}{1-\tlip}\right)\delta_{M}.
\end{equation*}
\end{corollary}
\begin{proof}
Due to the construction in Definition \ref{def:piece-wise-extension-I},   $\overline{\ct^{\pi,M}}$ also satisfies the minorization condition with the probability measure $\rho$ and constant $\epsilon$. Consequently, the proof is complete by Corollary \ref{cor:average_cost_final}.
\end{proof}

Note that in dealing with the approximate MDP $(\bx,\bu,\overline{\ct^{\pi,M}},\overline{c^{\pi,M}})$ above, we work with the $\wc_1$-distance between $\ct$ and $\overline{\ct^{\pi,M}}$, and use Corollary \ref{upperbound3cor}  and Corollary  \ref{cor:average_cost_final}. As shown in Section \ref{secDistApprx}, this approach is not tight. On the other hand,  in Lemma \ref{piecewiseconstantextensionerror}, the estimation errors, $\Vert c-\overline{c^{\pi,M}}\Vert_{\infty}$ and $ 
d_{\wc_{1}}(\ct,\overline{\ct^{\pi,M}})$, are of the same order, so there is no benefit in applying tighter results in Theorem \ref{thrm:discounted_robust} and Theorem \ref{upperbound3acoe}. However, as we shall see in the next subsection, those tighter results will be important in finite model learning.

In the next subsections, we adopt a different extension for the kernel for the purpose of technical analysis. Denote by $\mathbb{Q}_{M} := \{y_i\}_{i=1}^{M}$.

\begin{definition}[piecewise constant extension II]\label{def:piece-wise-extension-II-kernel}
Let  $\cs: \cF_{\mathbb{Q}_{M}} \times \mathbb{Q}_{M} \times \bu \to [0,1]$ be a controlled transition kernel. Define its piecewise constant extension $\widetilde{\cs}:\cF_{\bx} \times \bx \times \bu \to [0,1]$ as follows: for $i =1,\ldots,M$,  $x \in B_i$, $u \in \bu$, and $A \in  \cF_{\bx}$,
$\widetilde{\cs}(A|x,u) := \sum_{j = 1}^{M}
\cs(y_j|y_i,u) \mathbbm{1}\{y_j \in A\}$.
\end{definition}

We can bound the approximation error in a manner similar to Lemma \ref{piecewiseconstantextensionerror}. Since the proof closely parallels that of Lemma \ref{piecewiseconstantextensionerror}, we omit it here.

 \begin{lemma}\label{anotherextension_error}
Under Assumption \ref{myassump}(b), 
    $d_{\wc_1}(\widetilde{\ct^{\pi,M}}, \ct) \leq  (1+\tlip)\delta_M.
    $
\end{lemma}

\subsection{Simultaneous Finite Model Approximation and  Learning}\label{subsec:finitemodellearning}
In this section, we focus on the case where the quantized model in \eqref{finitemodeldef} is unknown and needs to be learned from data. We consider two scenarios regarding data availability: (i) a single trajectory of a controlled Markov process and (ii) independent transitions where a simulation device is available for each initial state and action. 

For either case, our goal is to show the robustness of the optimal policies derived from \textit{simultaneous} finite model approximation and the learning of these approximate finite models obtained via empirical estimates.  We provide explicit sample complexity bounds that relate performance loss to the number of samples. 

 As noted earlier, a related approach is presented in \cite{Dufour2015-qu}, which involves a different construction for scenarios where the model is known and satisfies an absolute continuity condition with respect to a reference measure. This construction aims to obtain a finite model by utilizing empirical data collected from the reference measure. Together with the known density, the empirical measure is used to construct a finite model.

\subsubsection{Empirical model learning from a single trajectory}\label{markovempiricalmodel}
\revise{Let $\phi\in \mathcal{P}(\bu)$ be a state-independent exploration policy such that $\phi(u)>0$ for every $u\in \bu$. Under this policy, actions are sampled independently from $\phi$, and the state process evolves according to the Markov kernel $\mathcal T_\phi(A\mid x):=\sum_{u\in\bu}\phi(u)\mathcal T(A\mid x,u)$ for each $A\in\mathcal F_{\bx}$ and $x \in \bx$.}

We observe a trajectory $\{(X_n, U_n, C_n,X_{n+1}), 0 \leq n \leq N-1\}$, whose dynamics are determined by the kernel $\ct$ and exploration policy $\phi$. Specifically, for each $x \in \bx$, $u,u' \in \bu$, and $A \in \mathcal{F}_{\bx}$,
\begin{align}\label{def:collection_MC_dyn}
    \revise{\mathbb{P}(X_{n+1} \in A,\;\; U_{n+1} = u' \vert X_{n} = x, U_n = u) = \phi(u') \int_{A} \ct(dx'|x,u),}\quad
     C_n = c(X_n, U_n).
\end{align}
Define the following estimators for $i,j \in [M]$ and $u \in \bu$,
\begin{align*}
&\hat{\ct}_{N}(y_j|y_i, u) := \frac{\sum_{n=0}^{N-1} \mathbbm{1}\{X_{n+1} \in B_j, X_n \in B_i, U_n=u\}}{\sum_{n=0}^{N-1} \mathbbm{1}\{X_n \in B_i, U_n = u\}}, \;\;
 \hat{c}_N(y_i,u) := \frac{\sum_{n=0}^{N-1} C_n \mathbbm{1}\{X_n \in B_i, U_n=u\}}{\sum_{n=0}^{N-1} \mathbbm{1}\{X_n \in B_i, U_n = u\}},
\end{align*}
where if the denominator is zero, we use the convention that $\hat{c}_N(y_i,u) = 0$, and that $\hat{\ct}_{N}(y_j|y_i, u)$ equals $1$ if $j = i$ and $0$ if $j \neq i$. See also Algorithm \ref{alg:sampling_singletrajectory}.

Given the learned MDP $(\hat{\ct}_{N},\hat{c}_N)$, we can solve the finite state space MDP for the optimal $\beta$-discounted cost policy $\hat{\gamma}_{N,\beta}$ and the optimal average-cost policy $\hat{\gamma}_{N,\infty}$. We then extend them to the original state space as in Definition \ref{def:piece-wise-extension-I}, and denote the extensions by $\overline{\hat{\gamma}_{N,\beta}}$ and  $\overline{\hat{\gamma}_{N,\infty}}$. Next, we show that the extended policies are robust under the true dynamic $(c,\ct)$.

\begin{algorithm}[t!]
\caption{Sampling a Wasserstein Regular MDP along a Single Trajectory
}
\label{alg:sampling_singletrajectory}
\begin{algorithmic}[1]
\Require{Simulation Length $N$ and exploration policy $\phi$}
\ForAll{$n = 0,\ldots, N-1$}
\State{Observe $(X_{n},U_{n},C_{n},X_{n+1})$ that evolves according to the true kernel $\ct$ and policy $\phi$}
\EndFor
\State{Set $\hat{\ct}_{N}(y_j|y_i, u) := \frac{\sum_{n=0}^{N-1} \mathbbm{1}\{X_{n+1} \in B_j, X_n \in B_i, U_n=u\}}{\sum_{n=0}^{N-1} \mathbbm{1}\{X_n \in B_i, U_n = u\}}$}
\State{Set $\hat{c}_N(y_i,u) := \frac{\sum_{n=0}^{N-1} C_n \mathbbm{1}\{X_n \in B_i, U_n=u\}}{\sum_{n=0}^{N-1} \mathbbm{1}\{X_n \in B_i, U_n = u\}}$}
\State \Return $(\hat{\ct}_{N} , \hat{c}_{N})$
\end{algorithmic}
\end{algorithm}

Consider the augmented Markov chain $\{Z_n :=(X_n, U_n, X_{n+1}): n \geq 0\}$ with the state space $\mathbb{Z} := \bx \times \bu \times \bx$. \revise{We denote its transition kernel over $\mathbb{Z}$ by $\mathcal{P}^{\ct,\phi}$, which does not depend on the partitions $\{B_i: i\in [M]\}$.}

\begin{assumption}\label{Model_Learning_Markov}
\textbf{(a).} Assume that  $\{Z_{n}: n \geq 0\}$ is $\psi$-irreducible with a unique invariant distribution $\tilde{\pi}$ on $\mathbb{Z}$. Then necessarily
$\tilde{\pi}(A,u,A') = \int_{A}\left(\int_{A'} T(dx'|x,u)\right)\revise{\phi(u)} \pi(dx)$ for $A,A' \in \cF_{\bx}$ and $u \in \bu$, \revise{where 
$\pi$ denotes the first marginal of $\tilde{\pi}$}. Assume that $\pi(B_i) >0$ for $1 \leq i \leq M$.

\textbf{(b).} There exist constants $c_0, C_0 > 0$ such that for any measurable function $f:\bx \times \bu \times \bx \to [0,1]$, $\epsilon > 0$ and $n \geq 1$, we have 
\begin{align*}
    \Pro\left(\left|n^{-1}\sum_{t=0}^{n-1} f(X_{t}, U_{t}, X_{t+1}) - \int_{\bx}\sum_{u \in \bu} \int_{\bx}f(x,u,x') \ct(dx'|x,u) \revise{\phi(u)} \pi(dx) \right|\geq \epsilon \right)  \leq C_0 \exp\left(-c_0 n \epsilon^2\right).
\end{align*}
\end{assumption}

\begin{remark}
Assumption \ref{Model_Learning_Markov}(b) has been extensively studied in the literature, primarily through three closely interrelated approaches: (i) Spectral methods for both reversible and non-reversible Markov chains \cite{kontoyiannis2003spectral,miasojedow2014hoeffding}, (ii) Concentration inequalities and martingale difference methods \cite[Chapter 23]{douc2018markov} and in particular \cite[Theorem 23.3.1 and Corollary 23.2.4]{douc2018markov}, and (iii) coupling based methods via stochastic drift conditions (see e.g. \cite[Theorem 12]{Rosenthal} via drift conditions \cite{MeynBook} to a {\it small} set). 
\end{remark}

\begin{example} As an explicit example, Assumption \ref{Model_Learning_Markov}(b) holds if the Markov chain $\{Z_{n}: n \geq 0\}$ has a spectral gap and the distribution of $X_0$ is not too far away from $\pi$. Specifically, denote by $\mathcal{L}_2^{0}(\tilde{\pi}) := \{f: \mathbb{Z} \to \bR: \tilde{\pi}(f^2) < \infty, \tilde{\pi} (f) = 0\}$ the space of zero mean, square integrable functions on $\mathbb{Z}$ with respect to $\tilde{\pi}$, where $\tilde{\pi} (f) := \int_{\mathbb{Z}} f(z) \tilde{\pi}(z)$. Let $\lambda\in [0,1]$ be the operator norm of \revise{$\mathcal{P}^{\ct,\phi}$}  acting on  $\mathcal{L}_2^{0}(\tilde{\pi})$. Further,  denote by $\nu$ the distribution of $X_0$. Assume that $\lambda < 1$,  that
$\nu$ is absolutely continuous with respect to $\pi$ with  the Radon–Nikodym derivative $d\nu/d \pi$, and that
$C_0' := \left(\int_{\bx} \left(\frac{d\nu}{d\pi}(x)\right)^2 \pi(dx) \right)^{1/2} < \infty$. Then by Corollary 3.11 in \cite{miasojedow2014hoeffding}, Assumption \ref{Model_Learning_Markov}(b) holds with $c_0 = (1-\lambda)/(1+\lambda), C_0= C_0'$.
\end{example}

Note that the above assumption is only imposed on the Markov chain $\{Z_n: n \geq 0\}$, and not on the finite approximation. Further, we do not need to know $\pi$ above (as the model is unknown); this is \textit{just used for analysis purposes}. 

Recall the definitions of $c^{\pi,M}$ and $\ct^{\pi,M}$ in \eqref{finitemodeldef}. Define 
\begin{align}
    \label{def:kappa_pi_M_init}
    \kappa_{\pi,M} := \min\left\{\revise{\phi(u)} \times \int_{B_i}   \pi(dx): i \in [M], u \in \bu\right\}, \;\;
    \kappa_{\ct,M} := \min\left\{\ct^{\pi,M}(y_j|y_i,u): (i,j,u) \in \mathcal{I}_{M}\right\},
\end{align}
where $\mathcal{I}_M := \{
i,j \in [M], u \in \bu: \ct^{\pi,M}(y_j|y_i,u) > 0
\}$. In particular, $\kappa_{\ct,M}$ is the smallest non-zero element in the matrix $\ct^{\pi,M}$. Further, define the event
\begin{align}
    \label{def:E_N_M}
\mathcal{E}_{N, M} := \bigcup_{1 \leq i,j \leq M, u \in \bu}
\left\{\hat{T}_N(y_j|y_i,u) <  \frac{1}{2}\ct^{\pi,M}(y_j|y_i,u) \right\}.
\end{align}

 Now, we bound the difference between the estimated $(\hat{\ct}_{N}, \hat{c}_{N})$ and their limits $(c^{\pi,M}, \ct^{\pi,M})$.  
\begin{lemma}\label{lemma:Markov_con}
 Suppose that $c$ is non-negative with $\|c\|_{\infty} < \infty$, and that Assumption \ref{Model_Learning_Markov} holds.  Then, there exists an absolute constant $C > 0$ such that
    \begin{align*}
        \Exp\left[\left\Vert 
        \hat{c}_{N} - c^{\pi,M}
        \right\Vert_{\infty} \right] \leq C K_{1} \|c\|_{\infty}  \frac{1}{\kappa_{\pi,M}}\sqrt{\frac{\log(M |\bu|)}{N}} + 2C_0 \|c\|_{\infty} M |\bu| \exp\left(-K_{2} \kappa_{\pi,M}^2 N\right),
    \end{align*}
where we denote $K_1 :=  \sqrt{ (1+ 2C_0)/c_0}$ and $K_2 := c_0/4$. Further, for any   bounded function $g: \{y_j: j \in [M]\}  \to [-L,L]$ with $L > 0$, there exists an absolute constant $C > 0$ such that
         \begin{align*}
        \Exp\left[d_{g}(\hat{\ct}_N, \ct^{\pi,M}) \right] \leq C K_{1} L\frac{1}{\kappa_{\pi,M}} \sqrt{\frac{\log(M |\bu|)}{N}} + 4C_0 L M |\bu| \exp\left(-K_{2} \kappa_{\pi,M}^2 N\right).
    \end{align*}
Finally, the following upper bound holds: 
$\mathbb{P}\left(\mathcal{E}_{N, M}\right)  \leq 6C_0 M^2 |\bu|\exp\left(-K_2  \kappa_{\pi,M}^2\kappa_{\ct,M}^2 N/16 \right)$.
\end{lemma}

\begin{proof}
    The proof is presented in Appendix \ref{subapp:proof_Markov_con}.
\end{proof}

Next, we extend the cost function $\hat{c}_{N}$ to the original state space as in Definition \ref{def:piece-wise-extension-I}, and denote it as $\overline{\hat{c}_{N}}$. Further, we extend $\hat{\ct}_N$ to $\widetilde{\hat{\ct}}_N$ as in Definition \ref{def:piece-wise-extension-II-kernel}. It is clear that $\overline{\hat{\gamma}_{N,\beta}}$ and $\overline{\hat{\gamma}_{N,\infty}}$, which 
are extensions of $\hat{\gamma}_{N,\beta}$ and $\hat{\gamma}_{N,\infty}$ as in Definition \ref{def:piece-wise-extension-I},
are the optimal policies for the discounted-cost and average-cost problems associated with the MDP $(\bx, \bu, \widetilde{\hat{\ct}_{N}}, \overline{\hat{c}_{N}})$.  We compare the original and learned MDP, i.e., $(\ct,c)$ and $(\widetilde{\hat{\ct}_{N}}, \overline{\hat{c}_{N}})$, through the finite model approximation $(\widetilde{\ct^{\pi,M}},\overline{c^{\pi,M}})$; see  \eqref{finitemodeldef}. We start with the discounted problem. Recall $\delta_M$ from \eqref{diameter_def}.

\begin{theorem}\label{cor:single_tranjectory}
Suppose Assumptions \ref{myassump} and \ref{Model_Learning_Markov} hold. There exists a constant $C > 0$, depending only on $C_0,c_0,\clip, \tlip, \beta, \|c\|_{\infty}$, such that
\begin{align*}
\Exp\left[ \left\Vert\jb(c,\ct,{\overline{\hat{\gamma}_{N,\beta}}})-\jsc\right\Vert_{\infty}\right] \leq C \left(\delta_M + \frac{1}{\kappa_{\pi,M}}\sqrt{\frac{\log(M|\bu|)}{N}}\right) + C M|\bu| \exp\left(-C^{-1} \kappa_{\pi,M}^2 N\right).    
\end{align*}
\end{theorem}
\begin{proof}
By Lemma \ref{piecewiseconstantextensionerror} and the triangle inequality, we have $\|\overline{\hat{c}_{N}} - c\|_{\infty}\leq \clip \delta_{M} + \|\hat{c}_{N} - c^{\pi,M}\|_{\infty}$.

Clearly, $\|\jsc\|_{\infty} \leq \|c\|_{\infty}/(1-\beta)$. By Lemma \ref{lipbellman}, $\jsc$ is ${\clip}/{(1-\beta\tlip)}$-Lipschitz. Thus, by Lemma \ref{anotherextension_error} and Lemma \ref{trivialinequality}, due to the triangle inequality, we have
\begin{align*}
    d_{\jsc}(\ct, \widetilde{\hat{\ct}}_N) \leq  \frac{\clip(1+\tlip)}{1-\beta\tlip} \delta_{M} + d_{\jsc}({\ct}^{\pi,M}, {\hat{\ct}}_N).
\end{align*}
The result then follows by combining Lemma \ref{lemma:Markov_con} with \eqref{eq:discountedrobust1} stated in Theorem \ref{thrm:discounted_robust}.
\end{proof}

\begin{remark}
Note that in Theorem \ref{cor:single_tranjectory}, the constant \( C \) is independent of the partition. 
\end{remark}

Finally, we focus on the average-cost problem. Assume the reference kernel $\ct$ satisfies the minorization condition in Definition \ref{minordef} with some probability measure $\rho$ and scaling constant $\epsilon$. We first find a minorant $\tau^M$ for the approximating kernel $\widetilde{{\ct}^{\pi,M}}$, and then bound the Wasserstein-1 distance between $\rho$ and $\tau^M$.

\begin{lemma}\label{lemma:minorization_for_discretization}
Assume there exist a constant $\epsilon>0$ and a probability measure $\rho \in\mathcal{P}(\bx)$, such that $\forall\;x\in\bx,\;u\in\bu$, and $A\in\mathcal{F}_{\bx}$, $\ct(A|x,u)\geq \epsilon \rho(A)$. Define the following probability measure $\tau^M$: 
\begin{align*}
\tau^{M}(A) := \sum_{j=1}^{M} \rho(B_j)  \mathbbm{1}\{y_j\in A\}, \quad
\text{ for any } A\in\mathcal{F}_{\bx}.
\end{align*}
Then for any $x\in\bx,\;u\in\bu$ and $A\in\mathcal{F}_{\bx}$, 
$\widetilde{\ct^{\pi,M}}(A|x,u)\geq \epsilon \tau^M (A)\;  \text{ and }\; \wc_{1}(\rho, \tau^M) \leq \delta_M.
$
\end{lemma}
\begin{proof}
For any $A\in\mathcal{F}_{\bx}$, $u \in \bu$ and $x \in B_i$ with $1 \leq i \leq M$, by Definition \ref{def:piece-wise-extension-II-kernel},
\begin{align*}
\widetilde{\ct^{\pi,M}}(A|x,u) 
&=\sum_{j = 1}^{M} \frac{1}{\pi(B_i)}\int_{B_{i}}\ct(B_{j}|z,u){\pi}(dz)
 \mathbbm{1}\{y_j \in A\} \geq \sum_{j = 1}^{M} \frac{1}{\pi(B_i)} \int_{B_{i}} \rho(B_j) {\pi}(dz)
 \mathbbm{1}\{y_j \in A\},
\end{align*}
which implies the first claim.
Further, for any $f: \bx \to \bR$ with $\flip \leq 1$, by the triangle inequality,
\begin{align*}
\left|\int f(x) \rho(dx) - \int f(x) \tau^M(dx)\right| \leq
\sum_{j=1}^M \left|\int_{B_j} f(x) \rho(dx) - f(y_j) \rho(B_j)\right| \leq
\sum_{j=1}^M \int_{B_j} |f(x)  - f(y_j)| \rho(dx),
\end{align*}
which can be further bounded by  
$\delta_M$. The proof is complete. 
\end{proof}

Recall the event 
$\mathcal{E}_{N, M}$ defined in \eqref{def:E_N_M}. By definition, on its complement $\mathcal{E}_{N, M}^c$, the approximation kernel $\widetilde{\hat{\ct}_N}$ is minorized by the probability measure $\tau^M$ and constant $\epsilon/2$. Then we can apply the minorization approach in Theorem \ref{upperbound3acoe} on $\mathcal{E}_{N, M}^c$.

\begin{theorem}\label{cor:quantize}
Suppose that Assumptions \ref{acoelipassump} and \ref{Model_Learning_Markov} hold, and that $(\bx,\bu,\ct,c)$ satisfies the minorization condition with a probability measure $\rho$ and a constant $\epsilon >0$. There exists a constant $C > 0$, depending only on $C_0,c_0,\clip, \tlip, \|c\|_{\infty}, \epsilon$, such that
\begin{align*}
 \Exp\left[ \left\Vert J_{\infty}(c,\ct,{\overline{\hat{\gamma}_{N,\infty}}})-J_{\infty}^{*}(c,\ct) \right\Vert_{\infty} \right] \leq &C \left(\delta_M + \frac{1}{\kappa_{\pi,M}}\sqrt{\frac{\log(M|\bu|)}{N}}\right) + C M|\bu| \exp\left(-C^{-1} \kappa_{\pi,M}^2N\right) \\
 &+ C M^2 |\bu| \exp(-C^{-1} \kappa_{\pi,M}^2\kappa_{\ct,M}^2 N).
\end{align*}
\end{theorem}

\begin{proof}
Since $c$ is non-negative and bounded, and by Lemma \ref{lemma:Markov_con}, we have
\begin{align*}
    \Exp\left[ \left\Vert J_{\infty}(c,\ct,{\overline{\hat{\gamma}_{N,\infty}}})-J_{\infty}^{*}(c,\ct) \right\Vert_{\infty}; \mathcal{E}_{N,M}  \right] \leq 6C_0 \|c\|_{\infty}M^2 |\bu|\exp\left(-K_2  \kappa_{\pi,M}^2\kappa_{\ct,M}^2 N/16 \right).
\end{align*}
Now, we focus on the event $\mathcal{E}_{N,M}^c$, on which $\widetilde{\hat{\ct}_N}$ satisfies the minorization condition with $\tau^M$ and  $\epsilon/2$. Note that $\ct$ satisfies the minorization condition with $\rho$ and  $\epsilon/2$.

In the proof of Theorem \ref{cor:single_tranjectory}, we have shown $\|\overline{\hat{c}_{N}} - c\|_{\infty}\leq \clip \delta_{M} + \|\hat{c}_{N} - c^{\pi,M}\|_{\infty}$.
By Lemma \ref{acoelip}, $h_{c,\ct,\epsilon/2}^{*}$ is $\clip/(1-\tlip)$-Lipschitz. Thus, by Lemma \ref{anotherextension_error} and Lemma \ref{trivialinequality}, due to the triangle inequality, we have
\begin{align*}
    d_{h_{c,\ct,\epsilon/2}^{*}}(\ct, \widetilde{\hat{\ct}}_N) \leq  \frac{\clip(1+\tlip)}{1-\tlip} \delta_{M} + d_{h_{c,\ct,\epsilon/2}^{*}}({\ct}^{\pi,M}, {\hat{\ct}}_N).
\end{align*}
Finally, by Lemma \ref{lemma:hproperties} (c),
$\|h_{c,\ct,\epsilon/2}^{*}\|_{\infty} \leq 2\|c\|_{\infty}/\epsilon$. Then the proof is complete by combining Lemma \ref{lemma:Markov_con} and Lemma \ref{lemma:minorization_for_discretization} with Theorem \ref{upperbound3acoe}.
\end{proof}

A few remarks are in order. First, for a fixed partition, and hence a fixed $M$, the upper bounds in Theorem \ref{cor:single_tranjectory} and \ref{cor:quantize} achieve the optimal parametric rate, i.e., $O(N^{-1/2})$.

Second, assume that for each $x \in \bx$, $\|x\| \leq L$. Recall that $\delta_M$ is the quantization error in \eqref{diameter_def}. By a volume argument, there exists a quantization scheme such that
\begin{align}\label{quantizaton_error_scale}
\delta_M \leq (C/M)^{1/d},
\end{align}
for some constant $C$, depending only on $d$ and $L$. Under this assumption, we choose $N$ so that the approximation error $\delta_M$, arising from finite model approximation, matches the statistical estimation error from model learning. Specifically, define
\begin{equation*}
N_{\text{disc}} = C\frac{1}{\kappa_{\pi,M}^2}  M^{2/d} \log(M|\bu|), \quad
N_{\text{ave}} = C\frac{1}{\kappa_{\pi,M}^2 } M^{2/d} \log(M|\bu|) +  C\frac{1}{\kappa_{\pi,M}^2 \kappa_{\ct,M}^2}  \log(M|\bu|).
\end{equation*}
If we let $N =N_{\text{disc}}$ in the $\beta$-discounted cost case and $N = N_{\text{ave}}$ in the average-cost case, then the upper bounds in Theorem \ref{cor:single_tranjectory} and \ref{cor:quantize} on the overall performance loss, that include finite model approximation and model learning, are both of order $M^{-1/d}$.

Recall the definition of $\kappa_{\pi,M}$ and $\kappa_{\ct,M}$ in \eqref{def:kappa_pi_M_init}. Assume that 
\begin{equation}
    \label{kappa_assump}
\kappa_{\pi,M} \geq 1/(CM|\bu|), \quad
\kappa_{\ct,M} \geq 1/(CM).
\end{equation}
Then $N_{disc}$ is of order $M^{2+2/d}|\bu|^2\log(M|\bu|)$, while $N_{\text{ave}}$ is of order $M^{4}|\bu|^2\log(M|\bu|)$.

\begin{remark} [Relations with Quantized Q-Learning] We also note that the empirically estimated model is closely related to Quantized Q-Learning, as discussed in \cite[Theorem 8]{KSYContQLearning} and \cite[Theorem 2.1]{karayukselNonMarkovian}. Specifically, the Q-learning algorithm for the quantized model converges to the fixed-point equation corresponding exactly to the approximate finite model learned in this section, with the exploration policy guiding the empirical estimation of the model.  In this setting, direct model learning may be more efficient than Q-learning, both in terms of sample complexity and because a learned model can be combined with analytical or computational methods for solving the resulting control problem.
\end{remark} 

\subsubsection{Empirical model learning from independently generated transition data}
Let $\pi \in \mathcal{P}(\bx)$ be a \textit{given} weighting measure such that
 $\pi(B_i) > 0$ for $i \in [M]$. We denote by $\hat{\pi}_i$ the restriction of $\pi$ on the bin $B_{i}$, that is, $\hat{\pi}_i(A) := \pi(A\cap B_i)/\pi(B_i)$ for $i \in [M]$. We consider the data collected under the scheme described in Algorithm \ref{alg:sampling_restart}. Specifically, for each bin $i \in [M]$ and action $u \in \bu$, we have $N_0$ \textit{independent} triplets of observations: for $1 \leq k \leq N_0$,
$$
(X_{k,i,u}, Y_{k,i,u},  c(X_{k,i,u},u)), \quad \text{ where }
X_{k,i,u} \sim \hat{\pi}_i, \quad
Y_{k,i,u}|X_{k,i,u} \sim \ct(\cdot|X_{k,i,u},u).
$$
Note that in the above $\hat{\pi}_{i}$ has its full measure on $B_{i}$ and that the total number of triplets is $N := M \times |\bu| \times N_0$. Given the data, the estimated controlled kernel and cost function $(\hat{\ct}_{N}, \hat{c}_{N})$ on the quantized space $\{y_i\}_{i=1}^{M}$ are defined as follows.
For each $1 \leq i \leq M$ and action $u \in \bu$,
\begin{align}
\label{def:independent_transitions}
\begin{split}
    \hat{c}_{N}(y_i,u):=\frac{1}{N_0} \sum_{k=1}^{N_0}c(X_{k,i,u},u),\quad
    \hat{\ct}_N(y_j|y_i,u) := \frac{1}{N_0} \sum_{k=1}^{N_0} \mathbbm{1}\{Y_{k,i,u}\in B_{j}\}, \text{ for } 1 \leq j \leq M.
\end{split}
\end{align}
As in the previous section,  we can solve the finite state space MDP $(\{y_i\}_{i=1}^{M},\bu, \hat{\ct}_{N},\hat{c}_N)$ for the optimal $\beta$-discounted cost policy $\hat{\gamma}_{N,\beta}$ and the optimal average-cost policy $\hat{\gamma}_{N,\infty}$. We then extend them to the original state space, denoted as $\overline{\hat{\gamma}_{N,\beta}}$ and $\overline{\hat{\gamma}_{N,\infty}}$,  which are the optimal policies for $(\bx,\bu, \widetilde{\hat{\ct}_{N}},\overline{\hat{c}_N})$, where we recall the extension of kernel $\widetilde{\hat{\ct}_{N}}$ in Definition \ref{def:piece-wise-extension-II-kernel}.

\begin{algorithm}[t!]
\caption{Sampling a Wasserstein Regular MDP with Restart
}
\label{alg:sampling_restart}
\begin{algorithmic}[1]
\Require{number of repetitions $N_0$ for each state and action pair}

\ForAll{$i = 1,\ldots, M$}
\ForAll{$u\in\bu$}
\ForAll{$k=1,\;\ldots,\;N_0$}
\State{Sample i.i.d.~$X_{k,i,u}  \sim\hat{\pi}_{i}$ and $Y_{k,i,u} \sim \ct(\cdot|X_{k,i,u},u)$}
\State{Obtain cost $c(X_{k,i,u},u)$}
\EndFor
\State{Set $\hat{c}_{N}(y_i,u):=\frac{1}{N_0} \sum_{k=1}^{N_0}c(X_{k,i,u},u)$}
\State{Set 
$\hat{\ct}_N(y_j|y_i,u) := \frac{1}{N_0} \sum_{k=1}^{N_0} \mathbbm{1}\{Y_{k,i,u}\in B_{j}\}
$ for $j \in [M]$}
\EndFor
\EndFor
\State \Return $(\hat{\ct}_{N} , \hat{c}_{N})$
\end{algorithmic}
\end{algorithm}

We note two key differences compared to the single-trajectory setup in Subsection \ref{markovempiricalmodel} when independently generated transition data is used. First, the weighting measure $\pi$ is provided in this subsection and therefore does not need to be estimated. As a result, we can remove the terms involving $\kappa_{\tau,M}$ in Theorems \ref{cor:single_tranjectory} and \ref{cor:quantize}. Second, due to independence, Assumption \eqref{Model_Learning_Markov} holds: in particular, part (b) of Assumption \eqref{Model_Learning_Markov} holds due to the Hoeffding inequality for i.i.d.~random variables (see, e.g., Theorem 2.8 in \cite{CI_book}).

We now present upper bounds on the performance loss incurred when applying  $\overline{\hat{\gamma}_{N,\beta}}$ and $\overline{\hat{\gamma}_{N,\infty}}$, optimized under the estimated MDP $(\bx,\bu, \widetilde{\hat{\ct}_{N}},\overline{\hat{c}_N})$  to the true MDP $(\bx,\bu, \ct,c)$. The analysis follows similarly to that in Subsection \ref{markovempiricalmodel}, except that we replace part (b) of Assumption \eqref{Model_Learning_Markov} with Hoeffding inequality for i.i.d.~random variables (see, e.g., Theorem 2.8 in \cite{CI_book}), and that in bounding $\Pro\left(\mathcal{E}_{N,M}\right)$ (see \eqref{def:E_N_M}), we apply Bernstein's inequality. Therefore, we omit the detailed arguments for brevity.

\begin{theorem}\label{thm:independent_transitions_both}
Under Assumption \ref{myassump},
 there exists a constant $C > 0$, depending only on $ \clip, \tlip, \|c\|_{\infty}, \beta$, such that
\begin{align*}
\Exp\left[ \left\Vert\jb(c,\ct,{\overline{\hat{\gamma}_{N,\beta}}})-\jsc\right\Vert_{\infty}\right] \leq C \left(\delta_M +  \sqrt{\frac{\log(M|\bu|)}{N_0}}\right).    
\end{align*}
 Suppose that Assumption   \ref{acoelipassump} holds, and that $(\bx,\bu,\ct,c)$ satisfies the minorization condition with a probability measure $\rho$ and a constant $\epsilon >0$. Then, there exists a constant $C > 0$, depending only on $\clip, \tlip, \|c\|_{\infty}, \epsilon$, such that
\begin{align*}
 \Exp\left[ \left\Vert J_{\infty}(c,\ct,{\overline{\hat{\gamma}_{N,\infty}}})-J_{\infty}^{*}(c,\ct) \right\Vert_{\infty} \right] \leq  C \left(\delta_M +  \sqrt{\frac{\log(M|\bu|)}{N_0}}\right) + C    M^{2 } |\bu| e^{-C^{-1} \kappa_{\ct,M}  N_0},
\end{align*}
where we recall that $\kappa_{\ct,M}$ is defined in \eqref{def:kappa_pi_M_init}.
\end{theorem}

For a fixed partition, hence a fixed $M$, the upper bounds in Theorem \ref{thm:independent_transitions_both} achieve the optimal $O(N^{-1/2})$ parametric rate for finite model learning. 

Further, recall that $\delta_M$ is the quantization error in \eqref{diameter_def}, and that the total number of observations is $N = M |\bu| N_0$. Assume that \eqref{quantizaton_error_scale} holds for $\delta_M$ and that \eqref{kappa_assump} holds for $\kappa_{\ct,M}$. As before, we choose $N$ so that the order of finite model approximation error matches that of the model learning error. Specifically, define 
\begin{equation*}
N_{\text{disc}}' = C  M^{2/d+1} |\bu| \log(M|\bu|), \quad
N_{\text{ave}}'= C M^{\max\{2,\ 2/d + 1 \}}|\bu|  \log(M|\bu|).
\end{equation*}
If we let $N =N_{\text{disc}}'$ in the $\beta$-discounted cost case and $N = N_{\text{ave}}'$ in the average-cost case, then the upper bounds in Theorem \ref{thm:independent_transitions_both} are both of order $M^{-1/d}$. Thus, with i.i.d.~transition data, the sample complexity is significantly improved compared to the single trajectory case. As discussed earlier, there are two main reasons: first, we do not need to estimate $\pi$, which is given; second, we use Bernstein's inequality to bound $P(\mathcal{E}_{N,M})$, instead of relying on the Hoeffding inequality in the proof of Lemma \ref{lemma:Markov_con}. Finally, we note that the sample complexity improves as the dimension $d$ of $\bx$ increases. This is because, as $d$ grows, the finite model approximation error becomes larger, which results in a less stringent requirement on the finite model learning error.

\section{Application to Robustness to Noise Distribution Misspecification and Empirical Noise Estimation}\label{secDistApprx}

In this section, we consider the disturbance approximation and the associated robustness properties as discussed in \cite{Gordienko2007,Gordienko2008, Gordienko2022, Dufour2015-qu}. Specifically,  we consider the following stochastic dynamical system:
\begin{equation}\label{def:disturbance}
X_{t+1}=f(X_{t}, U_{t},W_{t}) \text{ for } t \geq 0, \;\; \text{ where }\;\; \{W_{t}\}_{t=0}^{\infty} \text{ is  i.i.d.~with some distribution } \mu.
\end{equation}
Here, $\{W_{t}\}_{t=0}^{\infty}$ is referred to as the disturbance process \revise{taking value in $\bW \subset \bR^d$}. Consider a decision maker who knows $f$ and the true cost function $c$, but 
has no knowledge of the distribution $\mu$. However, we assume $\mu$ can be estimated, say, from realized samples $\{w_{t}\}_{t=0}^{n}$. The decision maker then computes an optimal policy using an estimated distribution $\nu$, which may depend on the observed samples. Our goal is to upper bound the robustness error due to this misspecification by the distance between $\mu$ and $\nu$. We show that this estimation procedure can be viewed as a controlled kernel approximation, which allows us to leverage the results from the previous section. Additional details are given in Section \ref{drivingnoiseapprox}.

\subsection{Robustness to Noise Distribution Approximations}\label{drivingnoiseapprox}\label{drivingnoisecont}
To reduce the problem to controlled kernel approximation, we first establish a relationship between the distance between disturbance distributions and the distance between their corresponding controlled kernels, along with their Lipschitz continuity.

Let  $\mu,\nu \in \mathcal{P}(\mathbb{W})$  be two probability measures. For $x \in \bx$, $u \in \bu$ and $A \in \mathcal{F}_{\bx}$, define two controlled kernels:
\begin{equation*}
\ct_{\mu}(A|x,u):=\mu(f^{-1}_{x,u}(A)),\quad \text{ and } \quad \ct_{\nu}(A|x,u):=\nu(f^{-1}_{x,u}(A))
\end{equation*}
where $f^{-1}_{x,u}(A) :=\{w \in \mathbb{W}: f(x,u,w) \in A\}$. For any continuous and bounded function $g \in C_{b}(\bx)$, by definition,
$$
d_{g}(\ct_{\mu},\ct_{\nu}) =\sup_{(x,u) \in \bx \times \bu}\left|
\int_{\bW} g(f(x,u,w)) \mu(dw) -
\int_{\bW} g(f(x,u,w)) \nu(dw)
\right|.
$$

We start with the continuity and Lipschitz continuity property of the controlled kernels $T_\mu$ and $T_\nu$. The following lemma is stated for $T_\mu$, but it also applies to $T_{\nu}$.

We say that $f$ is $\flipxu$-Lipschitz continuous (resp.~continuous) in $(x,u)$ if for each $w\in \mathbb{W}$, $f(\cdot,\cdot,w):\;\bx \times \bu \to\bx$ is $\flipxu$-Lipschitz continuous (resp.~continuous). Further, we say $f$ is $\flipx$-Lipschitz continuous in $x$ if for each $(u,w)\in\bu \times \mathbb{W}$, $f(\cdot,u,w):\;\bx \to\bx$ is $\flipx$-Lipschitz continuous. Finally, we say $f$ is $\flipw$-Lipschitz continuous in $w$ if for each $(x,u)\in\bx\times\bu$, $f(x,u,\cdot)$ is $\flipw$-Lipschitz continuous.

\begin{lemma}\label{lemma:noise_approx_cts}
 \begin{enumerate}[label=\roman*)]
    \item If $f$ is continuous in $(x,u)$, then the transition kernel $\ct_{\mu}$ is weakly continuous on $\bx\times\bu$, that is, for any $v\in C_{b}(\bx)$, $\int_{\bx}v(x')\ct_{\mu}(dx'|x,u)$ is a continuous function of $\bx \times \bu$.

\item Assume that $f$ is $\flipxu$-Lipschitz continuous in $(x,u)$. Let $v:\bx\to \bR$ be a bounded, Lipschitz function with a constant $\|v\|_{\mathrm{Lip}}$. Then for each $x,y \in \bx$ and $u,u' \in \bu$, we have
   $$
   \left|\int_{\bx} v(x') \ct_\mu(dx'|x,u) 
   - \int_{\bx} v(x') \ct_\mu(dx'|y,u') 
   \right|\leq \flipxu\|v\|_{\mathrm{Lip}}\left(d_{\bx}(x,y) + d_{\bu}(u,u') \right).
   $$
   
    \item Assume that $f$ is $\flipx$-Lipschitz continuous in $x$. Then for any $u\in\bu$,  and $x,y\in\bx$,
\begin{align*}
&\wc_{1}\left(\ct_{\mu}(\cdot|x,u),\ct_{\mu}(\cdot|y,u)\right)\leq\flipx d_{\bx}(x,y).
\end{align*}
\end{enumerate}
\end{lemma}

\begin{proof}
Note that for any $v \in C_b(\bx)$, by definition,
$\int_{\bx} v(x') T_{\mu}(dx'|x,u) = \int_{\bW} v(f(x,u,w)) \mu(dw).
$

Under the assumption of claim (i), for each $w \in \bW$, $v(f(\cdot,\cdot,w)):\bx\times \bu \to \mathbb{R}$ is a continuous function bounded by $\|v\|_{\infty}$. Then claim (i) follows from the bounded convergence theorem.

Under the assumption of claim (ii), for each $w \in \bW$, $v(f(\cdot,\cdot,w)):\bx\times \bu \to \mathbb{R}$ is Lipschitz continuous with a constant $\|v\|_{\mathrm{Lip}} \flip$. Thus, for each $x,y \in \bx$ and $u,u' \in \bu$, we have
\begin{align*}
   &\left|\int_{\bx} v(x') \ct_\mu(dx'|x,u) 
   - \int_{\bx} v(x') \ct_\mu(dx'|y,u') 
   \right|\leq \int_{\bW} 
   \left|v(f(x,u,w)) - v(f(y,u',w))\right| \mu(dw) \\
   \leq &\flipxu\|v\|_{\mathrm{Lip}}\left(d_{\bx}(x,y) + d_{\bu}(u,u') \right),
\end{align*}
which proves the claim (ii). Finally, we focus on claim (iii). By the dual formulation of $\wc_{1}$, 
\begin{align*}
\wc_{1}(\ct_{\mu}(\cdot|x,u),\ct_{\mu}(\cdot|y,u))
&=\sup_{\gglip\leq 1}\left(\int g(f(x,u,w))\mu(dw)-\int g(f(y,u,w))\mu(dw) \right)\\
&\leq \sup_{\gglip\leq 1}\left(\int \flipx d_{\bx}(x,y) \mu(dw)\right) = \flipx d_{\bx}(x,y),
\end{align*}
where the supremum is taken over all $g: \bx \to \mathbb{R}$ such that $\gglip\leq 1$. The proof is complete.
\end{proof}

We denote by $\gamma_{\nu, \beta}^{*}$ and $\gamma_{\nu}^{*}$ the optimal policy for the MDP $(\bx, \bu, \ct_{\nu}, c)$ under the $\beta$-discounted  and average-cost criterion respectively; that is, these are the optimal policies when the common distribution of the noise sequence $\{W_t\}_{t \geq 0}$ in \eqref{def:disturbance} is given by $\nu$. Our goal is to bound the performance degradation incurred when these policies, optimized under $\ct_\nu$, are applied to an MDP governed by the true dynamics $\ct_{\mu}$. 
Recall that $h_{c,\ct_{\mu}}^{*}$ and $h_{c,\ct_{\nu}}^{*}$ are defined in Theorem  \ref{thrm:vanish_main}.

\begin{theorem}\label{cor:drivingnoisemainresult1}
Suppose that parts (a) and (b) in Assumption  \ref{basicassump} hold, and that $f$ is continuous in $(x,u)$. Then 
 \begin{align*}
\left\Vert\jb(c,\ct_{\mu},\gamma_{\nu,\beta}^{*})-\jb^{*}(c,\ct_{\mu}) \right\Vert_{\infty} \leq  \frac{2\beta }{\revise{(1-\beta)^2}}  \Delta_1, \text{ where } \Delta_1 :=
\sup_{(x,u) \in \bx \times \bu}\left|\int_{\bW} \jb^{*}(c,\ct_{\mu})(f(x,u,w)) (\mu-\nu)(dw)\right|.
 \end{align*}
\revise{Further, let $\bx$ be compact 
and assume that $c$ is  $\clip$-Lipschitz in $x$, and that $f$ is $\flipx$-Lipschitz continuous in $x$ and $\flipxu$-Lipschitz continuous in $(x,u)$. If $\flipx < 1$, then}
\begin{align*}
&\left\Vert J_{\infty}(c,\ct_{\mu},\gamma_{\nu}^{*}) - J_{\infty}^{*}(c,\mathcal{T}_{\mu})\right\Vert_{\infty}\leq \Delta_2 + \Delta_3, \quad \text{ where } \\
&\;\Delta_2 :=
\sup_{(x,u) \in \bx \times \bu}\left|\int_{\bW} h_{c,\ct_{\mu}}^{*}(f(x,u,w)) (\mu-\nu)(dw)\right|, \quad 
\Delta_3 := \sup_{(x,u) \in \bx \times \bu}\left|\int_{\bW} h_{c,\ct_{\nu}}^{*}(f(x,u,w)) (\mu-\nu)(dw)\right|.
\end{align*}
\end{theorem}

\begin{proof} 
Due to Lemma \ref{lemma:noise_approx_cts}(i), Assumption \ref{basicassump} holds for both claims. \revise{The first claim then follows immediately from \eqref{eq:discountedrobust1} stated in Theorem \ref{thrm:discounted_robust}.}

Now, we focus on the second claim and apply Theorem \ref{vanishingproof2}. Specifically, we only need to verify Assumptions \ref{equicontinuousassump}(b) and (c) for both MDPs, $\ct_\mu$ and $\ct_{\nu}$. Due to Lemma \ref{lemma:noise_approx_cts}(iii) and Lemma \ref{lipbellman}, for each $\beta \in (0,1)$, $\ct_\mu$ and $\ct_{\nu}$ are Lipschitz with a constant $L := {\clip}/{(1-\flipx)}$, which implies that  Assumption  \ref{equicontinuousassump}(b) holds for both MDPs.
Finally, by Lemma \ref{lemma:noise_approx_cts}(ii), Assumption  \ref{equicontinuousassump}(c) holds with $\widetilde{L} := \flipxu$ for both MDPs. The proof  is complete by \revise{part two in} Theorem \ref{vanishingproof2}.
 \end{proof}

Next, we present results that involve the $\wc_1$-distance between $\mu$ and $\nu$. We start with a lemma that bounds the $\wc_{1}$-distance between controlled kernels $T_\mu$ and $T_\nu$.

\begin{lemma}\label{kernelvsnoisewass}
Assume that for each $(x,u)\in\bx\times\bu$, $f(x,u,\cdot)$ is $\flipw$-Lipschitz continuous on $\mathbb{W}$. 
 Then $d_{\wc_{1}}(\ct_{\mu},\ct_{\nu})
\leq {\flipw} \wc_{1}(\mu,\nu)$.

\end{lemma}
\begin{proof} By the dual formulation of Wasserstein-1 distance,
\begin{align*}
\sup_{x,u}\wc_{1}(\ct_{\mu}(\cdot|x,u),\ct_{\nu}(\cdot|x,u))
=\sup_{x,u}\sup_{\gglip\leq 1}\left|\int
 g(f(x,u,w))\mu(dw) -\int
 g(f(x,u,w))\nu(dw)\right|,
\end{align*}
where the supremum is taken over all $(x,u) \in \bx \times \bu$ and $g: \bx \to \mathbb{R}$ such that $\gglip\leq 1$. Clearly, for each fixed $x,u$, the function $g(f(x,u,\cdot)): \bW \to \mathbb{R}$ is $\flipw$-Lipschitz. Then the proof is complete due to the definition of Wasserstein-1 distance.
\end{proof}

\begin{theorem}\label{drivingnoisemainresult}
Assume that $\bu$ is compact and that $c:\bx\times\bu\to\bR$ is nonnegative, bounded,   continuous, and $\clip$-Lipschitz continuous in $x$. Further,  assume that $f$ is continuous in $(x,u)$, $\flipx$-Lipschitz continuous in $x$ and $\flipw$-Lipschitz continuous in $w$.
\begin{enumerate}[label=\roman*)]
    \item 
If $\beta \flipx < 1$, then 
 \begin{align*}
\left\Vert\jb(c,\ct_{\mu},\gamma_{\nu,\beta}^{*})-\jb^{*}(c,\ct_{\mu}) \right\Vert_{\infty} \leq  \frac{2\beta\clip\flipw}{(1-\beta)(1-\beta\flipx)}\wc_{1}(\mu,\nu).
 \end{align*}
\item  \revise{If $\bx$ is compact and $\flipx < 1$
, then }
\begin{align*}
\left\Vert J_{\infty}(c,\ct_{\mu},\gamma_{\nu}^{*}) - J_{\infty}^{*}(c,\mathcal{T}_{\mu})\right\Vert_{\infty} \leq\frac{2\clip\flipw}{1-\flipx} \wc_{1}(\mu,\nu).
\end{align*}
\end{enumerate}
\end{theorem}

\begin{proof}
Due to Lemma \ref{lemma:noise_approx_cts}(i), Assumption \ref{basicassump} holds. Due to Lemma \ref{lemma:noise_approx_cts}(iii) and Lemma \ref{lipbellman}, for each $\beta \in (0,1)$, $J_{\beta}^{*}(c,\ct)$ is Lipschitz with a constant ${\clip}/{(1-\revise{\beta}\flipx)}$.   Then the first claim follows immediately from Lemma \ref{kernelvsnoisewass}, Lemma \ref{trivialinequality} and result \eqref{eq:discountedrobust2} in Theorem  \ref{thrm:discounted_robust}. The second claim follows from Theorem \ref{vanishingproof2}, similarly to the proof of Theorem \ref{cor:drivingnoisemainresult1}.
\end{proof}

\begin{remark}Note that in Theorem \ref{drivingnoisemainresult}, the upper bounds involve the $\wc_1$-distance between $\mu$ and $\nu$. In contrast, Theorem \ref{cor:drivingnoisemainresult1} requires bounding only the difference between $\mu$ and $\nu$ integrated against a \textit{single} function. 
For this reason, as we shall see in Subsection \ref{subsec:complexity_bounds},  we can achieve the \textit{optimal} statistical rates under the conditions of Theorem \ref{cor:drivingnoisemainresult1}, when we estimate $\mu$ by empirical distributions, which is \textit{not} the case if we use Theorem \ref{drivingnoisemainresult}. The following two subsections will highlight this distinction.
\end{remark}

Further, Theorem \ref{cor:drivingnoisemainresult1} requires that   $f$ is $\flipxu$-Lipschitz in $(x,u)$, which is not required for Theorem \ref{drivingnoisemainresult}. On the other hand, Theorem \ref{drivingnoisemainresult} requires that   $f$ is $\flipw$-Lipschitz in $w$, which is not required for Theorem \ref{cor:drivingnoisemainresult1}.

\subsection{General Sample Complexity Bounds}
\label{subsec:complexity_bounds}
In this subsection, we consider a dynamical system in which the state evolution given the previous state, current action, and noise realization is known, but the noise distribution is unknown and must be estimated, and we derive sample-complexity bounds in terms of the resulting performance degradation.

\subsubsection{Sample Complexity Bounds  under Regularity in Noise}
Recall the stochastic dynamic system in \eqref{def:disturbance}, where the function $f$ is known but the disturbance distribution $\mu$ is \textit{unknown}. Denote by $\ct_{\mu}$ the associated controlled transition kernel.

Assume  that we can observe an i.i.d.~sequence $\{W_{t}\}_{t=0}^{\infty}$ such that for any $t\geq 0$, $W_{t}\sim \mu$, and that we use the empirical measures $\{\mu_n:n \geq 0\}$  to estimate $\mu$:
\begin{equation}\label{def:empirical_measures}
\mu_{n}(\cdot):=\frac{1}{n+1}\sum_{t=0}^{n}\delta_{W_{t}}(\cdot), \text{ for } n \geq 0,
\end{equation}

For $n \geq 0$, given the estimate $\mu_n$, we consider the following stochastic dynamic system:
\begin{align*}
\widetilde{X}_{t+1}=f(\widetilde{X_{t}},\widetilde{U}_{t},\widetilde{W}_{t}) \text{ for } t \geq 0, \;\; \text{ where }\;\; \{\widetilde{W}_{t}\}_{t=0}^{\infty} \text{ is  i.i.d.~with distribution } \mu_n,
\end{align*}
and denote by $\ct_{\mu_n}$ the associated controlled transition kernel. Since the cost function $c$ is assumed known, we can solve the $\beta$-discounted cost and the average-cost MDP $(\bx,\bu,\ct_{\mu_n},c)$, and denote the solutions by 
$\gamma^*_{\mu_n,\beta}$ and $\gamma^*_{\mu_n}$ respectively.

By applying the result from \cite{fournier2013rate} concerning the Wasserstein-1 distance between $\mu$ and its empirical counterpart $\mu_n$, Theorem \ref{drivingnoisemainresult} immediately yields the following result. For $q \geq 1$, denote the $q$-th moment of $\mu$ by 
$M_{q}(\mu):=\int_{\bR^{d}}|w|^{q}\mu(dw)$.

\begin{theorem} 
\label{thrm:empirical_driving_noise_rates}
\revise{Assume the conditions in Theorem \ref{drivingnoisemainresult} hold.} 
Let  $\beta\in(0,1)$ be given such that $\beta\flipx< 1$. 

\begin{enumerate}[label=\roman*)]
    \item  If there exists $q>1$ such that $M_{q}(\mu)<\infty$, then there exists a positive constant $C$ depending on $d$, $q$, $\mu$, $\beta$, $\clip$, $\flipx$, and $\flipw$, such that for all $n\geq 1$,
\begin{equation*}
\mathbb{E}\left[\left\Vert\jb(c,\ct_{\mu},\gamma_{\mu_{n},\beta}^{*})-J_{\beta}^{*}(c,\ct_{\mu})\right\Vert_{\infty}\right]\leq C
\begin{cases}
n^{-1/2}+n^{-(q-1)/q}, &\text{if }d<2,\;q\neq 2,\\
n^{-1/2}\log\left(1+n\right)+n^{-(q-1)/q}, &\text{if }d=2,\;q\neq 2,\\
n^{-1/d}+n^{-(q-1)/q}, &\text{if }d> 2,\;q\neq d/(d-1).
\end{cases}
\end{equation*}

\item Additionally, \revise{assume that $\bx$ is compact and $\flipx < 1$}. Then the above result holds under the average-cost criterion if we replace  $J_{\beta}$, $J_{\beta}^{*}$, and $\gamma_{\mu_{n},\beta}^*$  by $J_{\infty}$, $J_{\infty}^{*}$, and $\gamma_{\mu_{n}}^*$, respectively; in this case, $C$ no longer depends on $\beta$.  
\end{enumerate}
\end{theorem}

Assume that $M_q(\mu) < \infty$ for $q > \max\{2,d\}$. Then the convergence rate is $n^{-1/2}$, $n^{-1/2} \log(n)$, and $n^{-1/d}$ for $d = 1$, $2$, and $d > 2$, respectively. For $d = 1$, the rate is parametric and thus optimal. For $d \geq 2$, the rate is generally suboptimal. However, we emphasize that the above theorem does not assume $f$ to be jointly Lipschitz in both $x$ and $u$. Under this additional regularity condition, we derive improved rates \revise{under discounted cost} in the next subsection.

\subsubsection{Improved Sample Complexity Bounds under Regularity in State and Action}  As in the previous subsection, suppose that we do not know the probability measure $\mu \in \mathcal{P}(\mathbb{W})$, but we compute empirical estimates as in (\ref{def:empirical_measures}).

\revise{
In view of Theorem \ref{cor:drivingnoisemainresult1}, the discounted-cost case reduces to bounding $\Delta_1$ with $\nu$ replaced by the empirical measure $\mu_n$. This can be handled using standard empirical-process tools \cite{vanderVaart1996}. The key observation is that, under the regularity assumptions below, the function $(x,u)\mapsto J_\beta^*(c,\mathcal T_\mu)(f(x,u,w))$ is uniformly Lipschitz in $(x,u)$, with a Lipschitz constant that can be controlled explicitly. We denote $\|c\|_\infty:=\sup_{x,u}|c(x,u)|$, and let $\mathrm{diam}(\bx)$ and $\mathrm{diam}(\bu)$ be the diameters of $\bx$ and $\bu$, respectively.
}

\begin{theorem}\label{cor:mu_n_result}
Assume that both $\bx \subset \bR^{d_1}$ and $\bu\subset \bR^{d_2}$ are compact, and that $c:\bx\times\bu\to\bR$ is nonnegative, bounded,   continuous, and $\clip$-Lipschitz in $x$. Further, assume that $f$ is $\flipx$-Lipschitz continuous in $x$ and  $\flipxu$-Lipschitz continuous in $(x,u)$. If $\beta \flipx < 1$, then there exists a constant $C > 0$, depending only on $d_1$, $d_2$, $\mathrm{diam}(\bx)$ and $\mathrm{diam}(\bu)$, such that  
$$
\Exp\left[\left\Vert\jb(c,\ct_{\mu},\gamma_{\mu_n,\beta}^{*})-\jb^{*}(c,\ct_{\mu}) \right\Vert_{\infty} \right]\leq  C \frac{\beta }{\revise{(1-\beta)^2}}  \left( \frac{\|c\|_{\infty}}{1-\beta} + 
\frac{\flipxu \clip }{1-\beta \flipx}
\right) n^{-1/2}.
$$
\end{theorem}
\begin{proof}
Define the following constants
\begin{align}
 \label{aux:constants}   
m_{\beta} := \frac{\clip}{1-\beta \flipx}\;\;\text{ and }\;\; M_{\beta} := m_{\beta}\flipxu.
\end{align}
Due to Theorem \ref{cor:drivingnoisemainresult1}, it suffices to bound 
$$
\Exp\left[\Delta_1\right] = \Exp\left[\sup_{(x,u) \in \bx \times \bu} \left|
\frac{1}{n+1} \sum_{t=0}^{n} \left(
\jb^{*}(c,\ct_{\mu})(f(x,u,W_t)) - 
\int_{\bW} \jb^{*}(c,\ct_{\mu})(f(x,u,w)) \mu(dw) \right)
\right|\right],
$$
where we replace $\nu$ by $\mu_n$ in the definition of $\Delta_1$. Clearly, $|\jb^{*}(c,\ct_{\mu})(f(x,u,w))| \leq \|c\|_{\infty}/(1-\beta) := B_{\beta}$ for each $(x,u,w) \in \bx \times \bu \times \bW$. Further, due to Lemma \ref{lemma:noise_approx_cts}(iii) and Lemma \ref{lipbellman}, for each $\beta \in (0,1)$, $J_{\beta}^{*}(c,\ct_{\mu})$ is Lipschitz with a constant $m_{\beta}$. Thus, for each $\beta \in (0,1)$, $\jb^{*}(c,\ct_{\mu})\circ f$ is $M_{\beta}$-Lipschitz continuous in $(x,u)$. By applying Lemma \ref{app:empirical_proc} in Appendix \ref{app:statistical_analysis} with constants $B_{\beta}$ and $M_{\beta}$ (see also Remark \ref{remark:empir_proc}), we have
\begin{align}\label{aux:driving_noise_step}
    \Exp\left[\Delta_1 \right]\leq C    \left( \frac{\|c\|_{\infty}}{1-\beta} + 
\frac{\flipxu \clip }{1-\beta \flipx}
\right) n^{-1/2}.
\end{align}
This completes the proof for the $\beta$-discounted case  in view of Theorem \ref{cor:drivingnoisemainresult1}.
\end{proof}

In Theorem \ref{cor:mu_n_result}, we achieve the parametric statistical rate $O(n^{-1/2})$ \revise{under the discounted-cost criterion}, which in general cannot be improved. In \cite{joseblanchet2024statisticallearningdistributionallyrobust},  related results are derived for the $\beta$-discounted criterion under similar conditions. Specifically, in \cite{joseblanchet2024statisticallearningdistributionallyrobust}, it is assumed that both $\bx$ and $\bu$ are compact, and that $\jb^{*}(c,\ct_{\mu})\circ f$ is assumed to be $L$-Lipschitz in $(x,u)$ (part (3) of their Assumption 2). It is not a priori clear under what conditions the composition $\jb^{*}(c,\ct_{\mu}) \circ f$ is Lipschitz continuous in $(x,u)$. In this work, we provide concrete sufficient conditions---specifically, that $\flipx < \beta^{-1}$---under which the Lipschitz constant $L$ can be chosen as ${\flipxu \clip}/{(1 - \beta \flipx)}$.

\revise{
\begin{remark}
The parametric bound in Theorem \ref{cor:mu_n_result} is stated only for the discounted-cost criterion. 
Although Theorem \ref{cor:drivingnoisemainresult1} gives both discounted- and average-cost robustness bounds, the average-cost case does not lead to a comparable empirical-process argument directly. Indeed, applying the average-cost bound requires controlling both $\Delta_2$ and $\Delta_3$, and the main difficulty is 
$$
\Delta_{3}:=\sup_{(x,u)\in\bx\times\bu}
\left|
\int_{\bW} h_{c,\ct_{\mu_n}}^{*}(f(x,u,w))(\mu-\mu_n)(dw)
\right|.
$$
Here $h_{c,\ct_{\mu_n}}^{*}$ is data-dependent. 
By Theorem \ref{thrm:vanish_main}, $h_{c,\ct_{\mu_n}}^{*}$ is $L$-Lipschitz. 
Furthermore, we can require $h_{c,\ct_{\mu_n}}^{*}(z) =0$ for a fixed $z \in \bx$.  
However, without a finer characterization of how this relative value function depends on the empirical noise distribution $\mu_n$, one is led to control the larger class
$$
\mathcal H_L
:=
\left\{
w\mapsto h(f(x,u,w)):
(x,u)\in\bx\times\bu,\; 
h:\bx \to \bR \text{ is } L\text{-Lipschitz with }
\|h\|_\infty\le L\times\mathrm{diam}(\bx)
\right\}.
$$
This is a class of functions of $w \in \bW \subset \bR^d$, indexed by $(x,u,h)$.
Assume that $\bx \subset \bR^{d_1}$ is compact and convex, $\bu \subset \bR^{d_2}$ is compact, and
$f$ is Lipschitz in $(x,u)$.
Using the standard bracketing bound for bounded Lipschitz functions on $\bx$ \cite[Corollary 2.7.2]{vanderVaart1996}, 
and combining it with a finite $\epsilon$-net of $\bx\times\bu$, the uniform Lipschitz
continuity of $f$ in $(x, u)$ yields,
up to constants,
$$
\log N_{[]}\bigl(\epsilon,\mathcal H_L,\|\cdot\|_{\mathcal{L}_2(\mu)}\bigr)
\lesssim
(1/\epsilon)^{d_1}
+
(d_1+d_2)\log(1/\epsilon).
$$
Consequently, the entropy integral required by \cite[Theorem 2.14.2]{vanderVaart1996} contains
$\int_0^1 (1/\epsilon)^{d_1/2}\,d\epsilon,
$
which is finite only when $d_1<2$. Thus, this direct empirical-process approach yields a parametric average-cost bound only in the one-dimensional state space case, and does not provide a parametric average-cost counterpart in higher-dimensional state spaces.
\end{remark}
}

\revise{
\begin{remark}[Parametric disturbance models]
The nonparametric rates obtained by applying Theorem \ref{drivingnoisemainresult} with the empirical measure $\mu_n$ are driven by the convergence rate of $\wc_1(\mu,\mu_n)$. 
If the disturbance distribution belongs to a finite-dimensional parametric family, this bottleneck can disappear. Thus, the slower rates for $\mu_n$  reflect the nonparametric nature of Wasserstein estimation; for the average-cost criterion, whether an $n^{-1/2}$ rate is attainable more generally, without such parametric structure,  remains open.

Specifically, suppose that the true disturbance distribution is $\mu_{\theta_0}$, where $\theta_0\in\Theta$ and $\Theta\subset\mathbb R^m$ is the parameter space, and let $\widehat\theta_n\in\Theta$ be an estimator based on the i.i.d.~data $\{W_t:1 \leq t \leq n\}$. Assume that
$$
\Exp\left[\|\widehat\theta_n-\theta_0\|_2\right]=O(n^{-1/2}),
$$
as holds, for example, for regular parametric estimators under standard regularity conditions; see, e.g., \cite[Chapter 5]{van2000asymptotic}. If, in addition, the parametric family is Lipschitz in Wasserstein-1 distance, namely for some constant $C_\Theta > 0$,
$$
\wc_1(\mu_{\theta_1},\mu_{\theta_2})
\le C_\Theta\|\theta_1-\theta_2\|_2,
\qquad \theta_1,\theta_2\in\Theta,
$$
then under the average-cost assumptions of Theorem \ref{drivingnoisemainresult},
$$
\Exp\left[
\left\|J_\infty(c,\ct_{\mu_{\theta_0}},\gamma^*_{\mu_{\widehat\theta_n}})
-J^*_\infty(c,\ct_{\mu_{\theta_0}})\right\|_\infty
\right]
=O(n^{-1/2}).
$$

A simple sufficient condition for the Wasserstein Lipschitz property is the following. Suppose that $\bW=\mathbb R^d$, that $\mu_\theta$ admits a density $p_\theta$ with respect to Lebesgue measure, that $\Theta$ is convex, and that $p_\theta$ is continuously differentiable in $\theta$. Assume that
$$
C_{\Theta} := 
\sup_{\theta\in\Theta}
\int_{\mathbb R^d}
\|w\|_2\,\left(\|\nabla_\theta p_\theta(w)\|_2 + p_{\theta}(w)\right)\,dw
<\infty.
$$
The term involving $p_{\theta}$ ensures that each $\mu_{\theta}$
 has finite first moment, while the derivative term gives the Lipschitz bound. Specifically,
fix an arbitrary function $k:\bR^{d} \to \bR$ with $\|k\|_{\mathrm{Lip}} \leq 1$ and any $\theta_1,\theta_2 \in \Theta$. Note that
\begin{align*}
&\left|\int_{\bW} k(w) p_{\theta_1}(w) dw - \int_{\bW} k(w) p_{\theta_2}(w) dw\right| =
\left|\int_{\bW} (k(w) -k(0))(p_{\theta_1}(w) -p_{\theta_2}(w)) dw\right|\\
\leq &
\int_{\bW} |k(w) -k(0)| |p_{\theta_1}(w) -p_{\theta_2}(w)| dw 
\leq \int_{\bW} \|w\|_2 |p_{\theta_1}(w) -p_{\theta_2}(w)| dw.
\end{align*}
Note that for each $w \in \bR^{d}$, by the Cauchy-Schwarz inequality,
\begin{align*}
\left|p_{\theta_1}(w) -p_{\theta_2}(w)\right|=   \left|\int_0^1
\nabla_\theta p_{\theta_2+t(\theta_1-\theta_2)}(w)^\top(\theta_1-\theta_2)\;dt\right| 
\leq \int_0^1 \|\nabla_\theta p_{\theta_2+t(\theta_1-\theta_2)}(w)\|_2 \|\theta_1-\theta_2\|_2 dt,
\end{align*}
which implies that
\begin{align*}
&\left|\int_{\bW} k(w) p_{\theta_1}(w) dw - \int_{\bW} k(w) p_{\theta_2}(w) dw\right| \leq C_{\Theta} \|\theta_1 - \theta_2\|_2.
\end{align*}
Finally, by the Kantorovich--Rubinstein dual representation of $\wc_1$,
$$
\wc_1(\mu_{\theta_1},\mu_{\theta_2})
\leq
C_\Theta\|\theta_1-\theta_2\|_2.
$$
Thus, for regular parametric disturbance models, 
the robustness error has the parametric order.
\end{remark}
}

\subsection{Learning the Dynamics and Noise Simultaneously}
In this subsection, we assume that $\bx \subset \bR^{d_1}$ is convex and compact, and that $\bu \subset \bR^{d_2}$ is compact. \revise{Both spaces are equipped with the Euclidean metric $\|\cdot\|_2$. Denote by $\Pi_{\bx}: \bR^{d_1} \to \bx$ the  projection onto $\bx$ in terms of the Euclidean distance; note that $\Pi_{\bx}$ is $1$-Lipschitz. 
Consider the following additive dynamic system, which is a special case of \eqref{def:disturbance}:
\begin{equation*}
X_{t+1}=f(X_{t}, U_{t}, W_t) = \Pi_{\bx}\left(r(X_{t}, U_{t})+W_{t}\right), \;\;\text{ for } t\geq 0,
\end{equation*}
}
where $\{W_{t}\}_{t=0}^{\infty}$ are i.i.d.~with some distribution $\mu$ on $\bR^{d_1}$, and $r:\bx \times \bu \to \bR^{d_1}$ is continuous.

In contrast to the previous subsection, we now consider the case in which both the function 
$r(\cdot)$ and the distribution $\mu$ are \textit{unknown} and must be learned simultaneously from data. Specifically, we have access to the following data $(\tilde{Y}_i, \tilde{X}_i,\tilde{U}_i)$ for $0 \leq i \leq n$, where
\begin{align}\label{def:data_set}
    \tilde{Y}_i = r(\tilde{X}_i,\tilde{U}_i) + \tilde{W}_i, \quad \text{ where } \{\tilde{W}_i\}_{i=0}^{\infty} \;\text{ are  i.i.d.~with  distribution } \mu.
\end{align}
The data in \eqref{def:data_set} are assumed to be unprojected simulator observations of $r(\tilde{X}_i,\tilde{U}_i)+\tilde{W}_i$, rather than projected next-state observations from the controlled process. 
Note that the noise variables $\{\tilde{W}_i\}_{0 \leq i \leq n}$ are \textit{unobserved}, which means the empirical measure in \eqref{def:empirical_measures} is not available. %
Based on this data, we denote by $\tilde{r}_n(\cdot)$ an estimator of the function $r(\cdot)$, and define the following estimator for the distribution $\mu$:
\begin{align}\label{def:empirical_measure_approx}
    \tilde{\mu}_n(\cdot):=\frac{1}{n+1}\sum_{i=0}^{n}\delta_{\hat{W}_i}(\cdot), \text{ for } n \geq 0, \text{ where }
\hat{W}_i := \tilde{Y}_i -  \tilde{r}_n(\tilde{X}_i, \tilde{U}_i),
\end{align}
where we recall that $\delta_{\hat{W}_i}(\cdot)$ is the Dirac measure at $\hat{W}_i$. Concrete examples of the estimator $\tilde{r}_n$ will be presented later. For $x \in \bx$, $u \in \bu$ and $A \in \mathcal{F}_{\bx}$, define the following controlled kernels:
\revise{
\begin{align*}
\ct_{r,\mu}(A|x,u)&:=\mu(\{w\in \bR^{d_1}: \Pi_{\bx}(r(x,u)+w) \in A\}),\\
\ct_{\tilde{r}_n, \tilde{\mu}_n}(A|x,u)&:=\tilde{\mu}_n(\{w\in \bR^{d_1}: \Pi_{\bx}(\tilde{r}_n(x,u)+w) \in A\}).
\end{align*}
The projection ensures that $\ct_{r,\mu}$ and $\ct_{\tilde{r}_n, \tilde{\mu}_n}$ are transition kernels on $\bx$.}
 Thus, $(\bx, \bu, \ct_{r,\mu}, c)$ is the reference MDP, while  $(\bx, \bu, \ct_{\tilde{r}_n, \tilde{\mu}_n}, c)$ is an approximation.  We denote by $\gamma_{n, \beta}^{*}$ the optimal policy for the approximate MDP $(\bx, \bu, \ct_{\tilde{r}_n, \tilde{\mu}_n}, c)$ under the $\beta$-discounted-cost criterion. Our goal is to bound the performance degradation incurred when they are applied under the true dynamics $\ct_{r,\mu}$. We note that the asymptotic convergence for problems of this type has been considered \cite[Section 1.3.2(iv)]{kara2020robustness}. Here, we obtain explicit and quantitative bounds \revise{under the discounted-cost criterion}.

\begin{theorem}\label{cor:model_noise_both_learned}
Assume that $c:\bx\times\bu\to\bR$ is nonnegative, bounded,   continuous, and $\clip$-Lipschitz in $x$, and \revise{that $\tilde{r}_n$ is continuous in $(x,u)$ almost surely}. Furthermore, assume that $r$ is $\rlipx$-Lipschitz continuous in $x$ and  $\rlipxu$-Lipschitz continuous in $(x,u)$. If $\rlipx < \beta^{-1}$, then there exists a constant $C > 0$, depending only on $d_1, d_2$, $\mathrm{diam}(\bx)$ and $\mathrm{diam}(\bu)$, such that %
\begin{align*}
    &\Exp\left[\left\Vert\jb(c,\ct_{r, \mu},\gamma_{n,\beta}^{*})-\jb^{*}(c,\ct_{r, \mu}) \right\Vert_{\infty}\right]\\
\leq & C \frac{\beta }{\revise{(1-\beta)^2}} \left[ \left( \frac{\|c\|_{\infty}}{1-\beta} + 
\frac{\rlipxu \clip }{1-\beta \rlipx}
\right) n^{-1/2}
+ \frac{\clip}{1-\beta \revise{\rlipx}} \Exp\left[\|\tilde{r}_n - r\|_{\infty}
\right] \right].
\end{align*}

\end{theorem}

\begin{proof}
We first focus on the $\beta$-discounted case, and denote by $J(\cdot) := \jb^{*}(c,\ct_{r, \mu})(\cdot)$ and $\Delta_{\beta} := \left\Vert\jb(c,\ct_{r, \mu},\gamma_{n,\beta}^{*})-\jb^{*}(c,\ct_{r, \mu}) \right\Vert_{\infty} $.
Further, let $\mu_n := (n+1)^{-1} \sum_{i=0}^{n} \delta_{\tilde{W}_i}$ be the empirical measure.

By a similar argument as in Theorem \ref{cor:drivingnoisemainresult1}, due to Theorem \ref{thrm:discounted_robust}, we have
\begin{align*}
&\Delta_{\beta} \leq \revise{ \frac{2\beta}{(1-\beta)^2}}\sup_{x,u} \left| \int_{\bR^{d_1}} J(\revise{\Pi_{\bx}(r(x,u)+w)}) \mu(dw) - \int_{\bR^{d_1}} J(\revise{\Pi_{\bx}(\tilde{r}_n(x,u)+w)}) \tilde{\mu}_n (dw) 
\right| \\
\leq & \revise{\frac{2\beta}{(1-\beta)^2}}\left( (I)_{n} + (II)_{n} + (III)_{n}\right),
\end{align*}
where we define
\begin{align*}
&(I)_{n} := \sup_{x,u} \left| \int_{\bR^{d_1}} J(\revise{\Pi_{\bx}(r(x,u)+w)}) \mu(dw) - \int_{\bR^{d_1}} J(\revise{\Pi_{\bx}(\revise{r(x,u)}+w)}) {\mu}_n (dw) 
\right|, \\
& (II)_{n} := \sup_{x,u} \left| \int_{\bR^{d_1}} J(\revise{\Pi_{\bx}(r(x,u)+w)}) \mu_n(dw) - \int_{\bR^{d_1}} J(\revise{\Pi_{\bx}(\tilde{r}_n(x,u)+w)}) {\mu}_n (dw) 
\right|,\\
&(III)_{n} := \sup_{x,u} \left| \int_{\bR^{d_1}} J(\revise{\Pi_{\bx}(\tilde{r}_n(x,u)+w)}) \mu_n(dw) - \int_{\bR^{d_1}} J(\revise{\Pi_{\bx}(\tilde{r}_n(x,u)+w)}) \tilde{\mu}_n (dw) 
\right|.
\end{align*}
We have bounded $\Exp\left[(I)_{n}\right]$ in \eqref{aux:driving_noise_step}, where $\flipx = \rlipx$ and $\flipxu = \rlipxu$, since $\Pi_{\bx}$ is $1$-Lipschitz. Furthermore,  in the proof of Theorem \ref{cor:mu_n_result}, we have shown that $J(\cdot)$ is $m_{\beta}$-Lipschitz, where $m_{\beta}$ is defined in \eqref{aux:constants}; thus,
$(II)_{n} \leq m_{\beta} \|\tilde{r}_n - r\|_{\infty}$. Finally, by definition,
\begin{align*}
 (III)_{n}   = \sup_{x,u} \left| 
 \frac{1}{n+1}\sum_{i=0}^{n} \left(
 J(\revise{\Pi_{\bx}(\tilde{r}_n(x,u)+\tilde{W}_i)}) - J(\revise{\Pi_{\bx}(\tilde{r}_n(x,u)+ \hat{W}_i)})
 \right)
 \right|,
\end{align*}
 where $\hat{W}_i$ is defined in \eqref{def:empirical_measure_approx}. By definition, $|\hat{W}_i - \tilde{W}_i| \leq \|\tilde{r}_n - r\|_{\infty}$. As a result,
 $(III)_{n} \leq m_{\beta} \|\tilde{r}_n - r\|_{\infty}$.
 The proof is complete.
\end{proof}

Compared to Theorem \ref{cor:mu_n_result}, Theorem \ref{cor:model_noise_both_learned} has an extra term involving $\Exp\left[\|\tilde{r}_n - r\|_{\infty}\right]$. Next, we provide two examples where this sup-norm estimation error can be controlled.

\begin{example}
 Assume that $\Exp[\tilde{W}_0] = 0$, that $\bx \subset \bR$ and $\bu \subset \bR^{d_2}$ are compact, and that $r(x,u) = \alpha_0 x + \theta_0'u$ for $x \in \bx, u \in \bu$, where $\alpha_0 \in \bR$ and $\theta_0 \in \bR^{d_2}$ are both unknown. Assume that $\{(\tilde{X}_i, \tilde{U}_i): 0 \leq i \leq n\}$ are i.i.d.,  independent from $\{\tilde{W}_i: 0 \leq i \leq n\}$. We can estimate the coefficients by the least squares method:
$(\tilde{\alpha}_n, \tilde{\theta}_n) := \argmin_{(\alpha,\theta) \in \bR^{1+d_2}} \sum_{i=0}^{n} \left(\tilde{Y}_i - \alpha \tilde{X}_i - \theta'\tilde{U}_i\right)^2$. 
If the covariance matrix  of $(\tilde{X}_0,\tilde{U}_0)$ is non-singular, under mild moment conditions,  $\Exp\left[|\tilde{\alpha}_n - \alpha_0| +\| \tilde{\theta}_n - \theta_0\|_2\right] \leq C n^{-1/2}$ for some constant $C$. Since $\bx $ and $\bu$ are compact, it implies that 
$\Exp\left[\|\tilde{r}_n - r\|_{\infty} \right] \leq C n^{-1/2}$, where $\tilde{r}_n(x,u) :=   \tilde{\alpha}_n x + \tilde{\theta}_n'u$ for $x \in \bx, u \in \bu$. Thus,  under the conditions in Theorem \ref{cor:model_noise_both_learned}, the upper bound therein decays at the parametric rate $n^{-1/2}$.
\end{example}

\begin{example}
  Assume that $\Exp[\tilde{W}_0] = 0$ and that $\bx = [0,1]$ and $\bu$ is finite. Let $m \geq 1$ be an integer. Suppose for each $u \in \bu$ and $0 \leq i \leq m$, we observe $\tilde{Y}_{u,i} = r(i/m,u) + \tilde{W}_{u,i}$, where $\{\tilde{W}_{u,i}\}$ are independent copies of $\tilde{W}_0$. Thus, the dataset is $\{(\tilde{Y}_{u,i}, i/m, u): 0 \leq i \leq m, u \in \bu\}$, which has a sample size of $n+1 := (m+1)\times |\bu|$.  Assume that there exist some constants $s, L > 0$ such that for each $u \in \bu$, $r(\cdot,u)$ belongs to H\"older class with parameter
$s, L$ \revise{(see e.g. \cite[Definition 1.2]{Tsybakov2009})}, that is, $r(\cdot,u)$ is $\ell := \lceil s \rceil-1$ times differentiable and 
$|r^{(\ell)}(x,u)-r^{(\ell)}(y,u)| \leq L |x-y|^{s-\ell}$ for $x,y \in \bx$, where $\lceil s \rceil$  is the ceiling of $s$ and $r^{(\ell)}(x,u)$ is the $\ell$-th derivative of $r(\cdot,u)$ for a fixed $u$. Further, assume $\tilde{W}_0$ is sub-Gaussian, in the sense that for some constants $a,b > 0$, $\Exp[\exp(a (\tilde{W}_0-\Exp[\tilde{W}_0])^2)] \leq b$. 
For each $u \in \bu$, let $\tilde{r}_n(\cdot,u)$ be a local polynomial estimator of order $\ell$ for $r(\cdot,u)$ based on $\{(\tilde{Y}_{u,i}, i/m): 0 \leq i \leq m\}$,  as defined in Section 1.6 of \cite{Tsybakov2009}. Then as shown in Theorem 1.8 of \cite{Tsybakov2009}, with a proper choice of bandwidth and kernel, we have $\Exp\left[\sup_{x \in \bx} |\tilde{r}_n(x,u)- r(x,u)| \right] \leq (\log(m)/m)^{s/(1+2s)}$. Since $\bu$ is finite, under the conditions in Theorem \ref{cor:model_noise_both_learned}, the upper bound  therein decays at the  rate $(\log(n)/n)^{s/(1+2s)}$.
\end{example}

\section{Concluding Remarks}\label{ivpcontinuity}
In this paper, we consider the Wasserstein model approximation problem, where we upper bound the performance gap of applying an optimal policy from an approximate model to the true dynamics. This gap is bounded linearly by the sup-norm-induced metric between the approximate and true costs, as well as the uniform Wasserstein-1 distance between the approximate and true transition kernels. We study both discounted-cost and average-cost criteria. Based on these results, we develop offline model learning algorithms and obtain their sample complexity bounds. Additionally, we recover and generalize several existing results on continuous dependence, robustness and approximations. 

An extension of the present work would be robustness of finite-step value iteration, which generalizes \cite{Rudolf2018}, as clarified in the following observations. Consider applying Corollary \ref{drivingnoisecorollary2} to the setup where $c=\altc$ and both the true and approximate system are control-free, i.e., 
\begin{align*}
&c(x,u)=c(x),\;\ct(\cdot|x,u)=\ct(\cdot|x),\;\ctp(\cdot|x,u)=\ctp(\cdot|x).
\end{align*}
Furthermore, suppose both probability kernels admit invariant probability measures, and denote them by $\rho_{\ct}$ and $\rho_{\ctp}$, respectively. By the result in Corollary \ref{drivingnoisecorollary2}, we have the following bound
\begin{equation*}
\left|\int_{\bx}c(x)\rho_{\ct}(dx)-\int_{\bx}c(x)\rho_{\ctp}(dx)\right|\leq \frac{\clip}{1-\tlip}\sup_{x\in\bx}\left(\wc_{1}\left(\ct(\cdot|x),\ctp(\cdot|x)\right)\right).
\end{equation*}
We note that the choice of the cost function is arbitrary as long as it is Lipschitz. Restricting our attention to those with Lipschitz constant less than or equal to 1, and taking supremum over such cost functions on both sides, we have
\begin{equation*}
\wc_{1}(\rho_{\ct},\rho_{\ctp})\leq \frac{1}{1-\tlip}\sup_{x\in\bx}\left(\wc_{1}\left(\ct(\cdot|x),\ctp(\cdot|x)\right)\right).
\end{equation*}
This recovers Corollary 3.1 of \cite{Rudolf2018} in an unweighted form. It is therefore natural to consider that a performance bound for finite-step approximate value iteration would be analogous to Theorem 3.1 in \cite{Rudolf2018}.

\bibliography{ZSYModelApproximation}
\appendix

\section{Technical results for the average-cost setting}
\subsection{Proof of Lemma \ref{lemma:hproperties}}\label{sec:hpropertiesproof} In this subsection, we prove Lemma \ref{lemma:hproperties}. We start with the proof to claim \textbf{(a)}. For simplicity of notation, we denote $\Delta h=h_{c,\ct,\epsilon_{1}}^{*}-h_{c,\ct,\epsilon_{2}}^{*}$. By  measurable selection theorem (see e.g. \cite[Theorem 12.1]{Schal}), there exists measurable function $u_{1},u_{2}:\bx \to \bu$ such that for $i \in \{1,2\}$ and $x \in \bx$,
\begin{align*}
g_{c,\ct}^{*}+h_{c,\ct,\epsilon_{i}}^{*}(x)&=\inf_{u\in\bu} \left\{c(x,u)+\int_{\bx}h_{c,\ct,\epsilon_{i}}^{*}(y)\ct(dy|x,u)\right\}\\
&=c(x,u_{i}(x))+\int_{\bx}h_{c,\ct,\epsilon_{i}}^{*}(y)\ct(dy|x,u_{i}(x)).
\end{align*}
By definition, it follows that for each $x \in \bx$,
$$
g_{c,\ct}^{*}+h_{c,\ct,\epsilon_{2}}^{*}(x)\leq c(x,u_{1}(x))+\int_{\bx}h_{c,\ct,\epsilon_{2}}^{*}(y)\ct(dy|x,u_{1}(x)),
$$
which implies that for each $x \in \bx$,
\begin{align*}
\Delta h(x)&\geq \int_{\bx}\Delta h(y) \ct(dy|x,u_{1}(x))\\
&=\int_{\bx}\Delta h(y) \left(\ct(dy|x,u_{1}(x))-\epsilon \rho(dy)\right)+\epsilon\int_{\bx}\Delta h(y) \rho(dy)\\
&\geq (1-\epsilon)\inf_{\bx}\Delta h+\epsilon\int_{\bx}\Delta h(y) \rho(dy).
\end{align*}
As a result, taking the infimum over $x \in  \bx$ on both side, we have
\begin{align*}
    \inf_{\bx}\Delta h&\geq \int_{\bx}\Delta h(y) \rho(dy).
\end{align*}
By a similar argument, we have
$$
\sup_{\bx}\Delta h\leq \int_{\bx}\Delta h(y) \rho(dy).
$$
which implies that $\Delta h$ is a constant function. 

Part (b) follows immediately from part (a), since adding a constant offset
does not change the Lipschitz constant.

Finally, we prove part \textbf{(c)}. Recall the definition of $\mathbb{T}_{c,\ct,\epsilon}$ in \eqref{equ:minor_operator}.
By Banach fixed-point theorem and the triangle inequality, we have that  
\begin{align*}
\left\Vert h_{c,\ct,\epsilon}^{*} \right\Vert_{\infty}&=\left\Vert \mathbb{T}_{c,\ct,\epsilon}h_{c,\ct,\epsilon}^{*} -0\right\Vert_{\infty}\leq \left\Vert \mathbb{T}_{c,\ct,\epsilon}h_{c,\ct,\epsilon}^{*} -\mathbb{T}_{c,\ct,\epsilon}0+\mathbb{T}_{c,\ct,\epsilon}0-0\right\Vert_{\infty}\\
&\leq \left\Vert \mathbb{T}_{c,\ct,\epsilon}h_{c,\ct,\epsilon}^{*} -\mathbb{T}_{c,\ct,\epsilon}0\right\Vert_{\infty}+\left\Vert\mathbb{T}_{c,\ct,\epsilon}0-0\right\Vert_{\infty}\\
&\leq (1-\epsilon)\left\Vert h_{c,\ct,\epsilon}^{*} \right\Vert_{\infty}+\left\Vert \inf_{u\in\bu}\left\{c(\cdot,u)+\int_{\bx}0\left(\ct(dy|\cdot,u)-\epsilon\rho(dy)\right)\right\}\right\Vert_{\infty}\\
&\leq (1-\epsilon)\left\Vert h_{c,\ct,\epsilon}^{*} \right\Vert_{\infty}+\Vert c\Vert_{\infty},
\end{align*}
which implies
\begin{equation*}
\left\Vert h_{c,\ct,\epsilon}^{*} \right\Vert_{\infty}\leq \Vert c\Vert_{\infty}/\epsilon.
\end{equation*} 
The proof is now complete.

\section{Statistical analysis} 

\subsection{A concentration result}\label{app:statistical_analysis}
Let $\mu$ be a probability measure on $\bW$. 
Let $W,W_0,W_1,\ldots$ be i.i.d.~random variables following the distribution $\mu$. For each $\epsilon > 0$, denote by $N(\epsilon, \bx \times \bu, d_{\bx} \times d_{\bu})$ the $\epsilon$-covering number of $\bx \times \bu$, that is, it is the smallest $M$ such that there exists $\{(x_i,u_i), 1 \leq i \leq M \} \subset \bx \times \bu$ with the property that 
$$
\bx \times \bu = \bigcup_{i=1}^{M} \left\{(x,u) \in \bx \times \bu: d_{\bx}(x,x_i) + d_{\bu}(u,u_i) < \epsilon \right\}.
$$
In addition, let $\mathcal{G}$ be a collection of real-valued functions defined on $\bW$. For each $\epsilon > 0$, denote by
$N_{[]}(\epsilon, \mathcal{G}, \|\cdot\|_{\mathcal{L}_{2}(\mu)})$ the $\epsilon$-bracketing number of $\mathcal{G}$ under $\mathcal{L}_{2}(\mu)$ norm, that is, it is the smallest $M >0$ such that there exist functions $\{\ell_i,u_i: \bW \to \bR: i \in [M]\}$ with the property that (i) for each $i \in [M]$, $\ell_i(w)\leq u_i(w)$  for any $w \in \bW$, and $\|u_i-\ell_i\|_{\mathcal{L}_{2}(\mu)} < \epsilon$; (ii) for any $g \in \mathcal{G}$,
we can find some $i^* \in [M]$ such that $g(w) \in [\ell_i(w),u_i(w)]$ for all $w \in \bW$.

Furthermore, for a measurable function $h: \bW \to \bR$, denote by $\|h\|_{\mathcal{L}_{2}(\mu)} := \left(\int h^2(w) \mu(dw) \right)^{1/2}$ the $\mathcal{L}_{2}(\mu)$ norm with respect to $\mu$.

\begin{lemma} \label{app:empirical_proc}
Let $B, M >0$, and consider a bounded, measurable function $g: \bx\times\bu \times \bW \to [-B,B]$. 
\begin{enumerate}[label=\roman*)]
        \item Assume that for some constant $K,L \geq 1$, we have that for each $\epsilon > 0$, 
        $$
        N(\epsilon, \bx \times \bu, d_{\bx} \times d_{\bu}) \leq \left(\frac{K}{\epsilon} \right)^{L}.
        $$
        \item For $\mu$-almost every  $w$, we have
        $$
        |g(x,u,w) - g(x',u',w)| \leq M  \left(d_{\bx}(x,x') + d_{\bu}(u,u')\right), \text{ for any } (x,u),(x',u') \in \bx \times \bu.
        $$
    \end{enumerate}
Then, there exists constant $C > 0$, depending on $K$ and $L$,   such that
    \begin{align*}
        \Exp\left[\sup_{(x,u) \in \bx \times \bu} \left|\frac{1}{n+1}\sum_{i=0}^{n} g(x,u, W_i) - E[g(x,u,W)] \right| \right] \leq C (B+M) n^{-1/2}.
    \end{align*}
\end{lemma}

\begin{proof} For each $(x,u) \in \bx \times \bu$, define $g_{x,u}(w) = g(x,u,w)$ for $w \in \bW$. Let $\mathcal{G} := \{g_{x,u}: (x,u) \in \bx \times \bu\}$. Then by \cite[Theorem 2.7.11]{vanderVaart1996}, for each $\epsilon > 0$,
\begin{align}\label{app:rate_conv}
    N_{[]}(2\epsilon M, \mathcal{G}, \|\cdot\|_{\mathcal{L}_{2}(\mu)}) \leq N(\epsilon, \bx \times \bu, d_{\bx} \times d_{\bu}) \leq \left(\frac{K}{\epsilon} \right)^{L}.
\end{align}

By Theorem 2.14.2 in \cite{vanderVaart1996}, for some absolute constant $C > 0$, we have
\begin{align*}
            \Delta := &\Exp\left[\sup_{(x,u) \in \bx \times \bu} \left|\frac{1}{n+1}\sum_{i=0}^{n} g(x,u, W_i) - E[g(x,u,W)] \right| \right] \\
            \leq &C n^{-1/2} \int_0^1 \sqrt{1 + \log\left(N_{[]}(\epsilon B, \mathcal{G}, \|\cdot\|_{\mathcal{L}_{2}(\mu)}) \right)} \, d\epsilon B,
\end{align*}
where we note that the constant function $B$ is an envelope function for $\mathcal{G}$. By a change-of-variable, 
\begin{align*}
    \Delta &\leq C n^{-1/2} \frac{2M}{B}\int_{0}^{{B}/{(2M})} \sqrt{1 + \log\left(N_{[]}(2\epsilon M, \mathcal{G}, \|\cdot\|_{\mathcal{L}_{2}(\mu)}) \right)} \, d\epsilon B \\
    & \leq 2C n^{-1/2} {M}(B/(2M) + 1) \int_{0}^{1} \sqrt{1 + L \log(K/\epsilon)} \, d\epsilon,
\end{align*}
where we apply the upper bound \eqref{app:rate_conv} on the bracketing number in the last step. Since $\int_0^{1} \sqrt{\log(1/\epsilon)} d\epsilon < \infty$, there exists an absolute constant $C > 0$ such that $\Delta \leq C n^{-1/2} (M + B)$.
\end{proof}

\begin{remark}\label{remark:empir_proc}
The condition i) above holds if $\bx$ and $\bu$ are compact Euclidean subsets, and the constants $K$ and $L$ depend on the dimensions and diameters of $\bx$ and $\bu$.
\end{remark}

\subsection{Proof of Lemma \ref{lemma:Markov_con}} \label{subapp:proof_Markov_con}
In this subsection, we prove Lemma \ref{lemma:Markov_con}. By definition, we have that $\|\hat{c}_{N}\|_{\infty} \leq \|c\|_{\infty}$ and $\|c^{\pi,M}\|_{\infty} \leq \|c\|_{\infty}$.
For each $i \in [M]$, $u \in \bu$, define
\begin{align*}
    & A_{i,u} :=   \frac{1}{N}\sum_{n=0}^{N-1} \mathbbm{1}\{X_n \in B_i, U_n = u\}, \quad a_{i,u} := \revise{\phi(u)} \int_{B_i}   \pi(dx), \\
    &  \Xi_{i,u}  :=  \frac{1}{N}\sum_{n=0}^{N-1} C_n \mathbbm{1}\{X_n \in B_i, U_n = u\},\quad \xi_{i,u} := \revise{\phi(u)} \int_{B_i} c(x,u) \pi(dx) 
\end{align*}
By definition of $\hat{c}_N$ following equation \eqref{def:collection_MC_dyn} and $c^{\pi,M}$ in \eqref{finitemodeldef}, and the triangle inequality,
\begin{align*}
 \|\hat{c}_N -c^{\pi,M}\|_{\infty} \leq \max_{i \in [M],u \in \bu} \frac{1}{A_{i,u}}\left|\Xi_{i,u} - \xi_{i,u} \right| + \frac{|\xi_{i,u}|}{A_{i,u} a_{i,u}} \left|A_{i,u} - a_{i,u} \right|.
\end{align*}
By definition, $\frac{|\xi_{i,u}|}{a_{i,u}} \leq \|c\|_{\infty}$. Recall the definition of $\kappa_{\pi,M}$ prior to Lemma \ref{lemma:Markov_con} and define the event 
\begin{align}\label{def:Sigma1}
    \Sigma_1 := \bigcap_{i \in [M],u \in \bu} \left\{A_{i,u} \geq \kappa_{\pi,M}/2 \right\}.
\end{align}
By conditioning on the event $\Sigma_1$ and its complement $\Sigma_1^{c}$, we have
\begin{align*}
    \Exp\left[ \|\hat{c}_N -c^{\pi,M}\|_{\infty} \right] \leq 
    \frac{2}{\kappa_{\pi,M}} \Exp\left[ \left\|\Xi - \xi \right\|_{\infty} \right] +
    \frac{2\|c\|_{\infty}}{\kappa_{\pi,M}} \Exp\left[ \left\|A - a \right\|_{\infty} \right] +  \|c\|_{\infty} \Pro\left( \Sigma_1^c \right).
\end{align*}

Applying Assumption \ref{Model_Learning_Markov}(b)  with the function $f_{i,u}(x,u',x') := \mathbbm{1}\{x \in B_i, u' = u\}$ and ${f}_{i,u}(x,u',x'):= c(x,u') \mathbbm{1}\{x \in B_i, u' = u\}$, respectively, 
to the Markov chain $\{(X_n,U_n,X_{n+1}): n \geq 0\}$, we obtain  for each $i \in [M]$, $u \in \bu$, and $\epsilon > 0$,
\begin{align*}
\max\{ \Pro\left( \left|A_{i,u} - a_{i,u} \right| \geq \epsilon \right), \;\Pro\left( \left|\Xi_{i,u} - \xi_{i,u} \right| \geq \|c\|_{\infty} \epsilon \right) 
\}
\leq 2 C_0 \exp\left( - c_0  N\epsilon^2  \right).
\end{align*}
By the union bound with $\epsilon = \kappa_{\pi,M}/2$, we have
\begin{align*}
   \Pro\left( \Sigma_1^c \right) \leq 2C_0 M |\bu| \exp\left( -c_0 \kappa_{\pi,M}^2 N/4 \right).
\end{align*}
Furthermore,  by Lemma 2.2.1 in \cite{vanderVaart1996}, we have $\left\|
   \sqrt{N}(A_{i,u} - a_{i,u}) 
    \right\|_{\psi_2}^2 \leq (1+2C_0)/c_0$ for $i \in [M], u \in \bu$, where for a random variable $Z$,  its $\psi_{2}$ norm is defined as follows: $ \|Z\|_{\psi_{2}} := \inf\{C > 0: \Exp\left[\exp(|Z/C|^2)\right] \leq 2\}$.  Then by Lemma 2.2  in \cite{vanderVaart1996}, for some absolute constant $C>0$,
    \begin{align*}
        \Exp\left[\|A - a\|_{\infty} \right]
=         \Exp\left[\max_{i \in [M], u \in \bu} \|A - a\|_{\infty} \right]
\leq C \sqrt{\frac{\log(M |\bu|)}{N}} \sqrt{ \frac{1+2C_0}{c_0}}.
    \end{align*}
By a similar argument, we have 
\begin{align*}
    \Exp\left[ \left\|\Xi - \xi \right\|_{\infty} \right] \leq C \ \|c\|_{\infty} \sqrt{\frac{\log(M |\bu|)}{N}} \sqrt{ \frac{1+2C_0}{c_0}}.
\end{align*}
The proof for the first claim then is complete by combining the above inequalities. 

Now, we focus on the second claim. For each $i \in [M]$, $u \in \bu$, define
\begin{align*}
&\widehat{\Xi}_{i,u} :=    \frac{1}{N}\sum_{n=0}^{N-1} \left(\sum_{j \in [M]}g(y_j) \mathbbm{1}\{X_{n+1} \in B_j, X_n \in B_i, U_n=u\}\right), \\
&\widehat{\xi}_{i,u} := \revise{\phi(u)}\times \sum_{j \in [M]}g(y_j)\int_{B_i} \ct(B_j|x,u) \pi(dx).
\end{align*}
By definition, $d_{g}(\hat{\ct}_N, \ct^{\pi,M}) = \sup_{i \in [M],u \in \bu} \left| 
  {\widehat{\Xi}_{i,u}}/{ A_{i,u}}
  - {\widehat{\xi}_{i,u}}/{a_{i,u}
    }
    \right|$. Then, since $\|g\|_{\infty} \leq L$, by the same argument as above, 
\begin{align*}
    \Exp\left[ d_{g}(\hat{\ct}_N, \ct^{\pi,M}) \right]
    \leq     
    \frac{2}{\kappa_{\pi,M}} \Exp\left[ \left\|\widehat{\Xi} - \widehat{\xi} \right\|_{\infty} \right] +
    \frac{2L}{\kappa_{\pi,M}} \Exp\left[ \left\|A - a \right\|_{\infty} \right] +  2L \Pro\left( \Sigma_1^c \right).
\end{align*}
Further, again applying Assumption \ref{Model_Learning_Markov}(b)   with the function $f_{i,u}(x,u',x') := L+\sum_{j \in [M]} g(y_j) \mathbbm{1}\{x' \in B_j, x \in B_i, u' = u \}$, we have
\begin{align*}
    \Pro\left( \left|\widehat{\Xi}_{i,u} - \widehat{\xi}_{i,u} \right| \geq  2L \epsilon \right) 
\leq 2 C_0 \exp\left( - c_0  N\epsilon^2  \right),
\end{align*}
which implies that   $\Exp\left[ \left\|\widehat{\Xi} - \widehat{\xi} \right\|_{\infty} \right] \leq C L \sqrt{\frac{\log(M |\bu|)}{N}} \sqrt{ \frac{1+2C_0}{c_0}}$. The proof is then complete.

Finally, we focus on the last claim.  For $i,j \in [M]$ and $u \in \bu$, define
\begin{align*}
    &\widetilde{\Xi}_{i,j,u} :=    \frac{1}{N}\sum_{n=0}^{N-1}  \mathbbm{1}\{X_{n+1} \in B_j, X_n \in B_i, U_n=u\}, \quad \widetilde{\xi}_{i,j,u} := \revise{\phi(u)} \times \int_{B_i} \ct(B_j|x,u) \pi(dx).
\end{align*}
 Applying Assumption \ref{Model_Learning_Markov}(b)  with the function $f_{i,j,u}(x,u',x') := \mathbbm{1}\{x \in B_i, u' = u, x' \in B_j\}$
to the Markov chain $\{(X_n,U_n,X_{n+1}): n \geq 0\}$, we obtain   for each $i,j \in [M]$, $u \in \bu$, and $\epsilon > 0$,
\begin{align*}
\Pro\left( \left|\widetilde{\Xi}_{i,j,u} - \widetilde{\xi}_{i,j,u} \right| \geq \epsilon \right) 
\leq 2 C_0 \exp\left( - c_0  N\epsilon^2  \right).
\end{align*}
By the definition of the event $\Sigma_1$ in \eqref{def:Sigma1}, on the event $\Sigma_1$, for $i,j \in [M]$ and $u \in \bu$,
$$
\left|\hat{\ct}_N(y_j|y_i,u) - \ct^{\pi,M} (y_j|y_i,u)\right| \leq \frac{2}{\kappa_{\pi,M}}\left|\widetilde{\Xi}_{i,j,u}-\widetilde{\xi}_{i,j,u}\right| + \frac{2}{\kappa_{\pi,M}}|A_{i,u} - a_{i,u}|.
$$

By the definition of  $\mathcal{E}_{N, M}$ in \eqref{def:E_N_M} and $\kappa_{\ct,M}$ in \eqref{def:kappa_pi_M_init}, we have
\begin{align*}
    \mathcal{E}_{N, M} \cap \Sigma_1 &\;\subset\; \bigcup_{i,j \in [M], u \in \bu  } \left\{  \left|\hat{\ct}_N(y_j|y_i,u) - \ct^{\pi,M} (y_j|y_i,u)\right| > \kappa_{\ct,M}/2 \right\} \cap \Sigma_1 \\
    &
\;\subset\; \bigcup_{i,j \in [M], u \in \bu  } \left(\left\{  \left|\widetilde{\Xi}_{i,j,u}-\widetilde{\xi}_{i,j,u}\right| > \kappa_{\pi,M}\kappa_{\ct,M}/8 \right\} \bigcup 
\left\{
|A_{i,u} - a_{i,u}| > \kappa_{\pi,M}\kappa_{\ct,M}/8
\right\}
\right).
\end{align*}
As a result, by the union bound, we have
\begin{align*}
  \Pro(\mathcal{E}_{N, M}) &\leq \Pro(\mathcal{E}_{N, M} \cap \Sigma_1) + \Pro(\Sigma_1^c) 
  \leq \sum_{i,j \in [M], u \in U} \Pro\left(\left|\widetilde{\Xi}_{i,j,u}-\widetilde{\xi}_{i,j,u}\right| > \kappa_{\pi,M}\kappa_{\ct,M}/8 \right) \\
  &+ \sum_{i \in [M], u \in \bu} \Pro\left( |A_{i,u} - a_{i,u}| > \kappa_{\pi,M}\kappa_{\ct,M}/8\right)+\Pro(\Sigma_1^c).
\end{align*}
Then the proof is complete by combining the previous upper bounds on these probabilities.

\end{document}